\documentclass[pra,longbibliography,showpacs,nofootinbib,superscriptaddress,notitlepage,accepted=2018-01-31]{quantumarticle}
\usepackage{amsmath}
\usepackage{amssymb,bm}
\usepackage{amsthm}

\newtheorem{innercustomgeneric}{\customgenericname}
\providecommand{\customgenericname}{}
\newcommand{\newcustomtheorem}[2]{%
  \newenvironment{#1}[1]
  {%
   \renewcommand\customgenericname{#2}%
   \renewcommand\theinnercustomgeneric{##1}%
   \innercustomgeneric
  }
  {\endinnercustomgeneric}
}

\newcustomtheorem{customthm}{Theorem}
\newcustomtheorem{definitionBob}{Definition}

\usepackage{color,dsfont} 
\usepackage{graphicx, caption} 
\usepackage{ragged2e}
\usepackage[colorlinks=true, hyperindex, breaklinks, linkcolor=blue, urlcolor=blue, citecolor=blue]{hyperref} 
\usepackage[normalem]{ulem}
\usepackage[capitalise]{cleveref}
\usepackage{mathrsfs}
\usepackage{subcaption}
\usepackage{mathtools}
\usepackage{float}

\newtheorem{claim}{Claim} \newtheorem{definition}{Definition}

\newcommand{\ket}[1]{|#1\rangle}  
 
\newcommand{\codepar}[1]{\ensuremath{[\![#1]\!]}}

\begin{document}

\title{Flag fault-tolerant error correction with arbitrary distance codes}

\author{Christopher Chamberland}
\email{c6chambe@uwaterloo.ca}
\affiliation{
    Institute for Quantum Computing and Department of Physics and Astronomy,
    University of Waterloo,
    Waterloo, Ontario, N2L 3G1, Canada
    }
\author{Michael E. Beverland}
\email{mibeverl@microsoft.com}
\affiliation{
   Station Q Quantum Architectures and Computation Group, Microsoft Research
Redmond, WA 98052, USA
    }

\begin{abstract}
In this paper we introduce a general fault-tolerant quantum error correction protocol using flag circuits for measuring stabilizers of arbitrary distance codes. In addition to extending flag error correction beyond distance-three codes for the first time, our protocol also applies to a broader class of distance-three codes than was previously known.
Flag circuits use extra ancilla qubits to signal when errors resulting from $v$ faults in the circuit have weight greater than $v$. The flag error correction protocol is applicable to stabilizer codes of arbitrary distance which satisfy a set of conditions and uses fewer qubits than other schemes such as Shor, Steane and Knill error correction. 
We give examples of infinite code families which satisfy these conditions and analyze the behaviour of distance-three and -five examples numerically. Requiring fewer resources than Shor error correction, flag error correction could potentially be used in low-overhead fault-tolerant error correction protocols using low density parity check quantum codes of large code length. 
\end{abstract}


\maketitle

\section{Introduction and formalism}
\label{sec:Intro}

Scalable quantum computers are expected to require some form of error correction (EC) to function reliably. Unfortunately, no practical model for a self-correcting quantum memory has been proposed to date, despite considerable effort \cite{BrownMemory16}.
The models that come closest to this goal involve topological protection in the presence of physically imposed symmetries \cite{KitaevWire01,KarzigScalable17}, but even these are not expected to reduce error rates sufficiently for large computations.
Therefore active protocols that require measuring the check operators of an error correcting code are probably necessary to realize scalable quantum computing.

There are three general approaches of fault-tolerant error correction (FTEC) applicable to a wide range of stabilizer codes due to Shor \cite{Shor96}, Steane \cite{Steane97}, and Knill \cite{KL2005}.
There are also a number of promising code-specific FTEC schemes, most notably the surface code with a minimum weight matching error correction scheme \cite{BK98,DKLP02,FMMC12}. 
This approach gives the best fault-tolerant thresholds to date and only requires geometrically local measurements.
A high threshold \cite{Shor96,AB97,Preskill98,KLZ98} implies that relatively imperfect hardware could be used to reliably implement long quantum computations. 
Despite this, the hardware and overhead requirements for the surface code are sufficiently demanding that it remains extremely challenging to implement in the lab. 

Fortunately, there are reasons to believe that there could be better alternatives to the surface code.
For example, dramatically improved thresholds could be possible using concatenated codes if they enjoyed the same level of optimization as the surface code has in recent years \cite{Poulin06,PhysRevA.83.020302}. 
Another enticing alternative is to find and use efficiently-decodable low density parity check (LDPC) codes with high rate \cite{Gallager1960,LDPC13,TZLDPC14} in a low-overhead FTEC protocol \cite{Gottesman13LDPC}.
For these and other reasons, it is important to have general FTEC schemes applicable to a wide range of codes and to develop new schemes.

Shor EC can be applied to any stabilizer code, but typically requires more syndrome measurement repetitions than Steane and Knill EC. Furthermore, all weight-$w$ stabilizer generators are measured sequentially using $w$-qubit verified cat states. On the other hand, Steane EC has higher thresholds than Shor EC and has the advantage that all Clifford gates are applied transversally during the protocol. However, Steane EC is only applicable to CSS \cite{CS96,Steane97} codes and uses a verified logical $\ket{+}$ state encoded in the same code to simultaneously obtain all $X$-type syndromes, using transversal measurement (similarly for $Z$). 
Knill EC can also be applied to any stabilizer code but requires two additional ancilla code blocks (encoded in the same code that protects the data) prepared in a logical Bell state. The Bell state teleports the encoded information to one of the ancilla code blocks and the extra information from the transversal Bell measurement gives the error syndrome. Knill EC typically achieves higher thresholds than Shor and Steane EC but often uses more qubits \cite{Knill05,Fern08KnillUpperBound}. It is noteworthy that for large LDPC codes, in which low weight generators are required be fault-tolerantly measured, Shor EC is much more favourable than Steane or Knill EC. Many improvements in these schemes have been made. For examples, in \cite{DA07}, ancilla decoding was introduced to correct errors arising during state preparation in Shor and Steane EC rather than simply rejecting all states which fail the verification procedure.

In this work, we build on a number of recent papers \cite{CR17v1,CR17v2,Yoder2017surfacecodetwist} that demonstrate flag error correction for particular distance-three and error detecting codes and present a general protocol for arbitrary distance codes. Flag error correction uses extra ancilla qubits to detect potentially problematic high weight errors that arise during the measurement of a stabilizer. We provide a set of requirements for a stabilizer code (along with the circuits used to measure the stabilizers) which, if satisfied, can be used for flag error correction. We are primarily concerned with extending the lifetime of encoded information using fault-tolerant error correction and defer the study of implementing gates fault-tolerantly to future work. Our approach can be applied to a broad class of codes (including but not limited to surface codes, color codes and quantum Reed-Muller codes).
Of the three general schemes described above, flag EC has most in common with Shor EC. 
Further, flag EC does not require verified state preparation, and for all codes considered to date, requires fewer ancilla qubits. Lastly, we note that in order to satisfy the fault-tolerant error correction definition presented in \cref{subsec:Section0}, our protocol applied to distance-three codes differs from \cite{CR17v1}.

We foresee a number of potential applications of these results. 
Firstly we believe it is advantageous to have new EC schemes with different properties that can be used in various settings.
Secondly, flag EC involves small qubit overhead, hence possibly the schemes presented here and in other flag approaches \cite{CR17v1,CR17v2,Yoder2017surfacecodetwist} will find applications in early qubit-limited experiments.
Thirdly, we expect the flag EC protocol presented here could potentially be useful for LDPC codes as described in \cite{Gottesman13LDPC}.

In \cref{subsec:ReviewChaoReichardt,subsec:Distance5protocol} we provide important definitions and introduce flag FTEC for distance-three and -five codes. In \cref{subsec:ApplicationProtocolColorCode} we apply the protocol to two examples: the \codepar{19,1,5} and \codepar{17,1,5} color codes, which importantly have a variety of different weight stabilizers. The general flag FTEC protocol for arbitrary distance codes is given in \cref{subsec:GeneralProtocol}. A proof that the general protocol satisfies the fault-tolerance criteria is given in \cref{app:ProtocolGeneralProof}. In \cref{subsec:Remarks} we provide examples of codes that satisfy the conditions that we required for flag FTEC.  Flag circuit constructions for measuring stabilizers of the codes in \cref{subsec:Remarks} are given  \cref{app:GeneralTflaggedCircuitConstruction}. We also provide a candidate circuit construction for measuring arbitrary weight stabilizers in \cref{App:GeneralwFlagCircuitConstruction}. In \cref{sec:CircuitLevelNoiseFTEC}, we analyze numerically a number of flag EC schemes and compare with other FTEC schemes under various types of circuit level noise. We find that flag EC schemes, which have large numbers of idle qubit locations, behave best in error models in which idle qubit errors occur with a lower probability than CNOT errors. The remainder of this section is devoted to FTEC and noise model/simulation methods. 

\subsection{Fault-tolerant error correction}
\label{subsec:Section0}

Throughout this paper, we assume a simple depolarizing noise model in which idle qubits fail with probability $\tilde{p}$ and all other circuit operations (gates, preparations and measurements) fail with probability $p$, which recovers standard circuit noise when $\tilde{p}=p$. A detailed description is given in \cref{subsec:NoiseAndNumerics}.

The weight of a Pauli operator $E$ ($\text{wt}(E)$) is the number of qubits on which it has non-trivial support. We first make some definitions,

\begin{definition}{\underline{Weight-$t$ Pauli operators}}
	\begin{align}
	\mathcal{E}_{t} = \{ E \in \mathcal{P}_{n} | \text{wt}(E) \le t \},
	\end{align}
	where $\mathcal{P}_{n}$ is the $n$-qubit Pauli group.
	\label{Def:EpsilontSet}
\end{definition}

\begin{definition}{\underline{Stabilizer error correction}}
	
	Given a stabilizer group $\mathcal{S} = \langle g_{1}, \cdots, g_{m} \rangle$, we define the syndrome $s(E)$ to be a bit string, with i'th bit equal to zero if $g_i$ and $E$ commute, and one otherwise.
	Let $E_{\text{min}}(s)$ be a minimal weight correction $E$ where $s(E)=s$. 
	We say operators $E$ and $E'$ are logically equivalent, written as $E \sim E'$, iff $E' \propto g E$ for $g \in \mathcal{S}$.
	\label{Def:LogEquivDef}
\end{definition}

An error correction protocol typically consists of a sequence of basic operations to infer syndrome measurements of a stabilizer code $C$, followed by the application of a Pauli operator (either directly or through Pauli frame tracking \cite{DA07,Barbara15,CIP17}) intended to correct errors in the system.
Roughly speaking, a given protocol is fault-tolerant if for sufficiently weak noise, the effective noise on the logical qubits is even weaker.
More precisely, we say that an error correction protocol is a $t$-FTEC if the following is satisfied:

\begin{definition}{\underline{Fault-tolerant error correction}}
	
	For $t = \lfloor (d-1)/2\rfloor$, an error correction protocol using a distance-$d$ stabilizer code $C$ is $t$-fault-tolerant if the following two conditions are satisfied:
	\begin{enumerate}
		\item For an input codeword with error of weight $s_{1}$, if $s_{2}$ faults occur during the protocol with $s_{1} + s_{2} \le t$, ideally decoding the output state gives the same codeword as ideally decoding the input state.
		\item For $s$ faults during the protocol with $s \le t$, no matter how many errors are present in the input state, the output state differs from a codeword by an error of at most weight $s$.
	\end{enumerate}
	\label{Def:FaultTolerantDef}
\end{definition}

Here ideally decoding is equivalent to performing fault-free error correction. 
By codeword, we mean any state $\ket{\overline{\psi}} \in C$ such that $g\ket{\overline{\psi}} = \ket{\overline{\psi}} \thinspace  \forall \thinspace g \in \mathcal{S}$ where $\mathcal{S}$ is the stabilizer group for the code $C$.
Note that for the second criteria in \cref{Def:FaultTolerantDef}, the output and input codeword can differ by a logical operator. 

The first criteria in \cref{Def:FaultTolerantDef} ensures that correctable errors don't spread to uncorrectable errors during the error correction protocol. Note however that the first condition alone isn't sufficient. For instance, the trivial protocol where no correction is ever applied at the end of the EC round also satisfies the first condition, but clearly is not fault-tolerant.

The second condition is not always checked for protocols in the literature, but it is important as it ensures that errors do not accumulate uncontrollably in consecutive rounds of error correction (see \cite{AGP06} for a rigorous proof and \cite{CDT09} for an analysis of the role of input errors in an extended rectangle). To give further motivation as to why the second condition is important, consider a scenario with $s$ faults introduced during each round of error correction, and assume that $t/n<s<(2t+1)/3$ for some integer $n$ (see Fig.~\ref{fig:ConditionTwoJustification}). Consider an error correction protocol in which $r$ input errors and $s$ faults in an EC block leads to an output state with at most $r+s$ errors\footnote{This is the case for Shor, Steane and Knill EC with appropriately verified ancilla states. However the surface code does not satisfy this due to hook errors but nonetheless still satisfies condition 1 of \cref{Def:FaultTolerantDef}. }. Clearly condition 1 is satisfied.

With the above considerations, an input state $E_{1}\ket{\bar{\psi}}$ with $\text{wt}(E_{1})\leq s$ is taken to $E_{2}\ket{\bar{\psi}}$, with $\text{wt}(E_{2})\leq 2s$ by one error correction round with $s$ faults.
After the $j$th round, the state will be $E_{j}\ket{\bar{\psi}}$ with the first condition implying $\text{wt}(E_{j})\leq j \cdot s$ provided that $j \leq n$. 
However, when $j > n$, the requirement of the first condition is no longer satisfied so we cannot use it to upper bound $\text{wt}(E_{j})$. 
Now consider the same scenario but assuming both conditions hold.
The second condition implies that after the first round, the input state $E_{1}\ket{\bar{\psi}}$ becomes $E'_{2}\ket{\bar{\phi}} = E_{2}\ket{\bar{\psi}}$, with  $\text{wt}(E_{2}')\leq s$, and where $\ket{\bar{\phi}}$ is a codeword. 
Therefore the codewords are related by: $\ket{\bar{\phi}}= (E_{2}^{'\dagger} E_{2}) \ket{\bar{\psi}}$, with logical operator $(E_{2}^{'\dagger} E_{2})$ having weight at most $3s$, since $\text{wt}(E_{2})+\text{wt}(E_{2}') \leq 3s$. 
However, the minimum non-trivial logical operator of the code has weight $(2t+1)>3s$, implying that $\ket{\bar{\psi}} = \ket{\bar{\phi}}$, and therefore that $\text{wt}(E_{2}) = \text{wt}(E_{2}') \leq s$. 
Hence, for the $j$th round, $\text{wt}(E_{j}) \leq s$ for all $j$, i.e. the distance from the codeword is not increased by consecutive error correction rounds with $s$ faults, provided $s < (2t+1)/3$.

\begin{figure}
	\centering
	\includegraphics[width=0.45\textwidth]{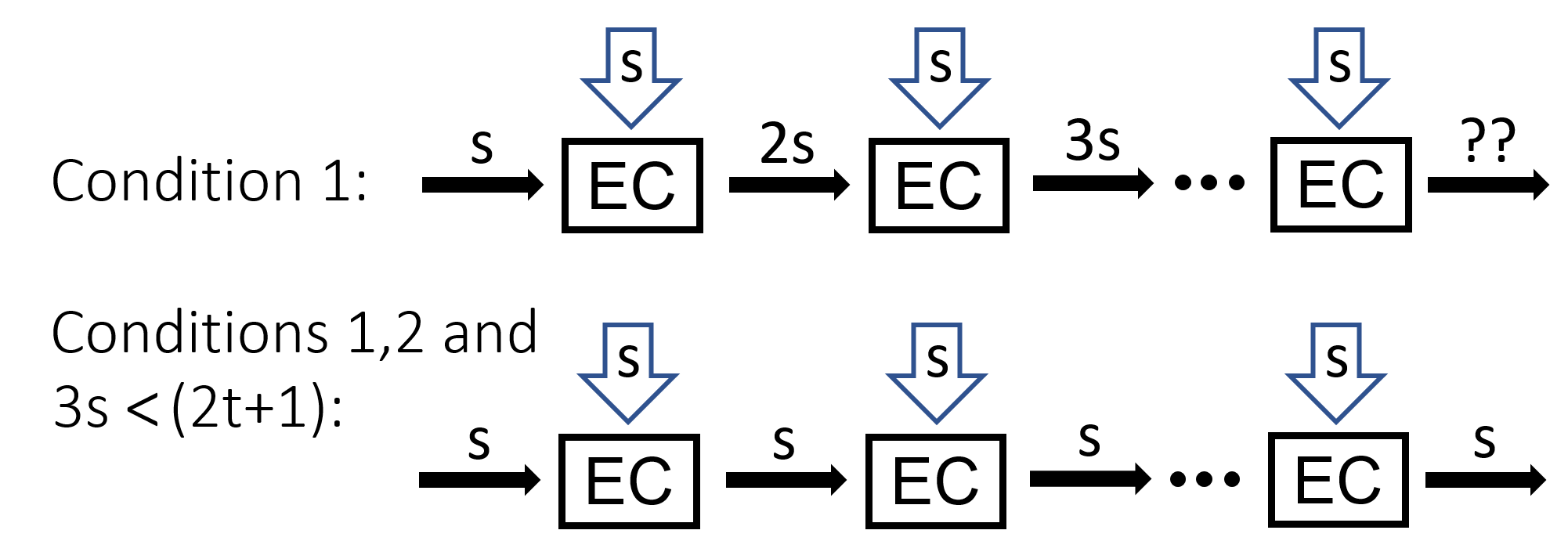}
	\caption{An example showing the first fault tolerance condition alone in \cref{Def:FaultTolerantDef} is not sufficient to imply a long lifetime. 
	We represent $s$ faults occurring during a round of error correction with a vertical arrow, and a state a distance $r$ from the desired codeword with a horizontal arrow with $r$ above. 
	The first condition alone allows errors to build up over time as in the top figure, which would quickly lead to a failure.
	However provided $s<(2t+1)/3$, both conditions together ensure that errors in consecutive error correction rounds do not build up, provided each error correction round introduces no more than $s$ faults, which could remain true for a long time.}
	\label{fig:ConditionTwoJustification}
\end{figure}

\subsection{Noise model and pseudo-threshold calculations}
\label{subsec:NoiseAndNumerics}

In \cref{sec:CircuitLevelNoiseFTEC}, we perform a full circuit level noise analysis of various error correction protocols. Unless otherwise stated, we use the following depolarizing noise model:

\begin{enumerate}
	\item With probability $p$, each two-qubit gate is followed by a two-qubit Pauli error drawn uniformly and independently from $\{I,X,Y,Z\}^{\otimes 2}\setminus \{I\otimes I\}$.
	\item With probability $\frac{2p}{3}$, the preparation of the $\ket{0}$ state is replaced by $\ket{1}=X\ket{0}$. Similarly, with probability $\frac{2p}{3}$, the preparation of the $\ket{+}$ state is replaced by $\ket{-}=Z\ket{+}$.
	\item With probability $\frac{2p}{3}$, any single qubit measurement has its outcome flipped. 
	\item Lastly, with probability $\tilde{p}$, each resting qubit location is followed by a Pauli error drawn uniformly and independently from $\{ X,Y,Z \}$.
\end{enumerate}

Some error correction schemes that we analyze contain a significant number of idle qubit locations. Consequently, most schemes will be analyzed using three ratios ($\tilde{p} = p$, $\tilde{p} = p/10$ and $\tilde{p} = p/100$) to highlight the impact of idle qubit locations on the logical failure rate.

The two-qubit gates we consider are: CNOT, XNOT$ = H_1(\text{CNOT})H_1$, and CZ$ = H_2(\text{CNOT})H_2$.

Logical failure rates are estimated using an $N$-run Monte Carlo simulation. During a particular run, errors are added at each location following the noise model described above. Once the error locations are fixed, the errors are propagated through a fault-tolerant error correction circuit and a recovery operation is applied. After performing a correction, the output is ideally decoded to verify if a logical fault occurred. 
For an error correction protocol implemented using a stabilizer code $C$ and a fixed value of $p$, we define the logical failure rate
\begin{equation}
p_{\text{L}}^{(C)}(p) = \lim_{N \to\infty} \frac{N_{\text{fail}}^{(C)}(p)}{N} ,
\end{equation}
where $N_{\text{fail}}^{(C)}(p)$ is the number of times a logical $X$ \textit{or} logical $Z$ error occurred during the $N$ rounds. 
In practice we take $N$ sufficiently large to estimate $p_{\text{L}}^{(C)}(p)$, and provide error bars \cite{AliferisCross07,CJL16b}.

In this paper we are concerned with evaluating the performance of FTEC protocols (i.e. we do not consider performing logical gates fault-tolerantly). We define the pseudo-threshold of an error correction protocol to be the value of $p$ such that 

\begin{align}
\tilde{p}(p) = p^{(C)}_{L}(p).
\label{Def:PseudoThreshDef}
\end{align}
Note that it is important to have $\tilde{p}$ on the left of \cref{Def:PseudoThreshDef} instead of $p$ since we want an encoded qubit to have a lower logical failure rate than an unencoded idle qubit. From the above noise model, a resting qubit will fail with probability $\tilde{p}$.

\section{Flag error correction for small distance codes}
\label{sec:Section1}

In this and the next section, we present a $t$-fault-tolerant flag error correction protocol with distance-$(2t+1)$ codes satisfying  a certain condition.
Our approach extends that introduced by Chao and Reichardt \cite{CR17v1} for distance three codes, which we first review using our terminology in \cref{subsec:ReviewChaoReichardt}. 
In \cref{subsec:Distance5protocol} we present the protocol for distance five CSS codes which contains most of the main ideas of the general case (which is provided in \cref{app:GeneralFTEC}).
Lastly, in section \cref{subsec:ApplicationProtocolColorCode} we provide examples of how the protocol is applied to the \codepar{19,1,5}  and \codepar{17,1,5} color codes. 

\subsection{Definitions and Flag $1$-FTEC with distance-3 codes}
\label{subsec:ReviewChaoReichardt}

In what follows, we use the term location to refer to a gate, state preparation, measurement or idle qubit where a fault may occur. 
Note also that a two-qubit Pauli error $P_{1}\otimes P_{2}$ arising at a two-qubit gate location counts as a single fault.

It is well known that with only a single measurement ancilla, a single fault in a blue CNOT of the stabilizer measurement circuit shown in \cref{fig:StabNonFT} can result in a multi-weight error on the data block.
This could cause a distance-$3$ code to fail, or more generally could cause a distance-$d$ code to fail due to fewer than $(d-1)/2$ total faults.
We therefore say the blue CNOTs are \textit{bad} according to the following definition:
\begin{definition}{\underline{Bad locations}}

A circuit location in which a single fault can result in a Pauli error $E$ on the data block with $\mathrm{wt}(E) \ge 2$ will be referred to as a bad location.
\label{Def:BadErrorDef}
\end{definition}

\begin{figure}
\centering
\begin{subfigure}{0.25\textwidth}
\includegraphics[width=\textwidth]{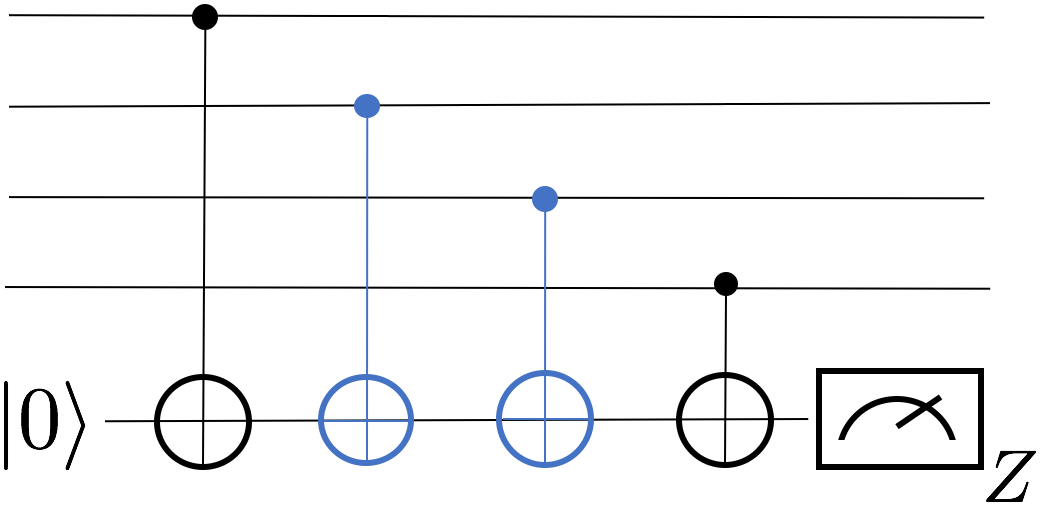}
\caption{}
\label{fig:StabNonFT}
\end{subfigure}
\begin{subfigure}{0.3\textwidth}
\includegraphics[width=\textwidth]{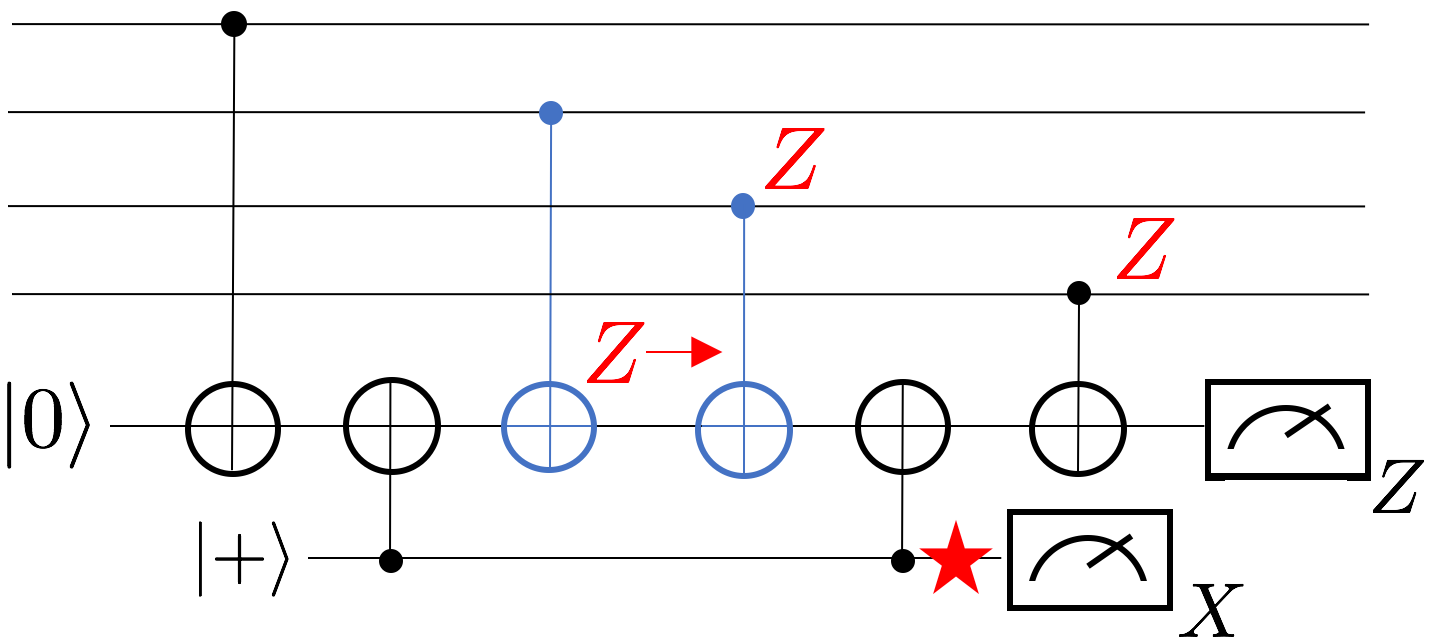}
\caption{}
\label{fig:StabFTwithAncilla}
\end{subfigure}
\begin{subfigure}{0.25\textwidth}
\includegraphics[width=\textwidth]{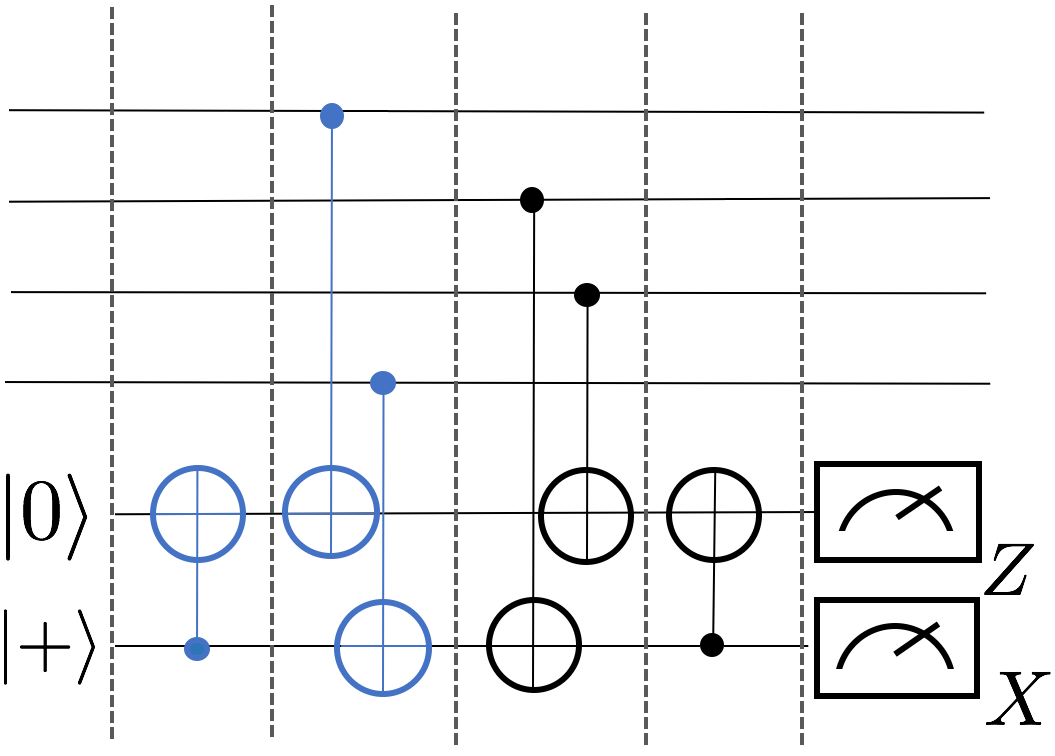}
\caption{}
\label{fig:StabFTwithAncillaCompact}
\end{subfigure}
\caption{Circuits for measuring the operator $ZZZZ$ (can be converted to any Pauli by single qubit Cliffords). (a) Non-fault-tolerant circuit. A single fault $IZ$ occurring on the third CNOT (from the left) results in the error $IIZZ$ on the data block.  (b) Flag version of \cref{fig:StabNonFT}. An ancilla (flag) qubit prepared in $\ket{+}$ and two extra CNOT gates signals when a weight two data error is caused by a single fault. Subsequent rounds of error correction may identify which error occurred. Consider an $IZ$ error on the second CNOT, in the non-flag circuit this would result in a weight two error, but in this case, this fault causes the circuit to flag. (c) An alternative flag circuit with lower depth than (b). All bad locations are illustrated in blue.}
\label{fig:ErrorPropSteane}
\end{figure}

As shown in \cref{fig:StabFTwithAncilla}, the circuit can be modified by including an additional ancilla (flag) qubit, and two extra CNOT gates.
This modification leaves the bad locations and the fault-free action of the circuit unchanged.
However, any single fault leading to an error $E$ with $\mathrm{wt}(E) \ge 2$ will also cause the measurement outcome of the flag qubit to flip \cite{CR17v1}. The following definitions will be useful:

\begin{definition}{\underline{Flags and measurements}}

Consider a circuit for measuring a stabilizer generator that includes at least one flag ancilla. The ancilla used to infer the stabilizer outcome is referred to as the \textit{measurement qubit}. We say the circuit has flagged if the eigenvalue of a flag qubit is measured as $-1$. If the eigenvalue of a measurement qubit is measured as $-1$, we will say that the measurement qubit flipped. 
\label{Def:GlagMeasureQubitsDef}
\end{definition}

The purpose of flag qubits is to signal when high weight data qubit errors result from few fault locations during a stabilizer measurement. Two key definitions are:

\begin{definition}{\underline{t-flag circuit}}

A circuit\footnote{To avoid confusing the notation of $C(P)$ that represents a circuit and $C$ that represents a code space, we always include the measured Pauli in parenthesis unless clear from context.} $C(P)$ which, when fault-free, implements a projective measurement of a weight-$w$ Pauli $P$ without flagging is a $t$-flag circuit if the following holds: For any set of $v$ faults at up to $t$ locations in $C(P)$ resulting in an error $E$ with $\text{min}(\text{wt}(E),\text{wt}(E P)) > v$, the circuit flags. 
\label{Def:tFlaggedCircuitDef}
\end{definition}
Note that a $t$-flag circuit for measuring a weight-$t$ stabilizer $P$ is also a $k$-flag circuit for any $k>t$. In \cref{app:GeneralTflaggedCircuitConstruction} we give constructions for some $t$-flag circuits.
\begin{definition}{\underline{Flag error set}}
	
	Let $\mathcal{E}(g_{i})$ be the set of all errors caused by one fault which caused the circuit $C(g_i)$ to flag. 
	\label{Def:FlagErrSetDef1}
\end{definition}
Note that the flag error set can contain the identity as well as weight one errors.

Suppose all errors in a flag error set $\mathcal{E}(g)$ for a 1-flag circuit $C(g)$ have distinct syndromes.  
As $C(g)$ is a 1-flag circuit, a single fault that leads to an error of weight greater than one will cause the circuit $C(g)$ to flag. 
Moreover, when a flag has occurred due to at most one fault, a complete set of fault-free stabilizer measurements will infer the resulting element of the flag error set which has been applied to the data qubits. In fact, one would only require distinct syndromes for errors in the flag error set that are logically inequivalent, as defined in \cref{Def:LogEquivDef}.

As an example, consider the 1-flag circuit in \cref{fig:StabFTwithAncilla}. A single fault at any of the blue CNOT gates can lead to an error $E_{b}$ with $\text{wt}(E_{b}) \le 2$ on the data. The set $\mathcal{E}(Z^{\otimes 4})$ contains all errors $E_{b}$ which resulted from a fault at a blue CNOT gate causing the circuit $C(Z^{\otimes 4})$ of \cref{fig:StabFTwithAncilla} to flag, i.e., 
$\mathcal{E}(g) = \{ I,Z_{q_{3}}Z_{q_{4}},X_{q_{2}}Z_{q_{3}}Z_{q_{4}},Z_{q_{1}}X_{q_{2}},Z_{q_{4}},$ $X_{q_{3}}Z_{q_{4}},Y_{q_{3}}Z_{q_{4}} \}$ with qubits $q_1$ to $q_4$.

With the above definitions, we can construct a fault-tolerant flag error correction protocol for $d=3$ stabilizer codes satisfying the following condition.

\begin{definition}{\underline{\textbf{Flag $1$-FTEC condition:}}}

Consider a stabilizer code $\mathcal{S} =  \langle g_{1},g_{2},\cdots , g_{r} \rangle$ and $1$-flag circuits $\{ C(g_{1}),C(g_{2}), \cdots , C(g_{r}) \}$. For every generator $g_{i}$, all pairs of elements $E,E'\in \mathcal{E}(g_{i})$ satisfy $s(E)\neq s(E')$ or $E \sim E'$.
\label{Def:Flag1FTECcondition}
\end{definition}

In other words, we require that any two errors that arise when a circuit $C(g_{i})$ flags due to a single fault must be either distinguishable or logically equivalent. For the following protocol to satisfy the FTEC conditions in \cref{Def:FaultTolerantDef}, one can assume there is at most 1 fault. If the Flag $1$-FTEC condition is satisfied, the protocol is implemented as follows: 

\vspace{17px}
\fbox{\begin{minipage}{23em}
        \textbf{Flag $1$-FTEC Protocol:}

Repeat the syndrome measurement using flag circuits until one of the following is satisfied: 

\begin{enumerate}
   \item If the syndrome $s$ is repeated twice in a row and there were no flags, apply the correction $E_{\text{min}}(s)$. 
   \item If there were no flags and the syndromes $s_{1}$ and $s_{2}$ from two consecutive rounds differ, repeat the syndrome measurement using non-flag circuits yielding syndrome $s$. Apply the correction $E_{\text{min}}(s)$. 
   \item If a circuit $C(g_{i})$ flags, stop and repeat the syndrome measurement using non-flag circuits yielding syndrome $s$. If there is an element $E \in \mathcal{E}(g_{i})$ which satisfies $s(E)=s$, then apply $E$, otherwise apply $E_{\text{min}}(s)$.
\end{enumerate}
\end{minipage}}

\begin{figure}
\centering
\includegraphics[width=0.5\textwidth]{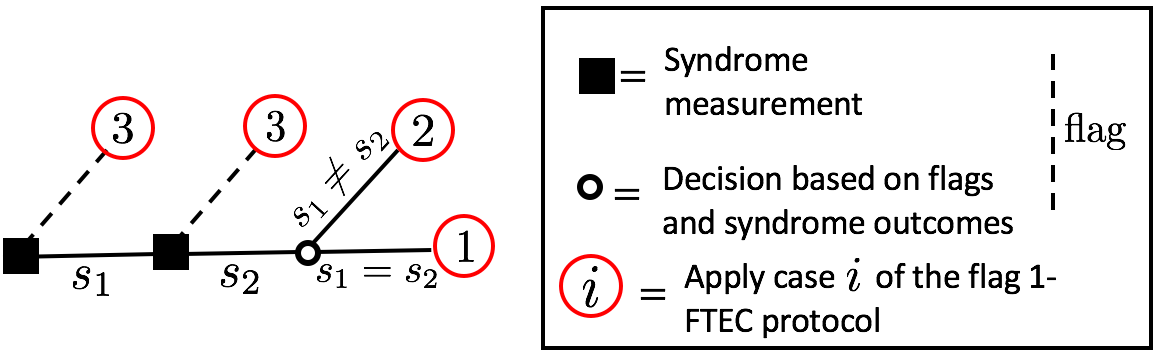}
\caption{Tree diagram illustrating the possible paths of the Flag $1$-FTEC Protocol. Numbers enclosed in red circles at the end of the edges indicate which step to implement in the Flag $1$-FTEC Protocol. A dashed line is followed when any of the 1-flag circuits $C(g_{i})$ flags. Solid squares indicate a syndrome measurement using 1-flag circuits whereas rings indicate a decision based on syndrome outcomes.  Note that the syndrome measurement is repeated at most three times.}
\label{fig:TreeD3Diag}
\end{figure}

\begin{figure}
\centering
\begin{subfigure}{0.45\textwidth}
\includegraphics[width=\textwidth]{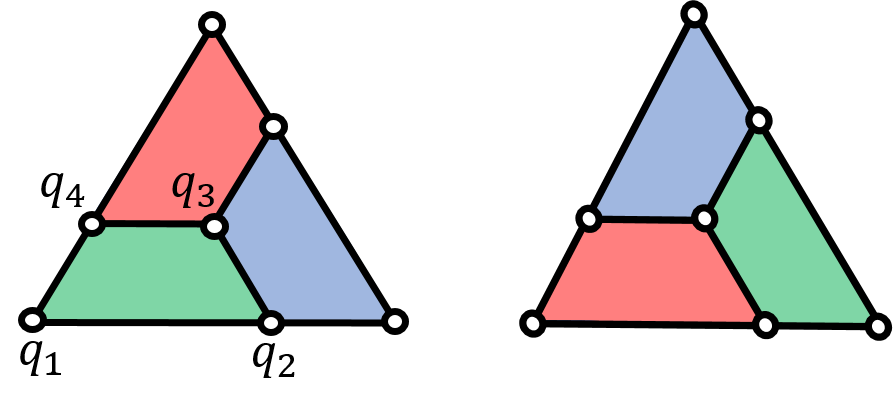}
\caption{}
\label{fig:SteaneColor1}
\end{subfigure}
\begin{subfigure}{0.45\textwidth}
\includegraphics[width=\textwidth]{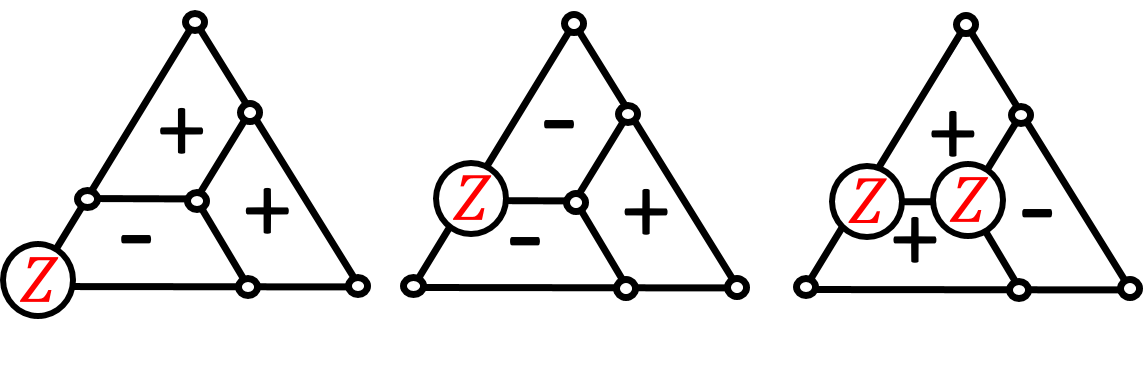}
\caption{}
\label{fig:SteaneColor1}
\end{subfigure}
\caption{
	(a) A representation of the Steane code where each circle is a qubit, and there is an $X$- and a $Z$-type stabilizer generator for each face. 
	Stabilizer cicuits are specified from that in Fig.~\ref{fig:ErrorPropSteane}(a) after rotating the lattice such that the relevant face is on the bottom left. 
	(b) For $g = Z_{q_{1}}Z_{q_{2}}Z_{q_{3}}Z_{q_{4}}$, the flag error set is $\mathcal{E}(g) = \{ I,Z_{q_{3}}Z_{q_{4}},X_{q_{2}}Z_{q_{3}}Z_{q_{4}},Z_{q_{1}}X_{q_{2}},Z_{q_{4}},$ $X_{q_{3}}Z_{q_{4}},X_{q_{3}}Z_{q_{3}}Z_{q_{4}} \}$ which contains all errors arising from a single fault that causes the stabilizer measurement circuit $C(g)$ to flag. 
	Since the Steane code is a CSS code, the $X$ component of an error will be corrected independently allowing us to consider the $Z$-part of the flag error set $\mathcal{E}_Z(g)=\{I,Z_{q_1},Z_{q_4},Z_{q_3}Z_{q_4} \}$.
	As required, the elements of $\mathcal{E}_Z(g)$ all have distinct syndromes (with satisfied stabilizers represented by a plus).
}
\label{fig:Steane}
\end{figure}

A tree diagram for the flag $1$-FTEC Protocol is illustrated in \cref{fig:TreeD3Diag}. We now outline the proof that the flag 1-FTEC protocol satisfies the fault-tolerance criteria of \cref{Def:FaultTolerantDef} (a more rigorous proof of the general case is presented in \cref{app:ProtocolGeneralProof}). To show that Flag $1$-FTEC Protocol satisfies the criteria of \cref{Def:FaultTolerantDef}, we can assume there is at most one fault during the protocol. If a single fault occurs in either the first or second round leading to a flag, repeating the syndrome measurement will correctly diagnose the error. If there are no flags and a fault occurs which causes the syndromes in the first two rounds to change, then the syndrome during the third round will correctly diagnose the error. There could also be a fault during either the first or second round that goes undetected. But since there were no flags it cannot spread to an error of weight-2. In this case applying a minimum weight correction based on the measured syndrome of the second round will guarantee that the output codeword differs from a valid codeword by an error of weight at most one. Note that the above argument applies irrespective of any errors on the input state, hence the second criteria of \cref{Def:FaultTolerantDef} is satisfied. It is worth pointing out that up to three repetitions are required in order to guarantee that the second criteria of \cref{Def:FaultTolerantDef} is satisfied (unless the code has the property that all states are at most a weight-one error away from a valid codeword, as in \cite{CR17v1}).

The Steane code is an example which satisfies the Flag $1$-FTEC condition with a simple choice of circuits.
To verify this, the representation of the Steane code given in \cref{fig:SteaneColor1} is useful.
There is an $X$- and a $Z$-type stabilizer generator supported on the four qubits of each of the three faces.
First let us specify all six stabilizer measurement circuits.
The circuit that measures $Z_{q_1}Z_{q_2}Z_{q_3} Z_{q_4}$ is specified by taking qubits $q_1$, $q_2$, $q_3$, and $q_4$ to be the four data qubits in descending order in the 1-flag circuit in \cref{fig:StabFTwithAncilla}. 
The other two $Z$-stabilizer measurement circuits are obtained by first rotating \cref{fig:SteaneColor1} by $120^{\circ}$ and $240^{\circ}$ and then using \cref{fig:StabFTwithAncilla}. 
The $X$-stabilizer circuit for each face is the same as the $Z$-stabilizer circuit for that face, replacing CNOT gates acting on data qubits by XNOT gates. The $Z$ component of the flag error set of the circuit in \cref{fig:StabFTwithAncilla} is $\mathcal{E}_Z(Z_{q_1}Z_{q_2}Z_{q_3}Z_{q_4}) = \{ I,Z_{q_1},Z_{q_4},Z_{q_3}Z_{q_4} \}$.
As can be seen from \cref{fig:SteaneColor1}, each of these has a distinct syndrome, thus the measurement circuit for $Z_{q_1}Z_{q_2}Z_{q_3} Z_{q_4}$ satisfies the flag $1$-FTEC condition, as do the remaining five measurement circuits by symmetry.

\subsection{Flag $2$-FTEC with distance-5 codes}
\label{subsec:Distance5protocol}

Before explicitly describing the conditions and protocol, we discuss some of the complications that arise for codes with $d>3$.

For distance-5 codes, we must ensure that if two faults occur during the error correction protocol, the output state will differ from a codeword by an error of at most weight-two. For instance, if two faults occur in a circuit for measuring a stabilizer of weight greater than four, the resulting error $E$ on the data should satisfy $\text{wt}(E) \le 2$ unless there is a flag. In other words, all stabilizer generators should be measured using 2-flag circuits. 

In another case, two faults could occur during the measurement of \textit{different} stabilizer generators $g_{i}$ and $g_{j}$. If two bad locations fail and are both flagged, and assuming there are no more faults, the measured syndrome will correspond to the product of the error caused in each circuit (which could have weight greater than two). Consequently, one should modify \cref{Def:FlagErrSetDef1} of the flag error set to include these types of errors.  
One then decodes based on the pair of errors that resulted in the measured syndrome, provided logically inequivalent errors have distinct syndromes. 

Before stating the protocol, we extend some definitions from \cref{subsec:ReviewChaoReichardt}.

Consider a stabilizer code $\mathcal{S} =  \langle g_{1},g_{2},\cdots , g_{r} \rangle$ and $t$-flag circuits $C(g_{i})$ for measuring the generator $g_{i}$. 
\begin{definition}{\underline{Flag error set}}

Let $\mathcal{E}_{m}(g_{i_{1}},\cdots , g_{i_{k}})$ be the set of all errors caused by precisely $m$ faults spread amongst the circuits $C(g_{i_{1}}),C(g_{i_{2}}), \cdots , C(g_{i_{k}})$ which all flagged. 
\label{Def:FlagErrSetDef}
\end{definition}
Note that there could be more than one fault in a single circuit $C(g_{i_{k}})$. Examples of flag error sets are given in \cref{tab:PossibleCorrelatedErrors} where only contributions from $Z$ errors are included (since the considered code is a CSS code).  We also define a general $t$-fault correction set:

\begin{align}
\tilde{E}_{t}^{m}(g_{i_{1}},\cdots , g_{i_{k}},s) =
\begin{cases}
\{ E \in \mathcal{E}_{m}(g_{i_{1}},\cdots , g_{i_{k}}) \times \mathcal{E}_{t-m}  \\
 \text{ such that } s(E) = s \} \\
\{ E_{\text{min}}(s) \} \text{ if above set empty. }
\end{cases}
\label{eq:GeneralLookupTable}
\end{align}

By $E \in \mathcal{E}_{m}(g_{i_{1}},\cdots , g_{i_{k}}) \times \mathcal{E}_{t-m}$, we are considering the set consisting of products between errors caused by $k$ flags and any error of weight $t-m$.

As will be seen below, the correction set will form a critical part of the protocol by specifying the correction applied based on the measured syndrome and flag outcomes over multiple syndrome measurement rounds. In the case where $k$ $t$-flag circuits flagged caused by $k \le m \le t$ faults, the correction applied to the data block will correspond to an element of $\mathcal{E}_{m}(g_{i_{1}},\cdots , g_{i_{k}}) \times \mathcal{E}_{t-m}$ if the measured syndrome corresponds to an element in this set (there could also be $t-m$ faults which did not give rise to a flag). However in practice, there could be more than $t$ faults and so the measured syndrome may not be consistent with any element of the set $\mathcal{E}_{m}(g_{i_{1}},\cdots , g_{i_{k}}) \times \mathcal{E}_{t-m}$. In this case, and for the error correction protocol to satisfy the second criteria of \cref{Def:FaultTolerantDef}, the correction will correspond to $E_{\text{min}}(s)$. These features are all included in the set $\tilde{E}_{t}^{m}(g_{i_{1}},\cdots , g_{i_{k}},s)$.

\begin{definition}{\underline{\textbf{Flag $2$-FTEC condition:}}}

Consider a stabilizer code $\mathcal{S} =  \langle g_{1},g_{2},\cdots , g_{r} \rangle$ and $2$-flag circuits $\{ C(g_{1}),C(g_{2}), \cdots , C(g_{r}) \}$. For any choice of generators $\{ g_{i}, g_{j} \}$:
\begin{enumerate}
\item $E,E' \in \mathcal{E}_{2}(g_{i},g_{j}) \Rightarrow  s(E)\neq s(E')$ or $E \sim E'$,
\item $E,E' \in \mathcal{E}_{2}(g_{i})\cup  (\mathcal{E}_{1}(g_{i}) \times \mathcal{E}_{1}) \Rightarrow s(E) \neq s(E')$ or $E \sim E'$.
\end{enumerate}
\label{Def:Flag2FTECcondition}
\end{definition}

In order to state the protocol, we define an update rule given a sequence of syndrome measurements using $t$-flag circuits for the counters\footnote{$n_{\text{diff}}$ tracks the minimum number of faults that could have caused the observed syndrome outcomes.  For example, if the sequence $s_{1},s_{2},s_{1}$ was measured, $n_{\text{diff}}$ would increase by one since a single measurement fault could give rise to the given sequence (for example, this could be caused by a single CNOT failure which resulted in a data qubit and measurement error). However for the sequence $s_{1},s_{2},s_{1},s_{2}$,  $n_{\text{diff}}$ would increase by two.}
$n_{\text{diff}}$ and $n_{\text{same}}$ as follows:

\vspace{10px}

\fbox{\begin{minipage}{23em}
\underline{\textbf{Flag $2$-FTEC protocol -- update rules:}} 

Given a sequence of consecutive syndrome measurement outcomes  $s_{k}$ and $s_{k+1}$:
\begin{enumerate}
\item If $n_{\text{diff}}$ didn't increase in the previous round, and $s_{k}\neq s_{k+1}$, increase $n_{\text{diff}}$ by one. 
\item If a flag occurs, reset $n_{\text{same}}$ to zero.

\item If $s_{k} = s_{k+1}$, increase $n_{\text{same}}$ by one.
\end{enumerate}
\end{minipage}}

\vspace{10px}

For the following protocol to satisfy \cref{Def:FaultTolerantDef}, one can assume there are at most 2 faults. If the Flag $2$-FTEC condition is satisfied, the protocol is implemented as follows:

\vspace{10px}

\fbox{\begin{minipage}{23em}
\underline{\textbf{Flag $2$-FTEC protocol -- corrections:}} 

Set $n_{\text{diff}}=0$ and $n_{\text{same}} = 0$.

Repeat the syndrome measurement using flag circuits until one of the following is satisfied: 
   \begin{enumerate}
    \item The same syndrome $s$ is repeated $3-n_{\text{diff}}$ times in a row and there were no flags, apply the correction $E_{\text{min}}(s)$.
     \item There were no flags and $n_{\text{diff}}=2$. Repeat the syndrome measurement using non-flag circuits yielding syndrome $s$. Apply the correction $E_{\text{min}}(s)$.  
   \item Some set of two circuits $C(g_{i})$ and $C(g_{j})$ have flagged. Repeat the syndrome measurement using non-flag circuits yielding syndrome $s$. Apply any correction from the set $\tilde{E}_{2}^{2}(g_{i},g_{j},s)$.  
    \item Any circuit $C(g_{i})$ has flagged and $n_{\text{diff}}=1$. Repeat the syndrome measurement using non-flag circuits yielding syndrome $s$. Apply any correction from the set $\tilde{E}_{2}^{1}(g_{i},s)$.
   \item Any circuit $C(g_{i})$ has flagged and $n_{\text{diff}}=0$ and $n_{\text{same}}=1$. Use the measured syndrome $s$ from the last round. Apply any correction from the set $\tilde{E}_{2}^{1}(g_{i},s)\cup \tilde{E}_{2}^{2}(g_{i},s)$.     
         \end{enumerate}
\end{minipage}}

\begin{figure}
	\centering
	\includegraphics[width=0.5\textwidth]{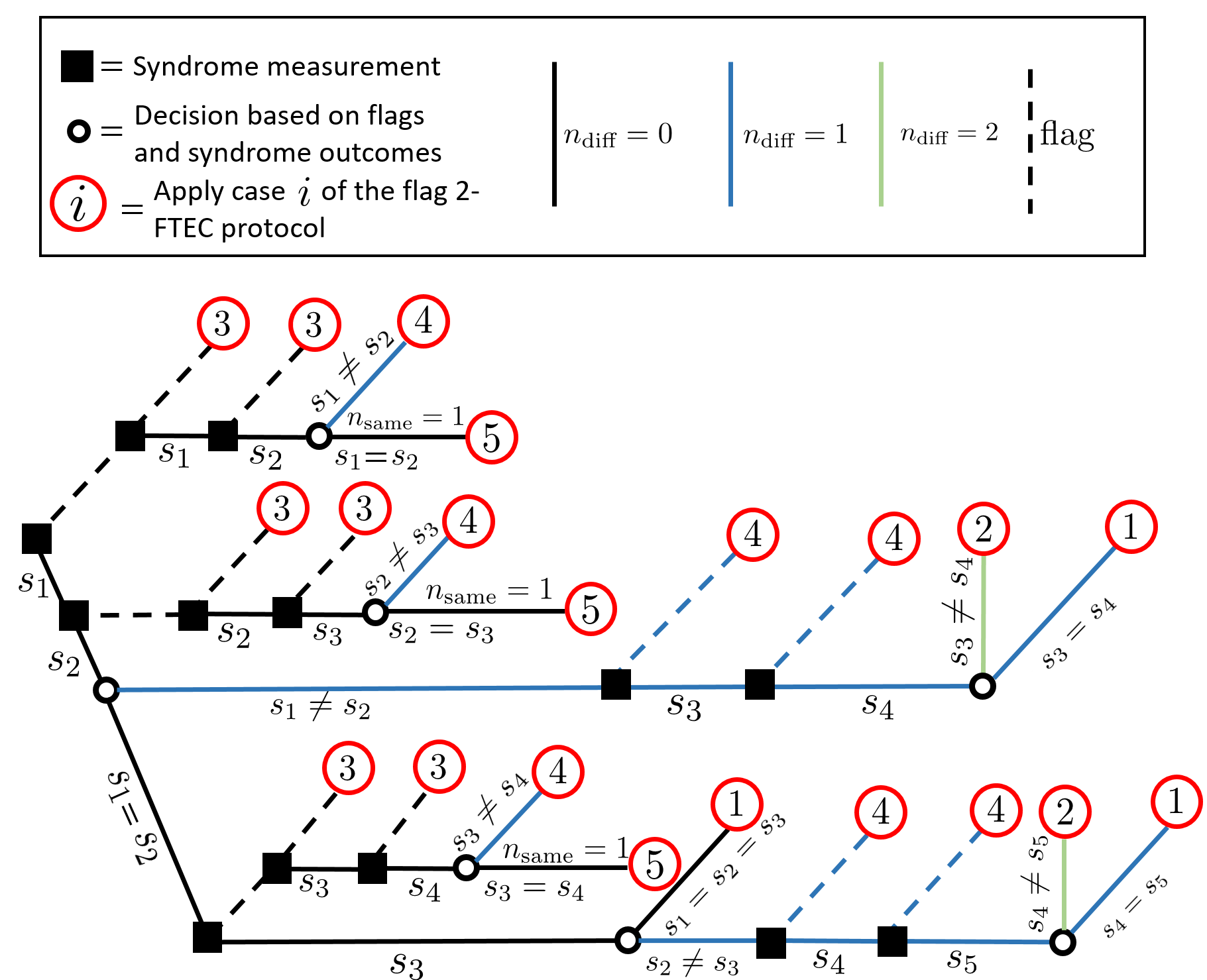}
	\caption{Tree diagram for the Flag $2$-FTEC protocol. Numbers encircled in red at the end of the edges indicate which step to implement in the Flag $2$-FTEC Protocol. A dashed line is followed when any of the 2-flag circuits $C(g_{i})$ flags. Solid squares indicate a syndrome measurement using 2-flag circuits whereas rings indicate a decision based on syndrome outcomes. Edges with different colors indicate the current value of $n_{\text{diff}}$ in the protocol. Note that the protocol is repeated at most 6 times.}
	\label{fig:TreeDiagramD5}
\end{figure}

\vspace{10px}

Note that when computing the update rules, if a flag occurs during the $j$'th round of syndrome measurements, the syndrome is not recorded for that round since all stabilizers must be measured. Thus when computing $n_{\text{diff}}$ and $n_{\text{same}}$ using consecutive syndromes $s_k$ and $s_{k+1}$, we are assuming that no flags occurred during rounds $k$ and $k+1$.

In each case of the protocol, the correction sets correspond to those data errors which could arise from up to two faults which are consistent with the conditions of the case.   
As the elements are logically equivalent (by \cref{eq:GeneralLookupTable,Def:Flag2FTECcondition}), which element is applied is unimportant.

The general protocol for codes of arbitrary distance is given in  \cref{app:GeneralFTEC}.

\subsection{Examples of flag 2-FTEC applied to $d=5$ codes}
\label{subsec:ApplicationProtocolColorCode}

\begin{figure}
\centering
\begin{subfigure}{0.22\textwidth}
\includegraphics[width=\textwidth]{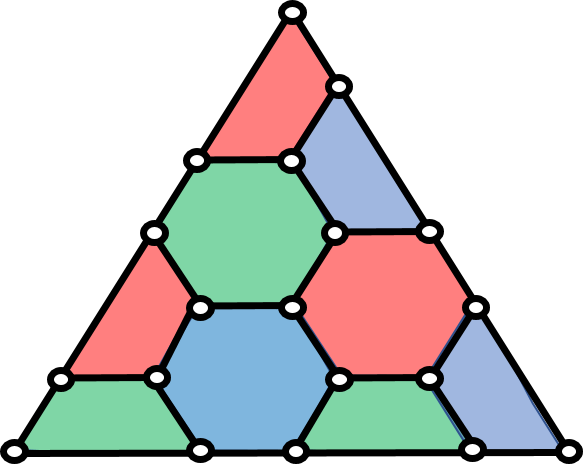}
\caption{}
\label{fig:19qubitLatticeColor}
\end{subfigure}
\begin{subfigure}{0.25\textwidth}
\includegraphics[width=\textwidth]{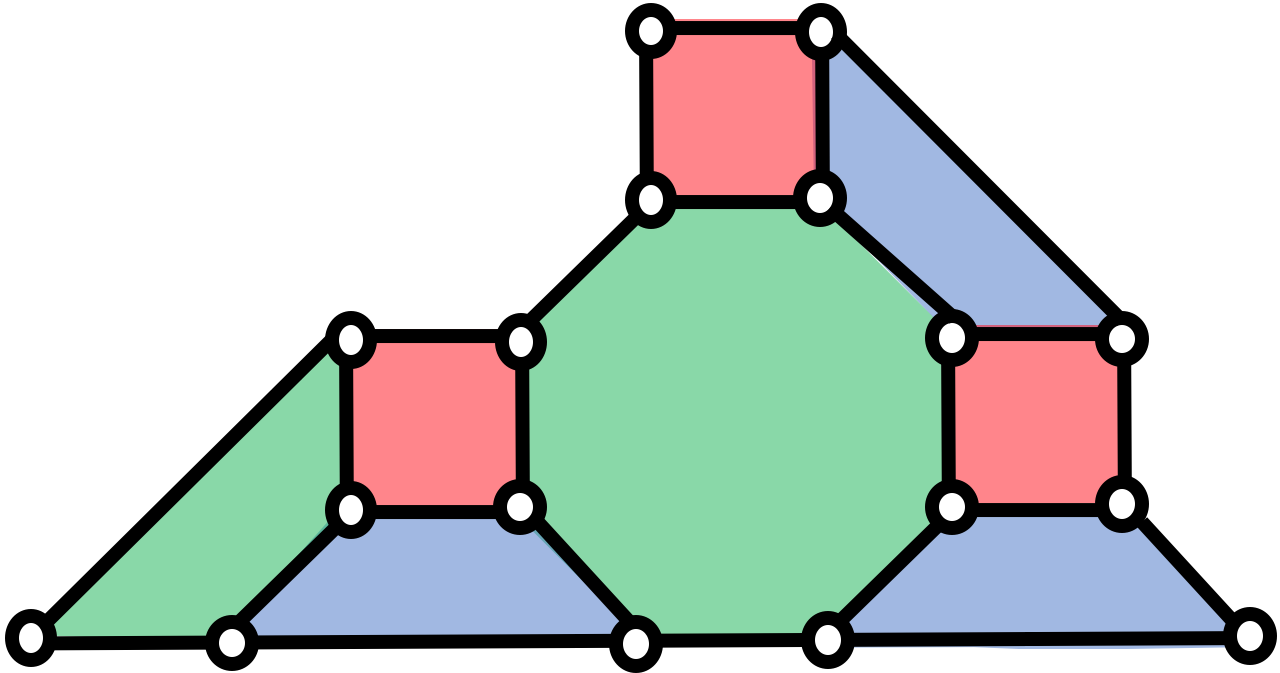}
\caption{}
\label{fig:17qubitLatticeColor}
\end{subfigure}
\caption{Graphical representation of (a) the 19-qubit 2D color code and (b) the 17-qubit 2D color code. The $X$ and $Z$ stabilizers of the code are symmetric, given by the vertices of each plaquette. Both codes have distance-5.}
\label{fig:ColorCodeLattices}
\end{figure}

In this section we give examples of the flag $2$-FTEC protocol applied to the 2-dimensional \codepar{19,1,5} and \codepar{17,1,5} color codes, (see \cref{fig:19qubitLatticeColor,fig:17qubitLatticeColor}). We first find 2-flag circuits for all generators (weight-4 and -6 for the 19-qubit code and weight-4 and -8 for the 17-qubit code). We also show that the flag 2-FTEC condition is satisfied for both codes. 

\begin{table}[t]
\begin{tabular}{ c|c|c|c}
\multicolumn{2}{c|}{Weight-4 measurement} & \multicolumn{2}{|c}{Weight-6 measurement}\\ \hline
1-fault & 2-faults & 1-fault & 2-faults  \\ \hline
$I$,$Z_{1}$ & $I$,$Z_{1}$,$Z_{2}$ & $I$,$Z_{1}$,$Z_{6}$ & $I$,$Z_{1}$,$Z_{2}$\\
$Z_{4}$ & $Z_{3}$,$Z_{4}$&$Z_{1}Z_{2}$  &$Z_{3}$, $Z_{4}$,$Z_{5}$,$Z_{6}$ \\
$Z_{3}Z_{4}$ &$Z_{1}Z_{2}$ & $Z_{5}Z_{6}$ & $Z_{1}Z_{2}$,$Z_{1}Z_{3}$  \\ 
&  $Z_{1}Z_{4}$& $Z_{4}Z_{5}Z_{6}$ & $Z_{1}Z_{4}$,$Z_{1}Z_{5}$  \\   
&  $Z_{2}Z_{4}$& & $Z_{1}Z_{6}$,$Z_{2}Z_{3}$  \\    
&  & & $Z_{2}Z_{6}$,$Z_{3}Z_{4}$    \\  
&  & & $Z_{3}Z_{6}$,$Z_{4}Z_{5}$  \\   
&  & & $Z_{4}Z_{6}$,$Z_{5}Z_{6}$  \\       
&  & & $Z_{1}Z_{2}Z_{3}$,$Z_{1}Z_{5}Z_{6}$  \\   
&  & & $Z_{2}Z_{5}Z_{6}$,$Z_{3}Z_{4}Z_{5}$ \\  
&  & & $Z_{3}Z_{4}Z_{6}$,$Z_{3}Z_{5}Z_{6}$ \\
&  & &  $Z_{4}Z_{5}Z_{6}$\\
\end{tabular}
\caption{$Z$ part of the flag error set of \cref{Def:FlagErrSetDef} for flag circuits used to measure the stabilizers $g_{1} = Z_{1}Z_{2}Z_{3}Z_{4}$ and $g_{3} = Z_{1}Z_{2}Z_{3}Z_{4}Z_{5}Z_{6}$ (we removed errors equivalent up to the stabilizer being measured).} 
\label{tab:PossibleCorrelatedErrors}
\end{table}

For a 2-flag circuit, two faults leading to an error of weight greater or equal to 3 (up to multiplication by the stabilizer) must always cause at least one of the flag qubits to flag. As shown in \cref{app:GeneralTflaggedCircuitConstruction}, a 2-flag circuit satisfying these properties can always be constructed using at most four flag qubits. We show 2-flag circuits for measuring weight six and eight generators in \cref{fig:Flag2CircuitExamples}.

In \cref{subsec:Remarks}, it will be shown that the family of color codes with a hexagonal lattice satisfy a sufficient condition which guarantees that the flag 2-FTEC condition is satisfied. However, there are codes that do not satisfy the sufficient condition but which nonetheless satisfy the 2-Flag FTEC condition. For the 19-qubit and 17-qubit color codes, we verified that the flag 2-FTEC condition was satisfied by enumerating all errors as one would have to for a generic code. In particular, in the case where the 2-flag circuits $C(g_{i})$ and $C(g_{j})$ flag, the resulting errors belonging to the set $\mathcal{E}_{2}(g_{i},g_{j})$ must be logically equivalent or have distinct syndromes (which we verified to be true). If a single circuit $C(g_{i})$ flags, there could either have been two faults in the circuit or a single fault along with another error that did not cause a flag. If the same syndrome is measured twice in a row after a flag, then errors in the set $\mathcal{E}_{2}(g_{i})\cup  (\mathcal{E}_{1}(g_{i}) \times \mathcal{E}_{1})$ must be logically equivalent or have distinct syndromes (which we verified). If there is a flag but two different syndromes are measured in a row, errors belonging to the set  $\mathcal{E}_{1}(g_{i}) \times \mathcal{E}_{1}$ must be logically equivalent or have distinct syndromes (as was already checked). The flag error sets (see \cref{Def:FlagErrSetDef}) for the 19-qubit code can be obtained using the Pauli's shown in \cref{tab:PossibleCorrelatedErrors}.

\begin{figure}
\centering
\begin{subfigure}{0.32\textwidth}
\includegraphics[width=\textwidth]{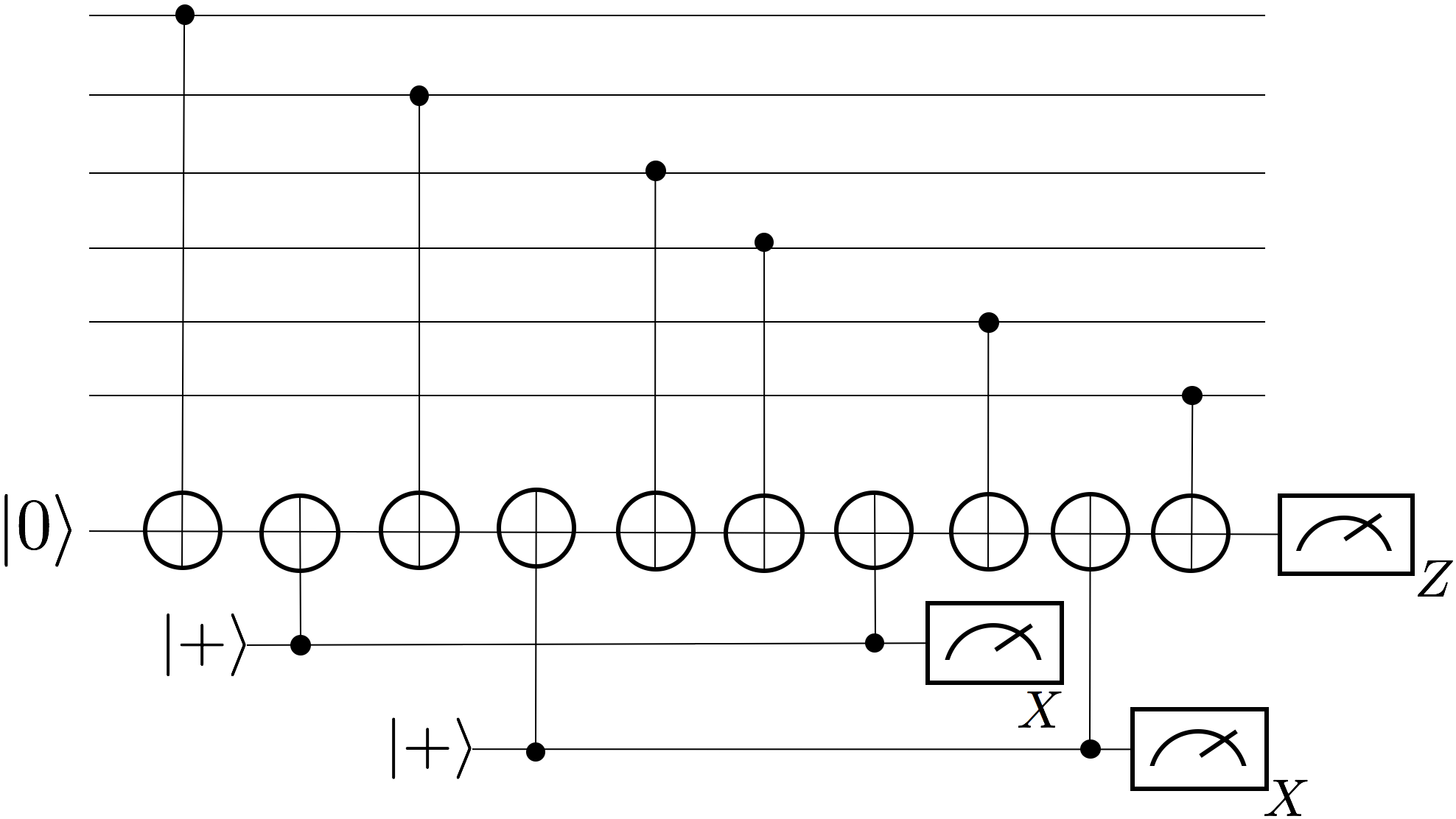}
\caption{}
\label{fig:WeightSixGenerators}
\end{subfigure}
\begin{subfigure}{0.37\textwidth}
\includegraphics[width=\textwidth]{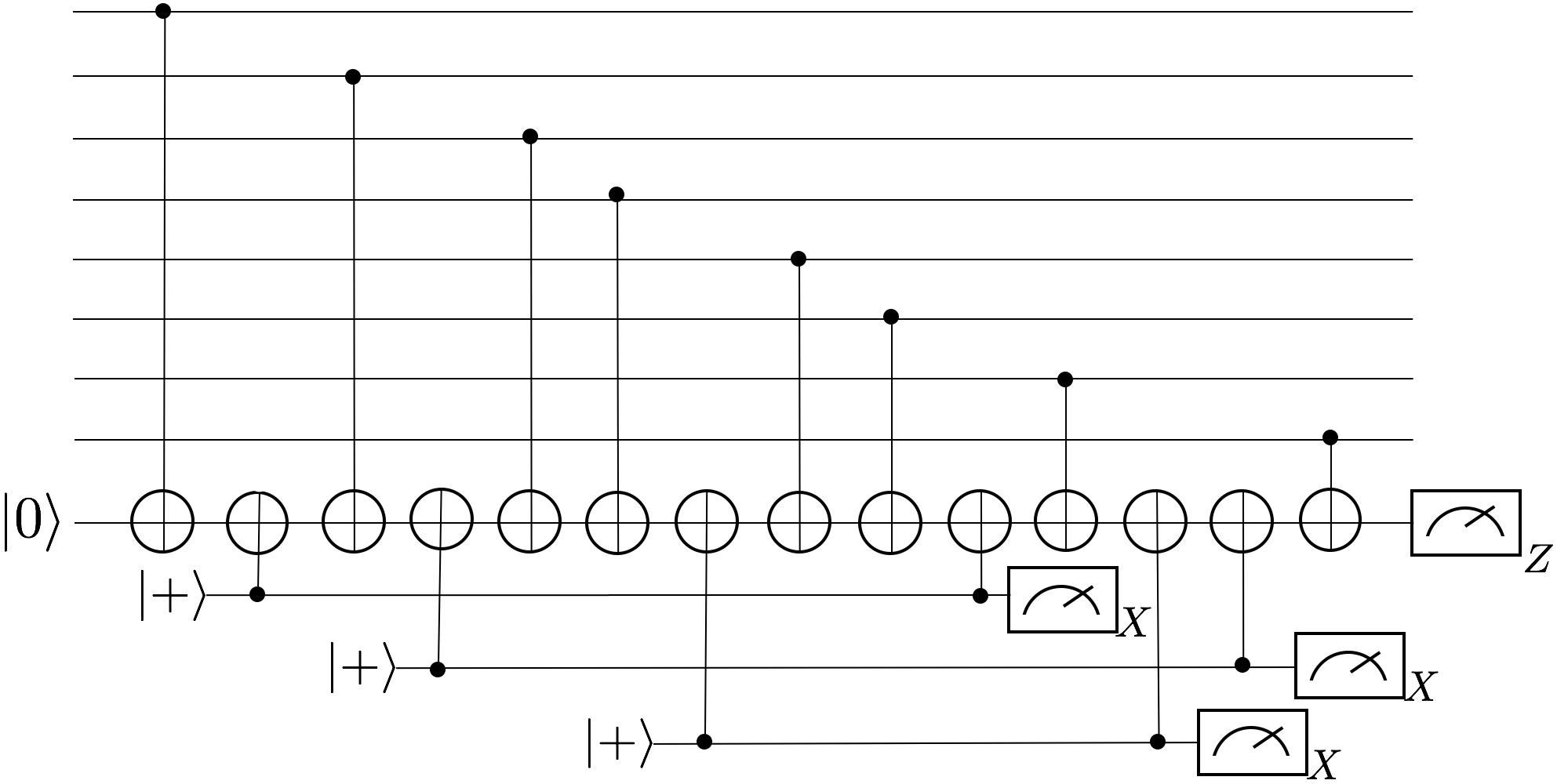}
\caption{}
\label{fig:Weight8Generator}
\end{subfigure}
\caption{
	Illustration of 2-flag circuits for measuring (a) $Z^{\otimes 6}$ requiring only two flag qubits and (b) $Z^{\otimes 8}$ requiring only three flag qubits. Flag qubits are prepared in the $\ket{+}$ state, and measurement qubits in the  $\ket{0}$ state.	
}
\label{fig:Flag2CircuitExamples}
\end{figure}

Given that the flag 2-FTEC condition is satisfied, the flag 2-FTEC protocol can be implemented following the steps of \cref{subsec:Distance5protocol} and the tree diagram illustrated in \cref{fig:TreeDiagramD5}. 

\section{Flag error correction protocol for arbitrary distance codes}
\label{app:GeneralFTEC}

In this section we first provide the general flag $t$-FTEC protocol in \cref{subsec:GeneralProtocol}. In \cref{subsec:Remarks} we give a sufficient condition for stabilizer codes that allow us to easily prove that flag FTEC can be applied to a number of infinite code families. We show that the families of surface codes, hexagonal lattice color codes and quantum Reed-Muller codes satisfy the sufficient condition. Lastly, in \cref{app:GeneralTflaggedCircuitConstruction}, we give general $t$-flag circuit constructions which are applicable to the code families described in \cref{subsec:Remarks}.  

We assume the reader is familiar with all previous definitions. However, to make this section reasonably self contained, we repeat some key definitions below.

\begin{definitionBob}{6}{\underline{$t$-flag ciruit}}
A circuit $C(P)$ which, when fault-free, implements a projective measurement of a weight-$w$ Pauli $P$ without flagging is a $t$-flag circuit if the following holds: For any set of $v$ faults at up to $t$ locations in $C(P)$ resulting in an error $E$ with $\text{min}(\text{wt}(E),\text{wt}(E P)) > v$, the circuit flags. 
\end{definitionBob}

\begin{definitionBob}{9}{\underline{Flag error set}}

Let $\mathcal{E}_{m}(g_{i_{1}},\cdots , g_{i_{k}})$ be the set of all errors caused by precisely $m$ faults spread amongst the circuits $C(g_{i_{1}}),C(g_{i_{2}}), \cdots , C(g_{i_{k}})$ which all flagged. 
\end{definitionBob}

We also remind the reader of the correction set
\begin{align}
\tilde{E}_{t}^{m}(g_{i_{1}},\cdots , g_{i_{k}},s) =
\begin{cases}
\{ E \in \mathcal{E}_{m}(g_{i_{1}},\cdots , g_{i_{k}}) \times \mathcal{E}_{t-m}  \\
 \text{ such that } s(E) = s \} \\
\{ E_{\text{min}}(s) \} \text{ if above set empty. }
\end{cases}
\label{eq:GeneralLookupTableV2}
\end{align}

\subsection{Conditions and protocol}
\label{subsec:GeneralProtocol}

In what follows we generalize the fault-tolerant error correction protocol presented in \cref{subsec:Distance5protocol} to stabilizer codes of arbitrary distance. 

\begin{definition}{\underline{\textbf{Flag $t$-FTEC condition:}}}

Consider a stabilizer code $\mathcal{S} =  \langle g_{1},g_{2},\cdots , g_{r} \rangle$ and $t$-flag circuits $\{ C(g_{1}),C(g_{2}), \cdots , C(g_{r}) \}$. For any set of $m$ stabilizer generators $\{ g_{i_{1}},\cdots ,g_{i_{m}} \}$ such that $1 \le m \le t$, every pair of elements $E,E' \in \bigcup_{j=0}^{t-m}\mathcal{E}_{t-j}(g_{i_{1}},\cdots ,g_{i_{m}})\times \mathcal{E}_{j}$ either satisfy $s(E)\neq s(E')$ or $E \sim E'$. 
\label{Def:FlagtFTECcondition}
\end{definition}

The above conditions ensure that if there are at most $t=\lfloor (d-1)/2 \rfloor$ faults, the protocol described below will satisfy the fault-tolerance conditions of \cref{Def:FaultTolerantDef}. 

In order to state the protocol, we define an update rule given a sequence of syndrome measurements using $t$-flag circuits for the counters $n_{\text{diff}}$ and $n_{\text{same}}$ as follows (see also \cref{subsec:Distance5protocol} and the associated footnote):

\vspace{10px}

\fbox{\begin{minipage}{23em}
\underline{\textbf{Flag $t$-FTEC protocol -- update rules:}} 

Given a sequence of consecutive syndrome measurement outcomes  $s_{k}$ and $s_{k+1}$:
\begin{enumerate}
\item If $n_{\text{diff}}$ didn't increase in the previous round, and $s_{k}\neq s_{k+1}$, increase $n_{\text{diff}}$ by one. 

\item If a flag occurs, reset $n_{\text{same}}$ to zero.

\item If $s_{k} = s_{k+1}$, increase $n_{\text{same}}$ by one.
\end{enumerate}
\end{minipage}}

\vspace{10px}

\fbox{\begin{minipage}{23em}
\underline{\textbf{Flag $t$-FTEC protocol -- corrections:}} 

Set $n_{\text{diff}}=0$ and $n_{\text{same}} = 0$.

Repeat the syndrome measurement using flag circuits until one of the following is satisfied: 
   \begin{enumerate}
   \item The same syndrome $s$ is repeated $t-n_{\text{diff}}+1$ times in a row and there are no flags, apply the correction $E_{\text{min}}(s)$.  
   \item There were no flags and $n_{\text{diff}}=t$. Repeat the syndrome measurement using non-flag circuits yielding the syndrome $s$. Apply the correction $E_{\text{min}}(s)$.
   \item Some set of $t$ circuits $\{ C(g_{i_{1}}), \cdots ,  C(g_{i_{t}}) \}$ have flagged. Repeat the syndrome measurement using non-flag circuits yielding syndrome $s$. Apply any correction from the set $\tilde{E}_{t}^{t}(g_{i_{1}},\cdots , g_{i_{t}},s)$.
    \item Some set of $m$ circuits $\{ C(g_{i_{1}}), \cdots ,  C(g_{i_{m}}) \}$ have flagged with $1 \le m < t$ and $n_{\text{diff}} = t -m $. Repeat the syndrome measurement using non-flag circuits yielding syndrome $s$. Apply any correction from the set $\tilde{E}_{t}^{m}(g_{i_{1}},\cdots , g_{i_{m}},s)$.
    \item Some set of $m$ circuits $\{ C(g_{i_{1}}), \cdots ,  C(g_{i_{m}}) \}$ have flagged with $1 \le m < t$; $n_{\text{diff}} <t-m$ and $n_{\text{same}} = t-m-n_{\text{diff}}+1$. Use the syndrome $s$ obtained during the last round and apply any correction from the set $\bigcup_{j=0}^{t-m-n_{\text{diff}}}\tilde{E}_{t}^{t-j-n_{\text{diff}}}(g_{i_{1}},\cdots ,g_{i_{m}}, s)$.        
         \end{enumerate}         
\end{minipage}}

\vspace{10px}

In each case of the protocol, the correction sets correspond to those data errors which could arise from up to $t$ faults which are consistent with the conditions of the case. As the elements are logically equivalent (by \cref{eq:GeneralLookupTableV2,Def:FlagtFTECcondition}), which element is applied is unimportant.

For the protocol to satisfy the fault-tolerance criteria, the syndrome measurement needs to be repeated a minimum of $t+1$ times. In the scenario where the most syndrome measurement rounds are performed, $t$ identical syndromes are obtained before a fault causes the $t+1$'th syndrome to change (in which case $n_{\text{diff}}$ would increase by one). Afterwords, one measures the same syndrome $t-1$ times in a row until another fault causes the syndrome to change. This continues until all of the $t$ possible faults have been exhausted. At this stage, $n_{\text{diff}}=t$ so an extra syndrome measurement round will be performed using non-flag circuits. Thus the maximum number of syndrome measurement rounds $n_{\text{max}}$ is given by
\begin{align}
n_{\text{max}} = \sum_{j=0}^{t-1}(t-j) + t+1 = \frac{1}{2}(t^{2}+3t+2). 
\label{Eq:Nmax}
\end{align}
Note that a similar approach by repeating syndrome measurements is used for Shor error correction \cite{AGP06,Gottesman2010}. However, our scheme requires fewer syndrome measurement repetitions than is often described for Shor error correction and does not require the preparation and verification of a $w$-qubit cat state when measuring a stabilizer of weight-$w$. \footnote{One could also define update rules analogous to those for $n_{\text{diff}}$ and $n_{\text{same}}$ when implementing Shor-EC which would only require at most $\frac{1}{2}(t^{2}+3t+2)$ syndrome measurement repetitions as in the flag $t$-FTEC protocol.}

For codes that satisfy the flag $t$-FTEC condition, we also show in \cref{app:StatePrepAndMeasure} how to fault-tolerantly prepare and measure logical states using the flag $t$-FTEC protocol.

\subsection{Sufficient condition and satisfying code families}
\label{subsec:Remarks}

The general flag $t$-FTEC condition can be difficult to verify for a given code since it depends on precisely which $t$-flag circuits are used.
A sufficient (but not necessary) condition that implies the flag $t$-FTEC condition is as follows:

\textbf{Sufficient flag $t$-FTEC condition:}

Given a stabilizer code with distance $d>1$, and $\mathcal{S} =  \langle g_{1},g_{2},\cdots , g_{r} \rangle$, we require that for all $v = 0,1, \dots t$, all choices $Q_{t-v}$ of $2(t-v)$ qubits, and all subsets of $v$ stabilizer generators $\{ g_{i_1},\dots ,g_{i_v} \} \subset \{ g_{1},\cdots , g_{r} \}$, there is no logical operator $l \in N(\mathcal{S}) \setminus \mathcal{S}$ such that
\begin{align}
\text{supp}(l) \subset \text{supp}(g_{i_1}) \cup \cdots \cup \text{supp}(g_{i_v}) \cup Q_{t-v},
\end{align}
where $N(\mathcal{S})$ is the normalizer of the stabilizer group.

If this condition holds, then the flag $t$-FTEC condition is implied for any choice of $t$-flag circuits $\{ C(g_{1}),C(g_{2}), \cdots , C(g_{r}) \}$.

To prove this, we must show that it implies that none of the sets appearing in the $t$-FTEC condition contain elements that differ by a logical operator. 
Consider the set $\bigcup_{j=0}^{t-m}\mathcal{E}_{t-j}(g_{i_{1}},\cdots ,g_{i_{m}})\times \mathcal{E}_{j}$ for some set of $m$ stabilizer generators $\{ g_{i_{1}},\cdots ,g_{i_{m}} \}$ with $1 \le m \le t$. 
An error $E$ from this set will have support in the union of the support of the $m$ stabilizer generators $\{ g_{i_{1}},\cdots ,g_{i_{m}} \}$, along with up to $t-m$ other single qubits. 
Another error $E'$ from this set will have support in the union of support of the same $m$ stabilizer generators $\{ g_{i_{1}},\cdots ,g_{i_{m}} \}$, along with up to $t-m$ other \textit{potentially different} single qubits. 
If the sufficient condition holds, then $\text{supp}(E E')$ cannot contain a logical operator.

The sufficient flag $t$-FTEC condition is straightforward to verify for a number of code families with a lot of structure in their stabilizer generators and logical operators.
We briefly provide a few examples.

\textbf{Surface codes flag $t$-FTEC:}

The rotated surface code \cite{KITAEV97Surface,TS14,BK98,PhysRevLett.90.016803} family \codepar{d^2,1,d} for all odd $d=2t+1$ (see \cref{fig:surfacecodeproof}) satisfies the flag $t$-FTEC condition using any 4-flag circuits.

Firstly, by performing an exhaustive search, we verified that the circuit of \cref{fig:StabFTwithAncilla} is a 4-flag circuit.

As a CSS code, we can restrict our attention to purely $X$-type and $Z$-type logical operators.
An $X$ type logical operator must have at least one qubit in each of the $2t+1$ rows of the lattice shown. 
However, each stabilizer only contains qubits in two different rows. 
Therefore, with $v$ stabilizer generators, at most $2v$ of the rows could have support.
With an additional $2(t-v)$ qubits, at most $2t$ rows can be covered, which is fewer than the number of rows, and therefore no logical $X$ operator is supported on the union of the support of $v$ stabilizers and $2(t-v)$ qubits.
An analogous argument holds for $Z$-type logical operators, therefore the sufficient $t$-FTEC condition is satisfied.

\begin{figure}
	\centering
	\includegraphics[width=0.3\textwidth]{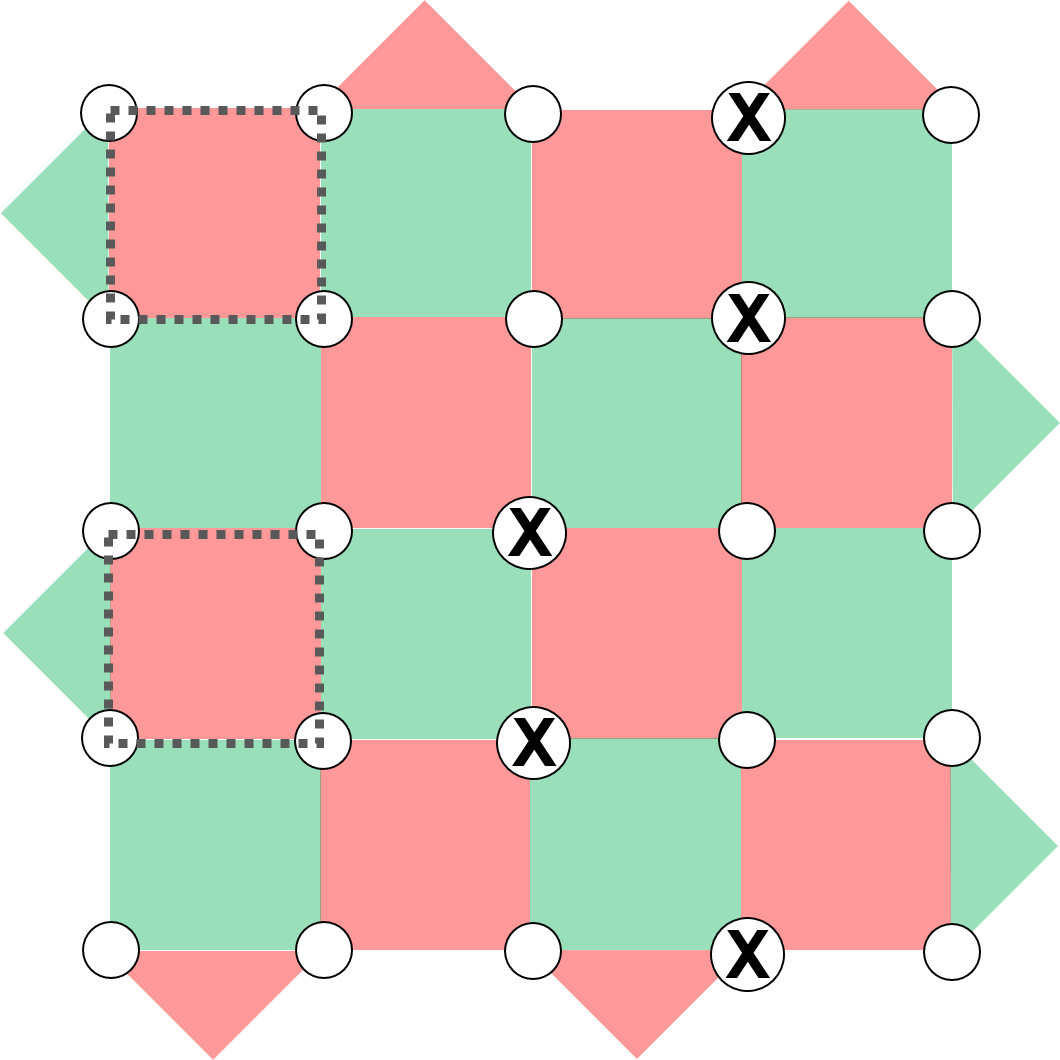}
	\caption{
		The $d=5$ rotated surface code. Qubits are represented by white circles, and $X$ and $Z$ stabilizer generators are represented by red and green faces. 
		As in the example, any logical $X$ operator has $X$ operators acting on at least five qubits, with at least one in each row of the lattice, involving an even number in any green face.
		In this case, no two stabilizer generators can have qubits in five rows, and therefore cannot contain an $X$ type logical operator.
		The argument is analogous for logical $Z$ operators.
	}
	\label{fig:surfacecodeproof}
\end{figure}

\textbf{Color codes flag $t$-FTEC:}

Here we show that any distance $d=(2t+1)$ self-dual CSS code with at most weight-6 stabilizer generators satisfies the flag $t$-FTEC condition using any 6-flag circuits (see \cref{fig:6flagCircuit} for an example).
Examples include the hexagonal color code \cite{Bombin06TopQuantDist} family \codepar{(3d^2+1)/4,1,d} (see \cref{fig:19qubitLatticeColor}).

As a self-dual CSS code, $X$ and $Z$ type stabilizer generators are identically supported and we can consider a pure $X$-type logical operator without loss of generality. 

Consider an $X$ type logical operator $l$ such that
\begin{align}
\text{supp}(l) \subset \text{supp}(g_{i_1}) \cup \cdots \cup \text{supp}(g_{i_v}) \cup Q_{t-v},
\label{eq:SuppRepeated}
\end{align}
for some set of $v$ stabilizer generators $\{ g_{i_1},\dots ,g_{i_v} \} \subset \{ g_{1},\cdots , g_{r} \}$ along with $2(t-v)$ other qubits $Q_{t-v}$.
Restricted to the support of any of the $v$ stabilizers $g_i$, $l|_{g_i}$ must have weight 0, 2, 4, or 6 (otherwise it would anti-commute with the corresponding $Z$ type stabilizer). 
If the restricted weight is 4 or 6, we can produce an equivalent lower weight logical operator $l' = g_i l$, which still satisfies \cref{eq:SuppRepeated}.
Repeating this procedure until the weight of the logical operator can no longer be reduced yields a logical operator $l_{\text{min}}$ which has weight either 0 or 2 when restricted to the support of any of the $v$ stabilizer generators. 
The total weight of $l_{\text{min}}$ is then at most $2v+2(t-v) =2t$, which is less than the distance of the code, giving a contradiction which therefore implies that $l$ could not have been a logical operator. 
An analogous arguments holds for $Z$-type logical operators, therefore the sufficient $t$-FTEC condition is satisfied.

This proof can be easily extended to show that any distance $d=(2t+1)$ self-dual CSS code with at most weight-$2 v$ stabilizer generators for some integer $v$ satisfies the flag $t'$-FTEC condition using any $(v-1)$-flag circuits, where $t'= t/\lfloor v/2\rfloor$. 

\textbf{Quantum Reed-Muller codes flag $1$-FTEC:}

The \codepar{n=2^m-1,k=1,d=3} quantum Reed-Muller code family for every integer $m\geq 3$ satisfies the flag 1-FTEC condition using any 1-flag circuits for the standard choice of generators.

We use the following facts about the Quantum Reed-Muller code family (see \cref{app:QRMcodes} and \cite{ADP14} for proofs of these facts): (1) The code is CSS, allowing us to restrict to pure $X$ type and pure $Z$ type logical operators, (2) all pure $X$ or $Z$ type logical operators have odd support, (3) every $X$-type stabilizer generator has the same support as some $Z$-type stabilizer generator, and (4) every $Z$-type stabilizer generator is contained within the support of an $X$ type generator. 

We only need to prove the sufficient condition for $v=0,1$ in this case. For $v=0$, no two qubits can support a logical operator, as any logical operator has weight at least three.
For $v=1$, assume the support of an $X$-type stabilizer generator contains a logical operator $l$. 
That logical operator $l$ cannot be $Z$ type or it would anti-commute with the $X$-stabilizer due to its odd support. 
However, by fact (3), there is a $Z$ type stabilizer with the same support as the $X$ type stabilizer, therefore implying $l$ cannot be $X$ type either.
Therefore, by contradiction we conclude that no logical operator can be contained in the support of an $X$ stabilizer generator. 
Since every other stabilizer generator is contained within the support of an $X$-type stabilizer generator, a logical operator cannot be contained in the support of any stabilizer generator.

Note that the Hamming code family has a stabilizer group which is a proper subgroup of that of the quantum Reed-Muller codes described here. The $X$-type generators of each Hamming code are the same as for a quantum Reed-Muller code, and the Hamming codes are self-dual CSS codes. It is clear that the sufficient condition cannot be applied to the Hamming code since it has even-weight $Z$-type logical operators (which are stabilizers for the quantum Reed-Muller code) supported within the support of some stabilizer generators. 

\textbf{Codes which satisfy flag $t$-FTEC condition but not the sufficient flag $t$-FTEC condition:}

\begin{figure}
	\centering
	\begin{subfigure}{0.3\textwidth}
		\includegraphics[width=\textwidth]{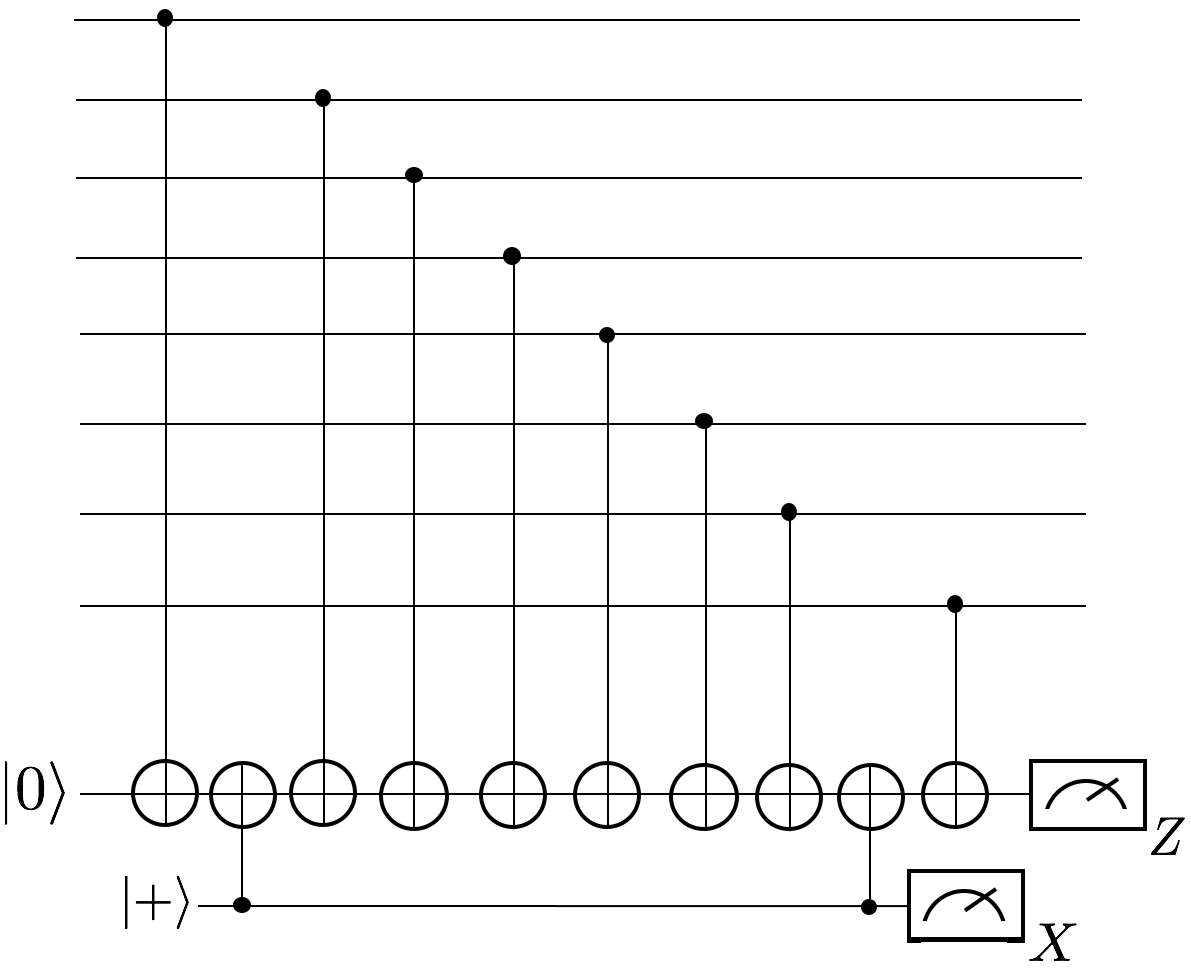}
		\caption{}
		\label{fig:CircuitReichardt1}
	\end{subfigure}
	\begin{subfigure}{0.3\textwidth}
		\includegraphics[width=\textwidth]{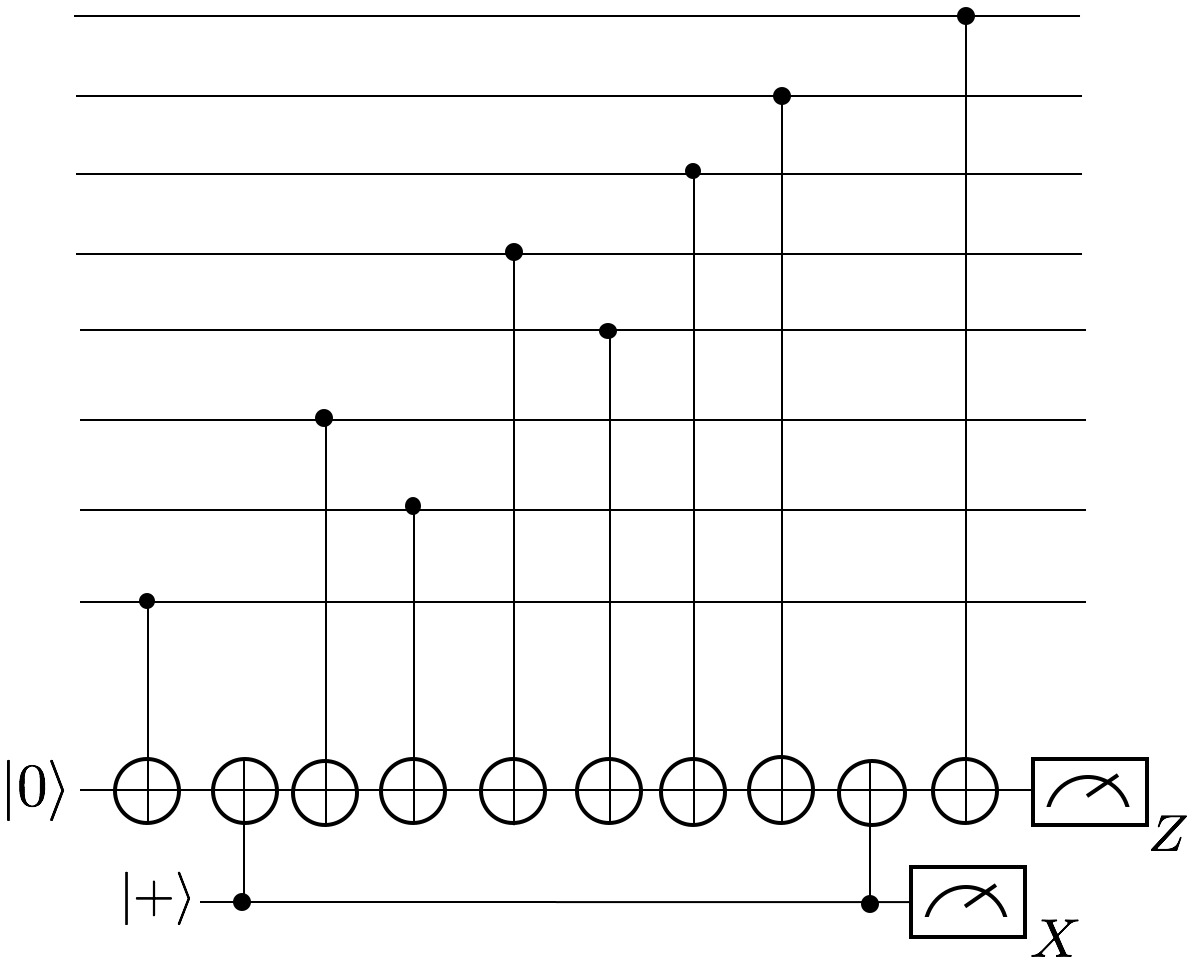}
		\caption{}
		\label{fig:CircuitReichardt2}
	\end{subfigure}
	\caption{(a) A 1-flag circuit for measuring the stabilizer $Z_{8}Z_{9}Z_{10}Z_{11}Z_{12}Z_{13}Z_{14}Z_{15}$ of the \codepar{15,7,3} Hamming code. However a single fault on the fourth or fifth CNOT can lead to the error $Z_{12}Z_{13}Z_{14}Z_{15}$ on the data which is a logical fault. With the CNOT gates permuted as shown in (b), the \codepar{15,7,3} satisfies the general flag 1-FTEC condition.}
\end{figure}

Note that there are codes which satisfy the general flag $t$-FTEC condition but not the sufficient condition presented in this section. An example of such a code is the \codepar{5,1,3} code (see \cref{tab:StabilizerGeneratorsLists} for the codes stabilizer generators and logical operators). Another example includes the Hamming codes as was explained in the discussion on quantum Reed-Muller codes. For instance, consider the \codepar{15,7,3} Hamming code. Using the 1-flag circuit shown in \cref{fig:CircuitReichardt1}, the \codepar{15,7,3} will not satisfy the general flag 1-FTEC condition since a single fault can lead to a logical error on the data. As was shown in \cite{CR17v1}, by permuting the CNOT gates resulting in the circuit illustrated in \cref{fig:CircuitReichardt2}, the flag 1-FTEC condition is satisfied.

\subsection{Circuits}
\label{app:GeneralTflaggedCircuitConstruction}

\begin{figure}
	\centering
	\begin{subfigure}{0.35\textwidth}
		\includegraphics[width=\textwidth]{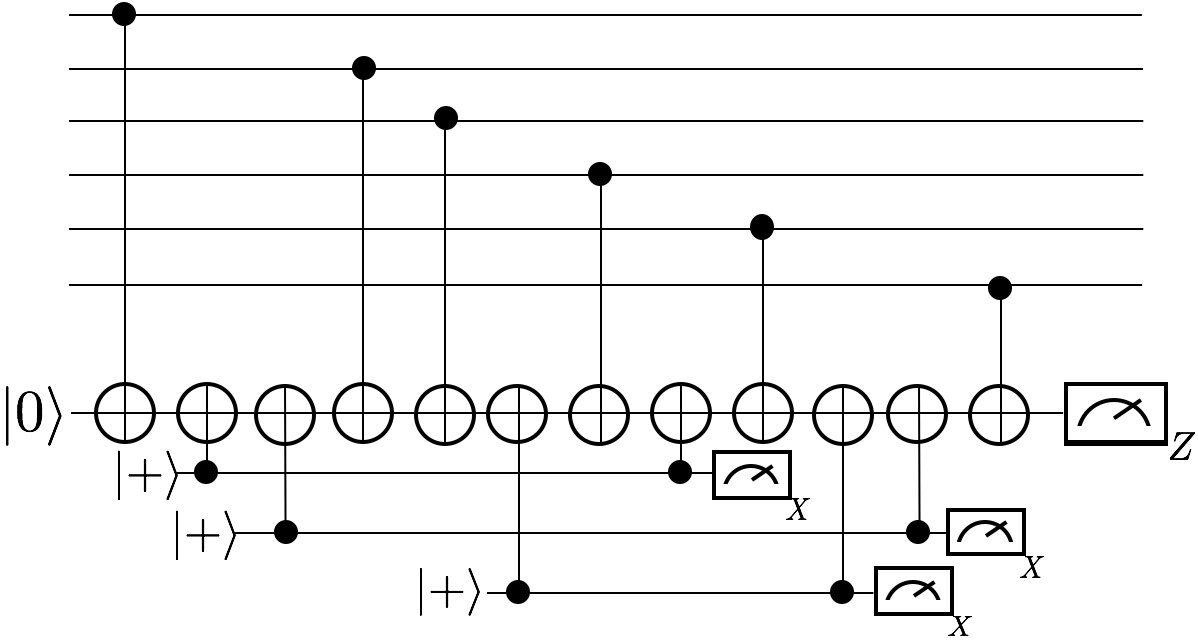}
		\caption{}
		\label{fig:6flagCircuit}
	\end{subfigure}
	\begin{subfigure}{0.35\textwidth}
		\includegraphics[width=\textwidth]{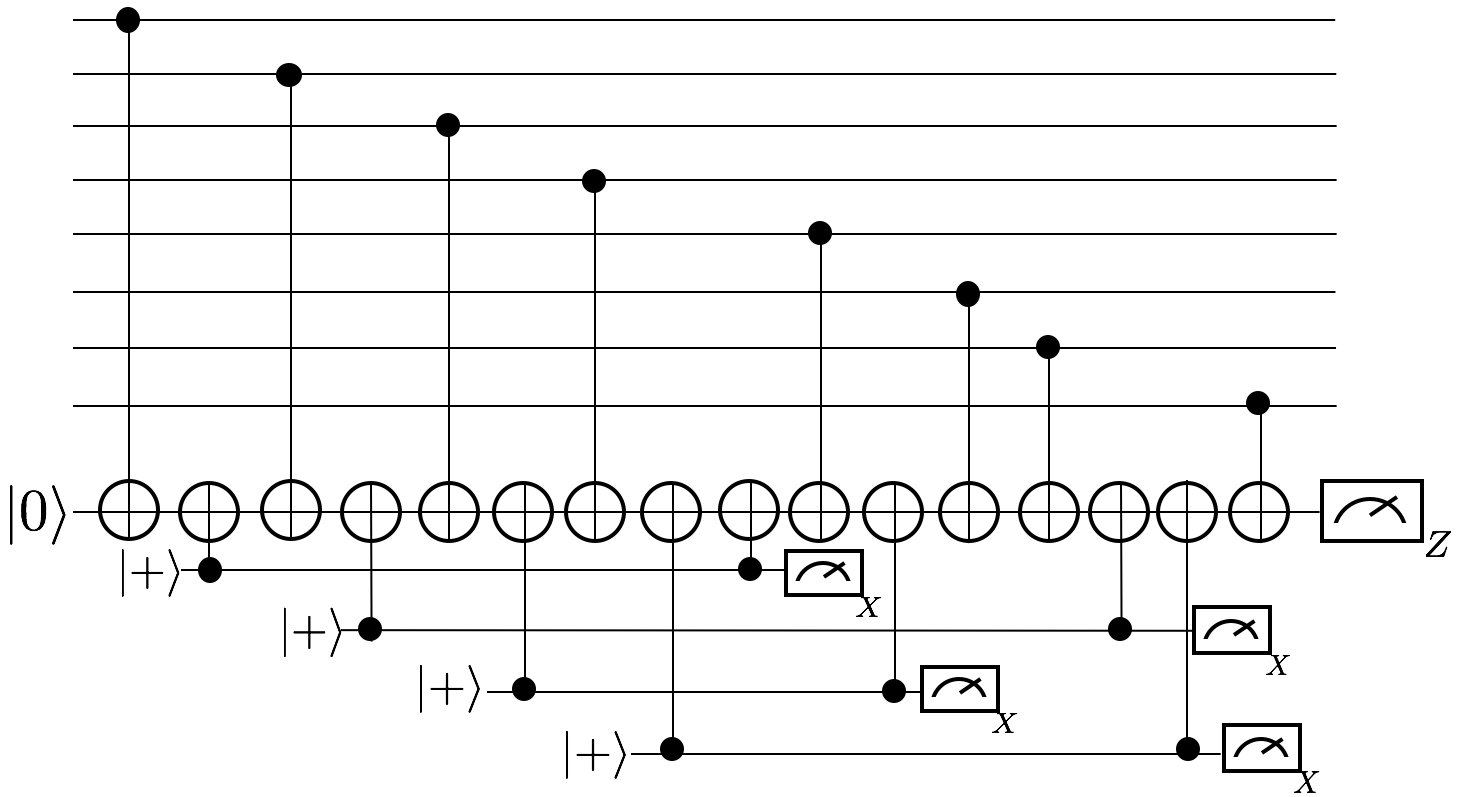}
		\caption{}
		\label{fig:3flagWeight8}
	\end{subfigure}
	\caption{(a) Illustration of a w-flag circuit for measuring the operator $Z^{\otimes w}$ where $w=6$ using the smallest number of flag qubits. (b) Illustration of a 3-flag circuit for measuring $Z^{\otimes 8}$ using the smallest number of flag qubits.}
	\label{fig:ExamplesOfLargeFlagCircuits}
\end{figure}

In \cref{subsec:Remarks} we showed that the family of surface codes, color codes with a hexagonal lattice and quantum Reed-Muller codes satisfied a sufficient condition allowing them to be used in the flag $t$-FTEC protocol. Along with the general 1-flag circuit construction of \cref{fig:General1FlagCircuitSecCircuit}, the $6$-flag circuit for measuring $Z^{\otimes 6}$ of \cref{fig:6flagCircuit} can be used as $t$-flag circuits for all of the codes in \cref{subsec:Remarks}. Note that the circuit in \cref{fig:StabFTwithAncilla} (which is a special case of \cref{fig:General1FlagCircuitSecCircuit} when $w=4$) is a 4-flag circuit which is used for measuring $Z^{\otimes 4}$. 

Before describing general 1- and 2-flag circuit constructions, we give the following two definitions which we will frequently use: Any CNOT that couples a data qubit to the measurement qubit will be referred to as $\text{CNOT}_{dm}$ and any CNOT coupling a measurement qubit to a flag qubit will be referred to as $\text{CNOT}_{fm}$. In both cases the target qubit will always be the measurement qubit. 

\textbf{1- and 2-flag circuits for weight $w$ stabilizer measurements:}

We provide 1- and 2-flag circuit constructions for measuring a weight-$w$ stabilizer.
The 1-flag circuit requires a single flag qubit, and the 2-flag circuit requires at most four flag qubits.

\begin{figure}
	\centering
	\begin{subfigure}{0.35\textwidth}
		\includegraphics[width=\textwidth]{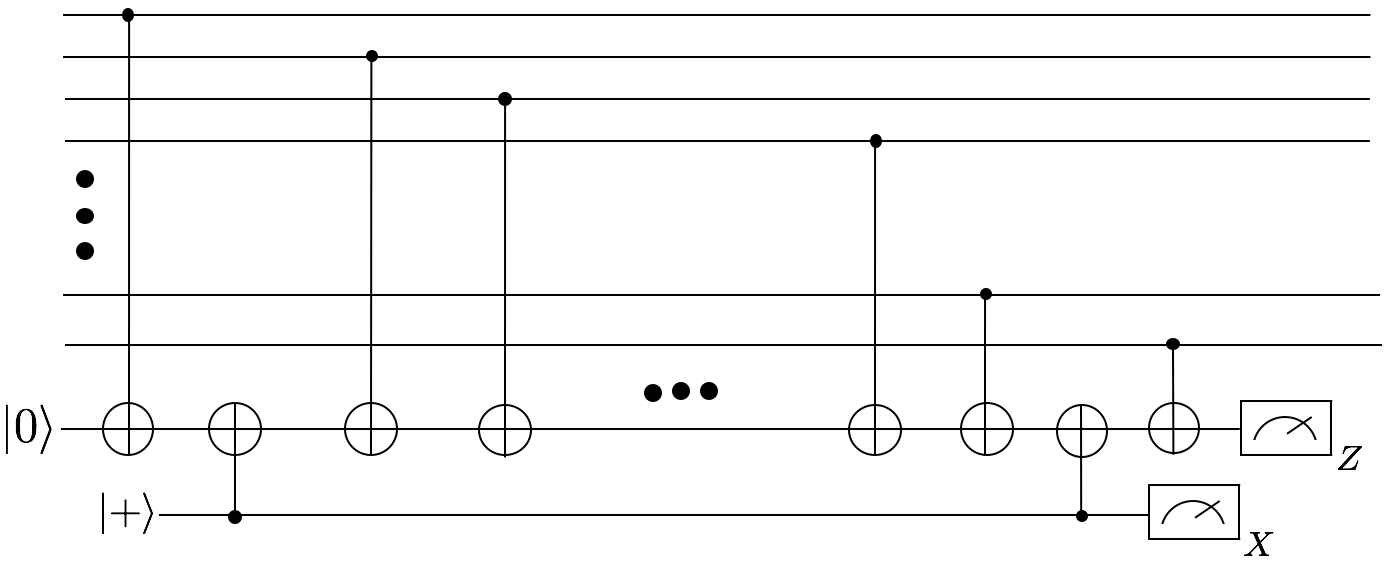}
		\caption{}
		\label{fig:General1FlagCircuitSecCircuit}
	\end{subfigure}
	\begin{subfigure}{0.43\textwidth}
		\includegraphics[width=\textwidth]{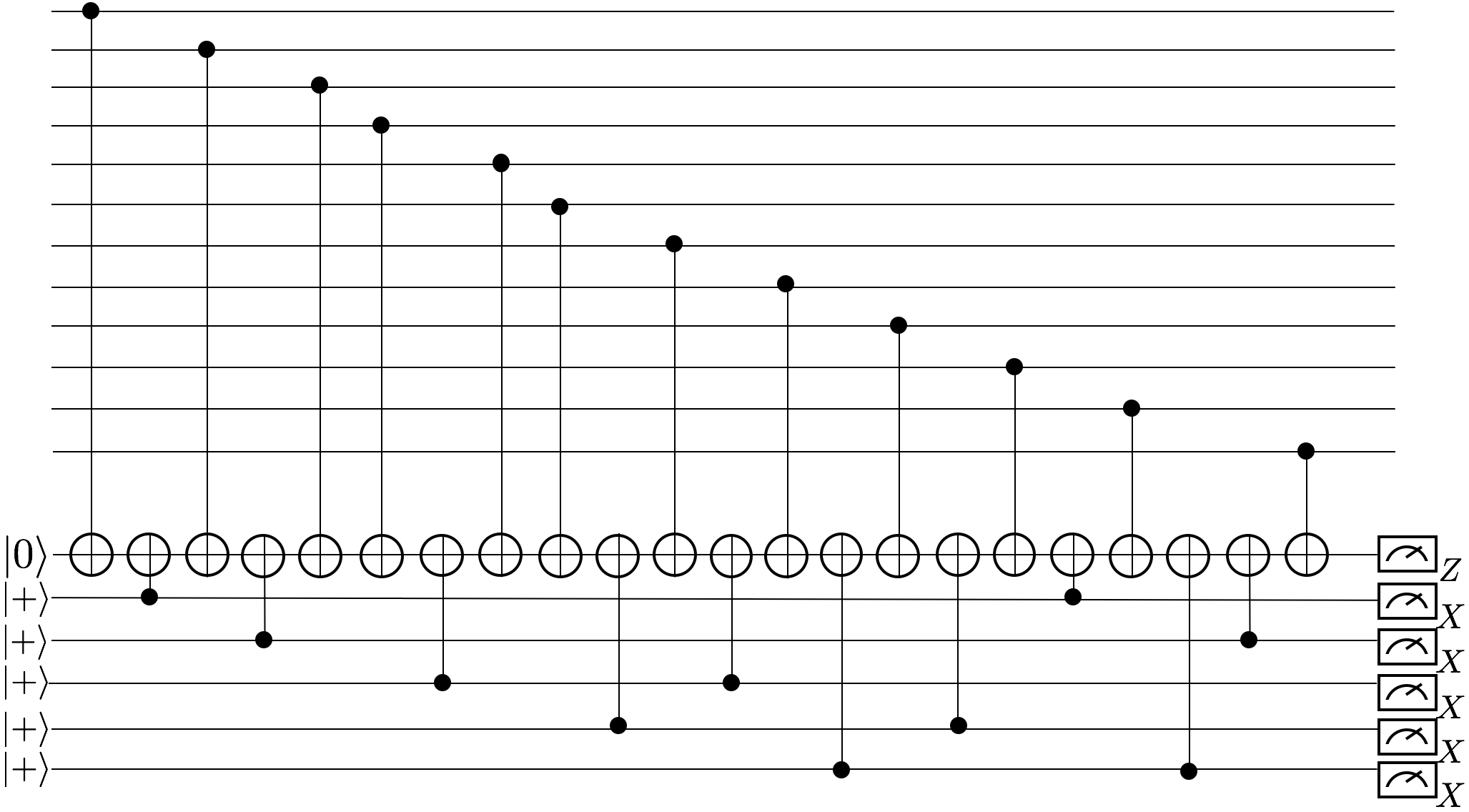}
		\caption{}
		\label{fig:General2FlagCircuitSecCircuit}
	\end{subfigure}
	\begin{subfigure}{0.43\textwidth}
		\includegraphics[width=\textwidth]{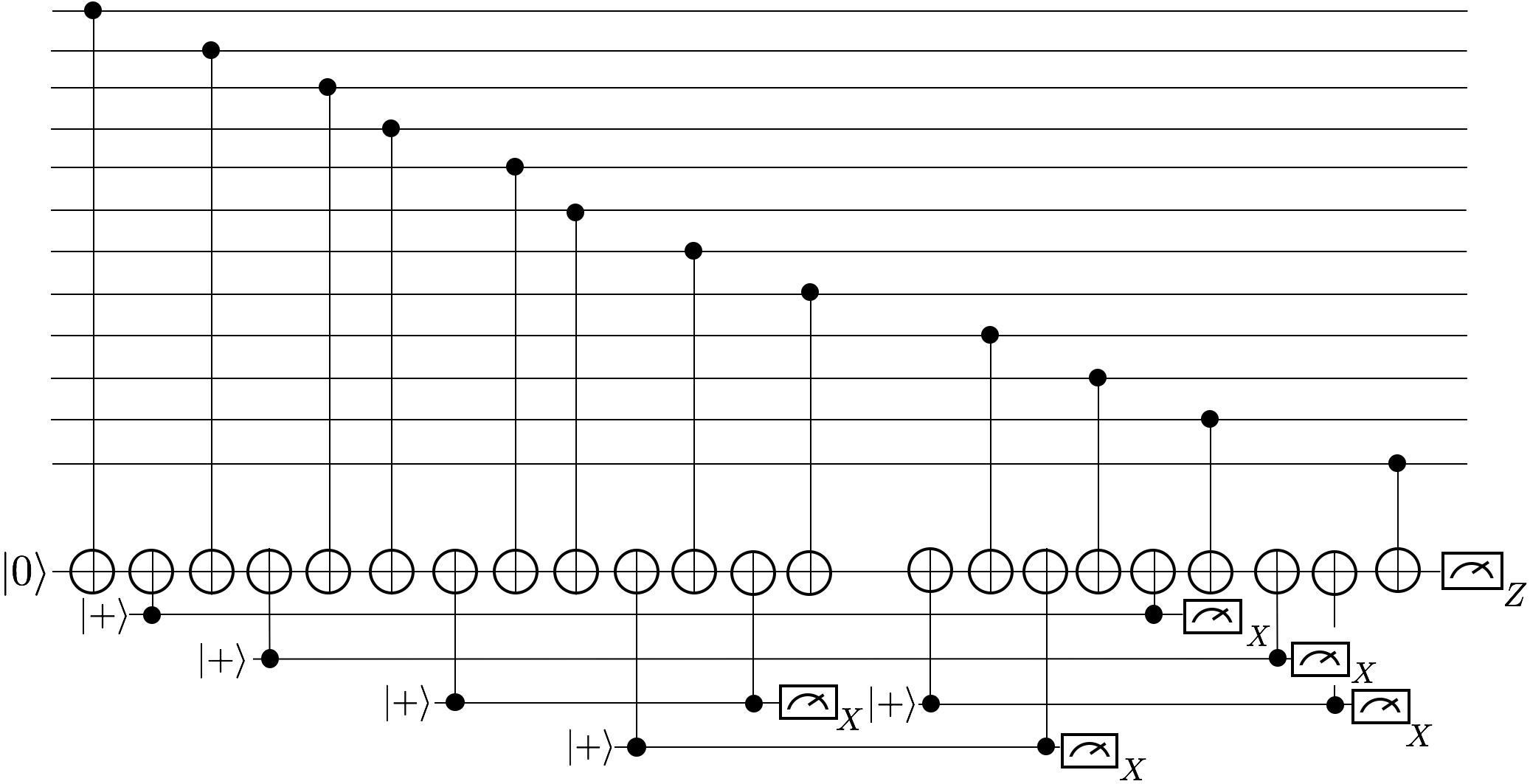}
		\caption{}
		\label{fig:General2FlagCircuitSecCircuitCompressed}
	\end{subfigure}
	\caption{(a) General 1-flag circuit for measuring the stabilizer $Z^{\otimes w}$. (b) Example of a 2-flag circuit for measuring $Z^{\otimes 12}$ using our general 2-flag circuit construction. (c) An equivalent circuit using fewer flag qubits by reusing a measured flag qubit and reinitializing it in the $\ket{+}$ state for use in another pair of  $\text{CNOT}_{\text{fm}}$ gates.}
	\label{fig:Generall1And2FlagCircuits}
\end{figure}

Without loss of generality, in proving that the circuit constructions described below are 1- and 2-flag circuits, we can assume that all faults occurred on CNOT gates. This is because any set of $v$ faults (including those at idle, preparation or measurement locations) will have the same output Pauli operator and flag measurement results as some set of at most $v$ faults on CNOT gates (since every qubit is involved in at least one CNOT).

As was shown in Ref.~\cite{CR17v1}, \cref{fig:General1FlagCircuitSecCircuit} illustrates a general 1-flag circuit construction for measuring the stabilizer $Z^{\otimes w}$ which requires only two $\text{CNOT}_{\text{fm}}$ gates. To see that the first construction is a 1-flag circuit, note that an $IZ$ error occurring on any CNOT will give rise to a flag unless it occurs on the first or last $\text{CNOT}_{\text{dm}}$ gates or the last $\text{CNOT}_{\text{fm}}$ gate. However, such a fault on any of these three gates can give rise to an error of weight at most one (after multiplying by the stabilizer $Z^{\otimes w}$). One can also verify that if there are no faults, the circuit in \cref{fig:General1FlagCircuitSecCircuit} implements a projective measurement of $Z^{\otimes w}$ without flagging. Following the approach in \cite{LAR11}, one simply needs to check that the circuit preserves the stabilizer group generated by $Z^{\otimes w}$ and $X$ on each ancilla prepared in the $\ket{+}$ state and $Z$ on each ancilla prepared in the $\ket{0}$ state. By using pairs of $\text{CNOT}_{\text{fm}}$ gates, this construction satisfies the requirement. 

We now give a general 2-flag circuit construction for measuring $Z^{\otimes w}$ for arbitrary $w$ (see \cref{fig:General2FlagCircuitSecCircuit} for an example). The circuit consists of pairs of $\text{CNOT}_{\text{fm}}$ gates each connected to a different flag qubit prepared in the $\ket{+}$ state and measured in the $X$ basis. The general 2-flag circuit construction involves the following placement of $w/2-1$ pairs of  $\text{CNOT}_{\text{fm}}$ gates: 

\begin{enumerate}
\item Place a $\text{CNOT}_{\text{fm}}$ pair between the first and second last $\text{CNOT}_{\text{dm}}$ gates.
\item Place a $\text{CNOT}_{\text{fm}}$ pair between the second and last $\text{CNOT}_{\text{dm}}$ gates.
\item After the second $\text{CNOT}_{\text{fm}}$ gate, place the first $\text{CNOT}_{\text{fm}}$ gate of the remaining  pairs after every two $\text{CNOT}_{\text{dm}}$ gates. The second $\text{CNOT}_{\text{fm}}$ gate of a pair is placed after every three $\text{CNOT}_{\text{dm}}$ gates.
\end{enumerate}
As shown in \cref{fig:General2FlagCircuitSecCircuitCompressed}, it is possible to reuse some flag qubits to measure multiple pairs of $\text{CNOT}_{\text{fm}}$ gates at the cost of introducing extra time steps into the circuit. For this reason, at most four flag qubits will be needed, however, if $w \le 8$, then $w/2-1$ flag qubits are sufficient. 

We now show that the above construction satisfies the requirements of a 2-flag circuit. If one CNOT gate fails, by an argument analogous to that used for the 1-flag circuit, there will be a flag or an error of at most weight-one on the data. If the first pair of $\text{CNOT}_{\text{fm}}$ gates fail causing no flag qubits to flag, after multiplying the data qubits by $Z^{\otimes w}$, the resulting error $E_{r}$ will have $\text{wt}(E_{r}) \le 2$. For any other pair of $\text{CNOT}_{\text{fm}}$ gates that fail causing an error of weight greater than two on the data, by construction there will always be another $\text{CNOT}_{\text{fm}}$ gate between the two that fail which will propagate a $Z$ error to a flag qubit causing it to flag. Similarly, if pairs of $\text{CNOT}_{\text{dm}}$ gates fail resulting in the data error $E_{r}$ with $\text{wt}(E_{r}) \ge 2$, by construction there will always be an odd number of $Z$ errors propagating to a flag qubit due to the  $\text{CNOT}_{\text{fm}}$ gates in between the $\text{CNOT}_{\text{dm}}$ gates that failed causing a flag qubit to flag. The same argument applies if a failure occurs between a $\text{CNOT}_{\text{dm}}$ and $\text{CNOT}_{\text{fm}}$ gate. 

Lastly, a proposed general $w$-flag circuit construction for arbitrary $w$ is provided in \cref{App:GeneralwFlagCircuitConstruction}. 

\textbf{Use of flag information:}
  
As seen in \cref{fig:6flagCircuit,fig:3flagWeight8,fig:General2FlagCircuitSecCircuit,fig:General2FlagCircuitSecCircuitCompressed}, in general $t$-flag circuits require more than one flag qubit. Apart from their use in satisfying the $t$-flag circuit properties, the extra flag qubits could be used to reduce the size of the flag error sets (defined in \cref{Def:FlagErrSetDef}) when verifying the Flag $t$-FTEC condition of \cref{app:GeneralFTEC}. To do so, we first define $f$, where $f$ is a bit string of length $u$ (here $u$ is the number of flag qubits) with $f_{i} = 1$ if the i'th flag qubit flagged and 0 otherwise. In this case, the correction set of \cref{eq:GeneralLookupTableV2} can be modified to include flag information as follows: 

\begin{align}
&\tilde{E}_{t}^{m}(g_{i_{1}},\cdots , g_{i_{k}},s,f_{i_{1}},\cdots ,f_{i_{k}}) = \nonumber \\
&\begin{cases}
\{ E \in \mathcal{E}_{m}(g_{i_{1}},\cdots , g_{i_{k}},f_{i_{1}},\cdots ,f_{i_{k}}) \times \mathcal{E}_{t-m}  \\
 \text{ such that } s(E) = s \} \\
\{ E_{\text{min}}(s) \} \text{ if above set empty., }
\end{cases}
\label{eq:CorrectionSetWithFlagInfo}
\end{align}
where $\mathcal{E}_{m}(g_{i_{1}},\cdots , g_{i_{k}},f_{i_{1}},\cdots,f_{i_{k}})$ is the new flag error set containing only errors caused by precisely $m$ faults spread amongst the circuits $C(g_{i_{1}}),C(g_{i_{2}}), \cdots , C(g_{i_{k}})$ which each gave rise to the flag outcomes $f_{i_{1}},\cdots,f_{i_{k}}$. 

Hence only errors which result from the measured flag outcome would be stored in the correction set. With enough flag qubits, this could potentially broaden the family of codes which satisfy the Flag $t$-FTEC condition.

\section{Circuit level noise analysis}
\label{sec:CircuitLevelNoiseFTEC}

The purpose of this section is to demonstrate explicitly the flag 2-FTEC protocol, and to identify parameter regimes in which flag FTEC presented both here and in other works offers advantages over other existing FTEC schemes.
In \cref{subsec:Numerics19} we analyze the logical failure rates of the \codepar{19,1,5} color code and compute it's pseudo-threshold for the three choices of $\tilde{p}$. In \cref{subsec:CompareFlagecschemes} we compare logical failure rates of several fault-tolerant error correction schemes applied to distance-three and distance-five stabilizer codes. The stabilizers for all of the studied codes are given in \cref{tab:StabilizerGeneratorsLists}. Logical failure rates are computed using the full circuit level noise model and simulation methods described in \cref{subsec:NoiseAndNumerics}.

\subsection{Numerical analysis of the \codepar{19,1,5} color code}
\label{subsec:Numerics19}

The full circuit-level noise analysis of the flag 2-FTEC protocol applied to the \codepar{19,1,5} color code was performed using the stabilizer measurement circuits of \cref{fig:StabFTwithAncilla,fig:WeightSixGenerators}. 

In the weight-six stabilizer measurement circuit of \cref{fig:WeightSixGenerators}, there are 10 CNOT gates, three measurement and state-preparation locations, and 230 resting qubit locations. When measuring all stabilizer generators using non-flag circuits, there are 42 CNOT and 42 XNOT gates, 18 measurement and state-preparation locations, and 2196 resting qubit locations. Consequently, we expect the error suppression capabilities of the flag EC scheme to depend strongly on the number of idle qubit locations. 

\begin{table}[t]
\begin{tabular}{ c|c}
 three-qubit flag EC & pseudo-threshold  \\ \hline
\codepar{19,1,5} and $\tilde{p}=p$    & $p_{\mathrm{pseudo}} = (1.14 \pm 0.02) \times 10^{-5}$ \\
\codepar{19,1,5} and $\tilde{p}=\frac{p}{10}$    & $p_{\mathrm{pseudo}} = (6.70 \pm 0.07) \times 10^{-5}$  \\
\codepar{19,1,5} and $\tilde{p}=\frac{p}{100}$    & $p_{\mathrm{pseudo}} = (7.74 \pm 0.16) \times 10^{-5}$  \end{tabular}
\caption{Table containing pseudo-threshold values for the flag 2-FTEC protocol applied to the \codepar{19,1,5} color code for $\tilde{p}=p$, $\tilde{p}=p/10$ and $\tilde{p}=p/100$.}
\label{tab:PseudoThresholAllThree1915}
\end{table}

\begin{figure}
\centering
\includegraphics[width=0.5\textwidth]{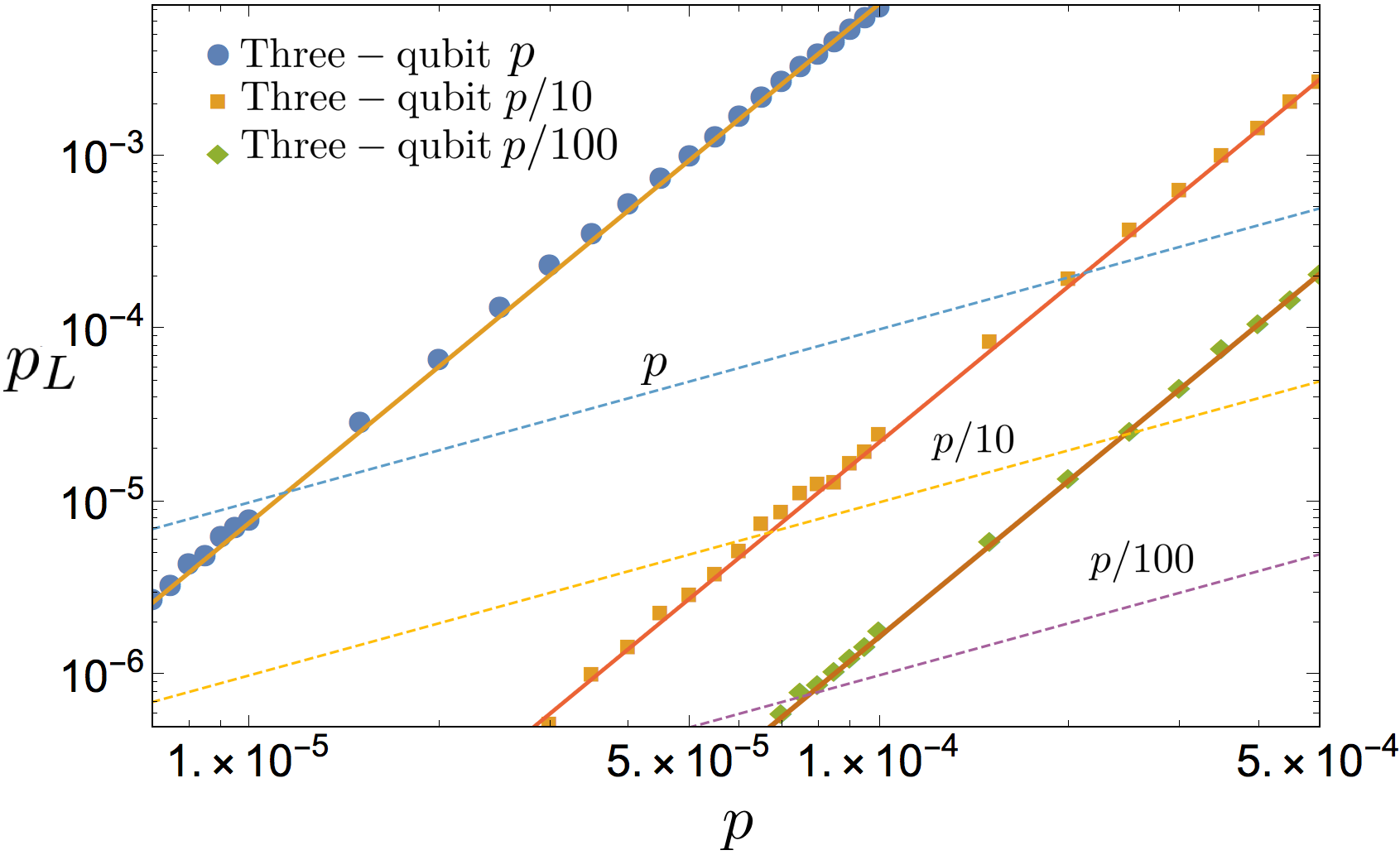}
\caption{Logical failure rates of the \codepar{19,1,5} color code after implementing the flag 2-FTEC protocol presented in \cref{subsec:Distance5protocol} for the three noise models described in \cref{subsec:NoiseAndNumerics}. The dashed curves represent the lines $\tilde{p}=p$, $\tilde{p}=p/10$ and $\tilde{p}=p/100$. The crossing point between $\tilde{p}$ and the curve corresponding to $p^{(\codepar{19,1,5})}_{L}(\tilde{p})$ in \cref{Def:PseudoThreshDef} gives the pseudo-threshold.}
\label{fig:PseudoThreshPlots19ColorAllThree}
\end{figure}

Pseudo-thresholds of the \codepar{19,1,5} code were obtained using the methods of \cref{subsec:NoiseAndNumerics}. Recall that for extending the lifetime of a qubit (when idle qubit locations fail with probability $\tilde{p}$), the probability of failure after implementing an FTEC protocol should be smaller than $\tilde{p}$. We calculated the pseudo-threshold using \cref{Def:PseudoThreshDef} for the three cases were idle qubits failed with probability $\tilde{p}=p$, $\tilde{p}=p/10$ and $\tilde{p}=p/100$. The results are shown in \cref{tab:PseudoThresholAllThree1915}.

The logical failure rates for the three noise models are shown in \cref{fig:PseudoThreshPlots19ColorAllThree}. It can be seen that when the probability of error on a resting qubit decreases from $p$ to $p/10$, the pseudo-threshold improves by nearly a factor of six showing the strong dependence of the scheme on the probability of failure of idle qubits. 

\subsection{Comparison of flag 1- and 2-FTEC with other FTEC schemes}
\label{subsec:CompareFlagecschemes}

\begin{figure*}[t!]
\begin{subfigure}{0.33\textwidth}
\begin{center}
\includegraphics[width=\textwidth]{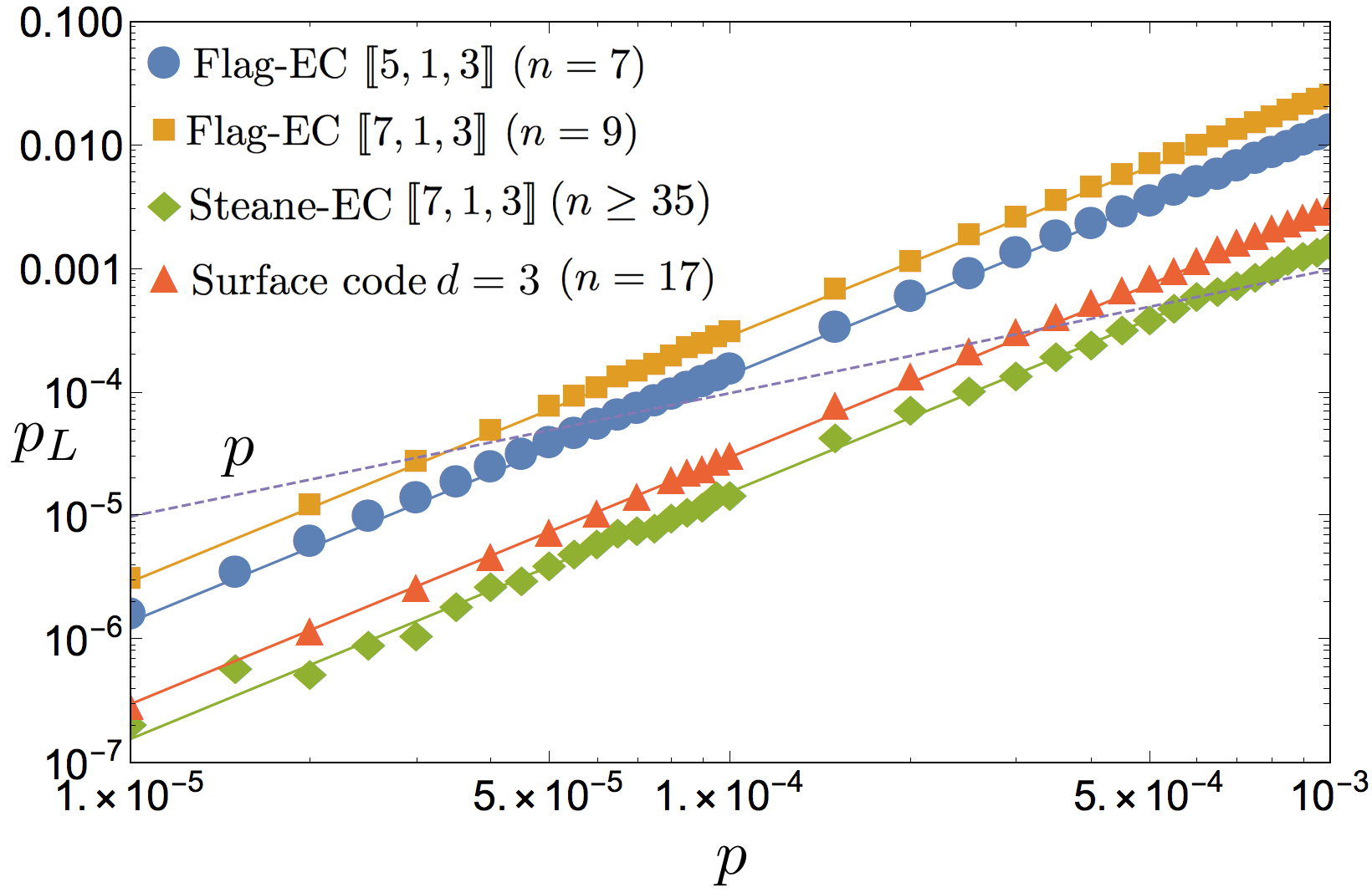}
\caption{}
\includegraphics[width=\textwidth]{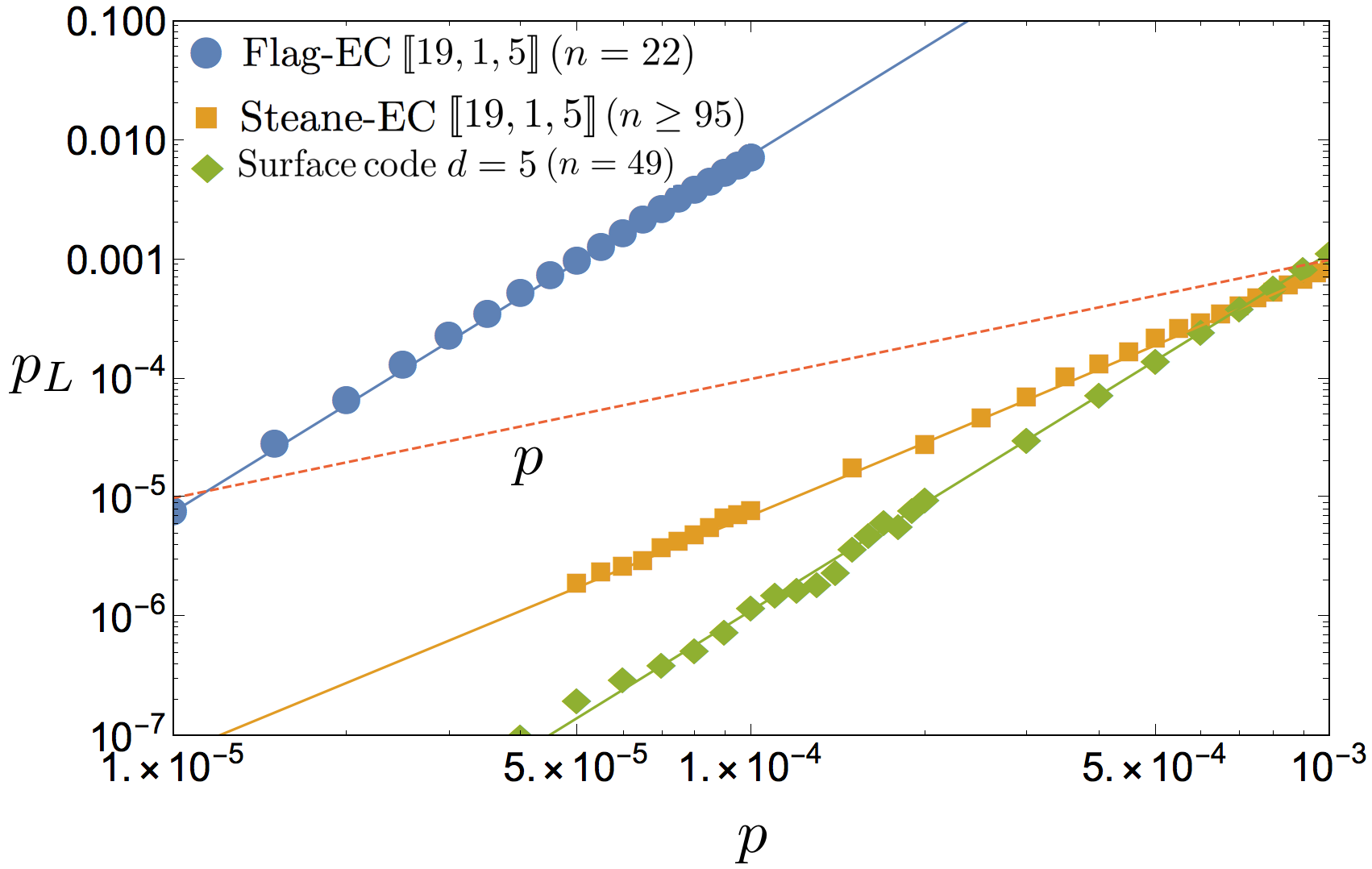}
\caption{}
\end{center}
\end{subfigure}\hfill
\begin{subfigure}{0.33\textwidth}
\includegraphics[width=\textwidth]{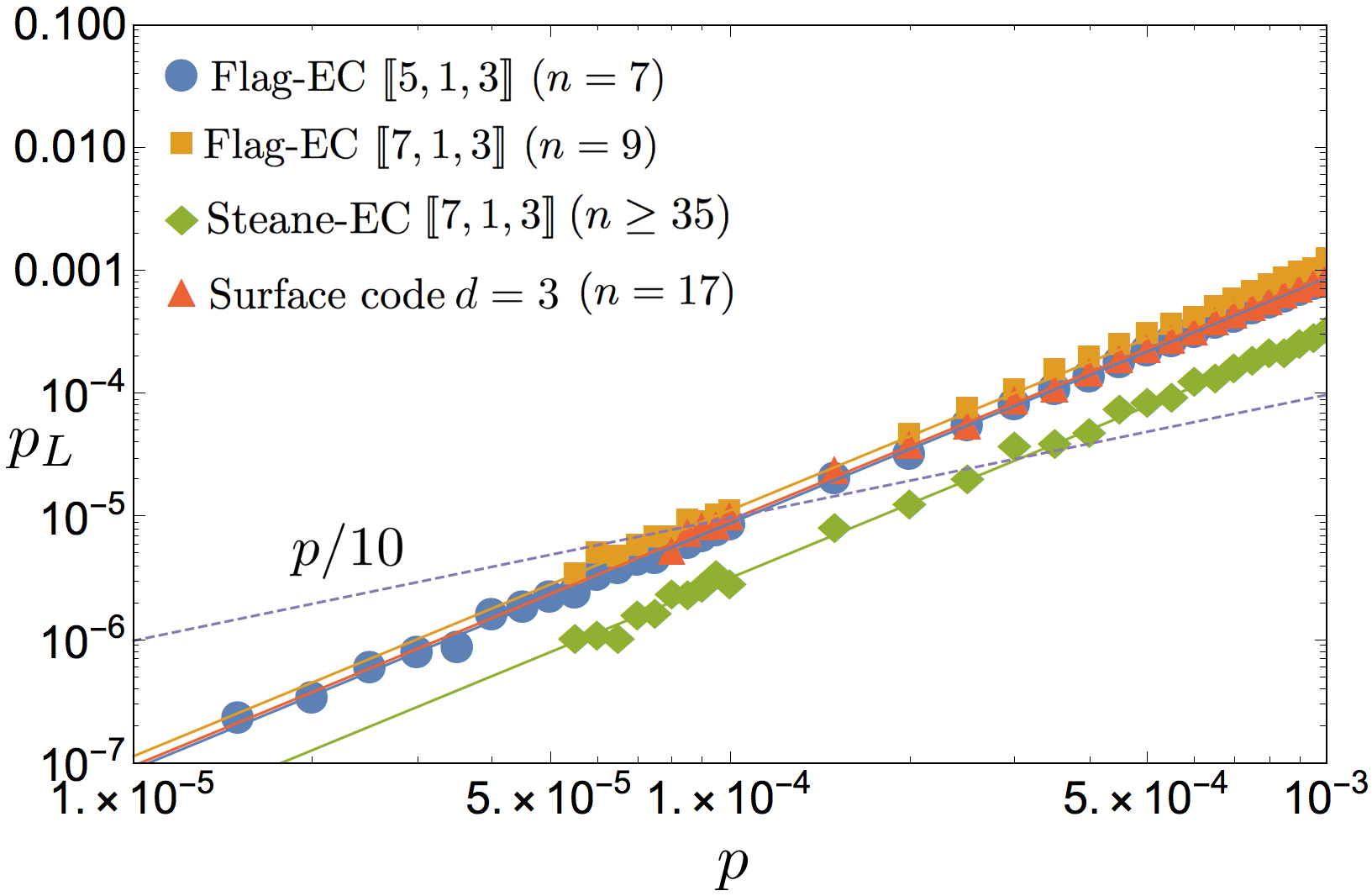}
\caption{}
\includegraphics[width=\textwidth]{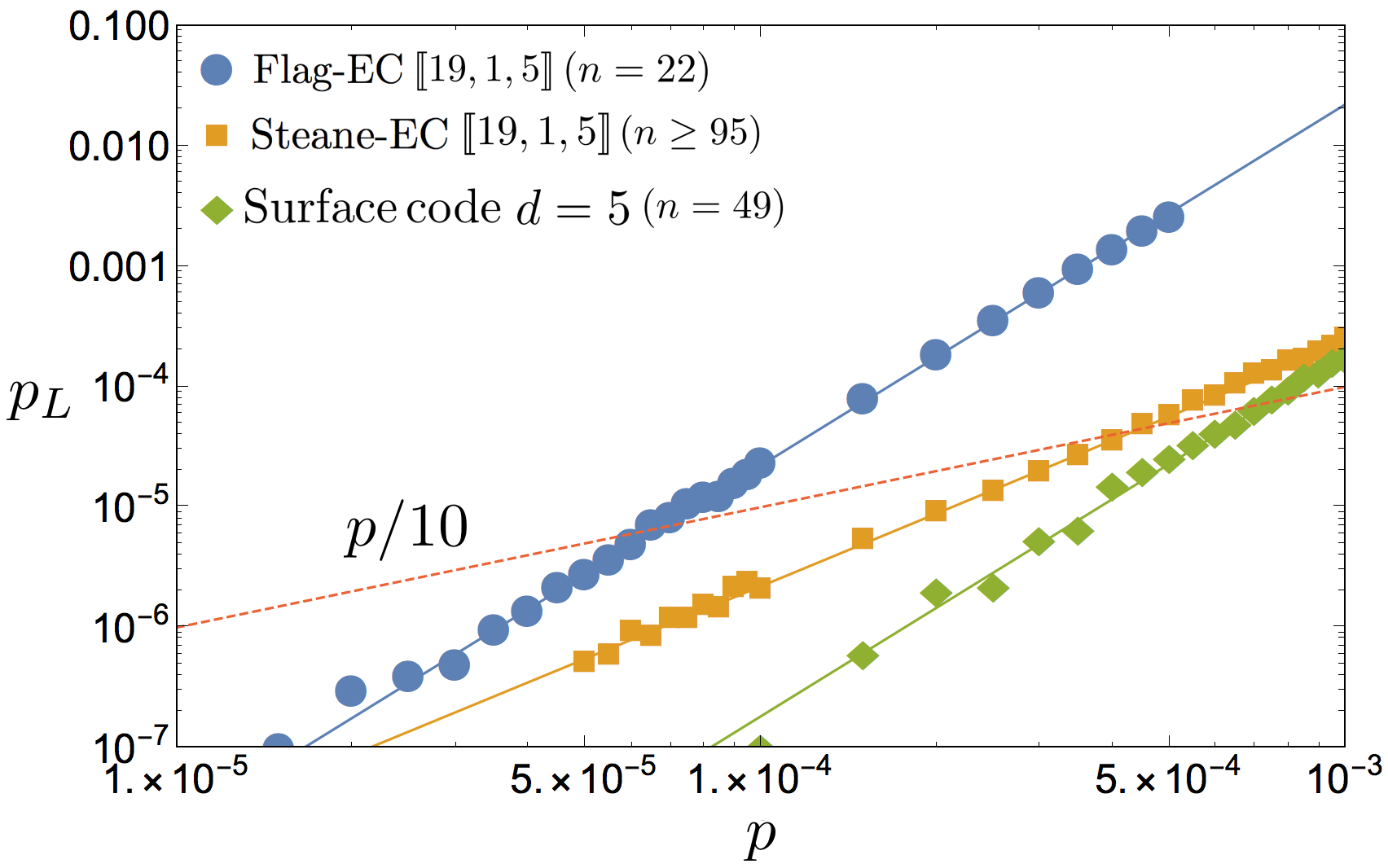}
\caption{}
\end{subfigure}\hfill
\begin{subfigure}{0.33\textwidth}
\includegraphics[width=\textwidth]{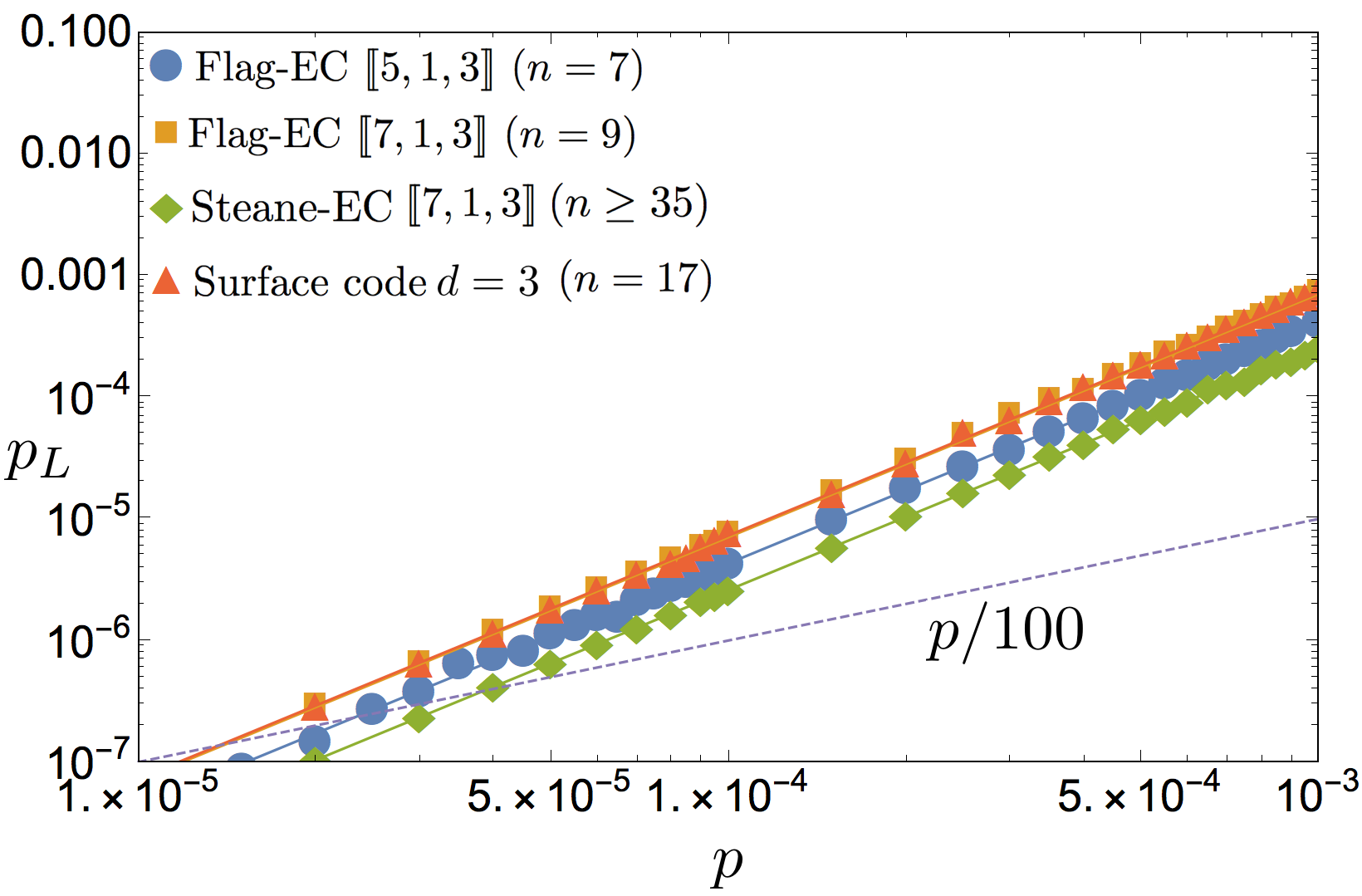}
\caption{}
\includegraphics[width=\textwidth]{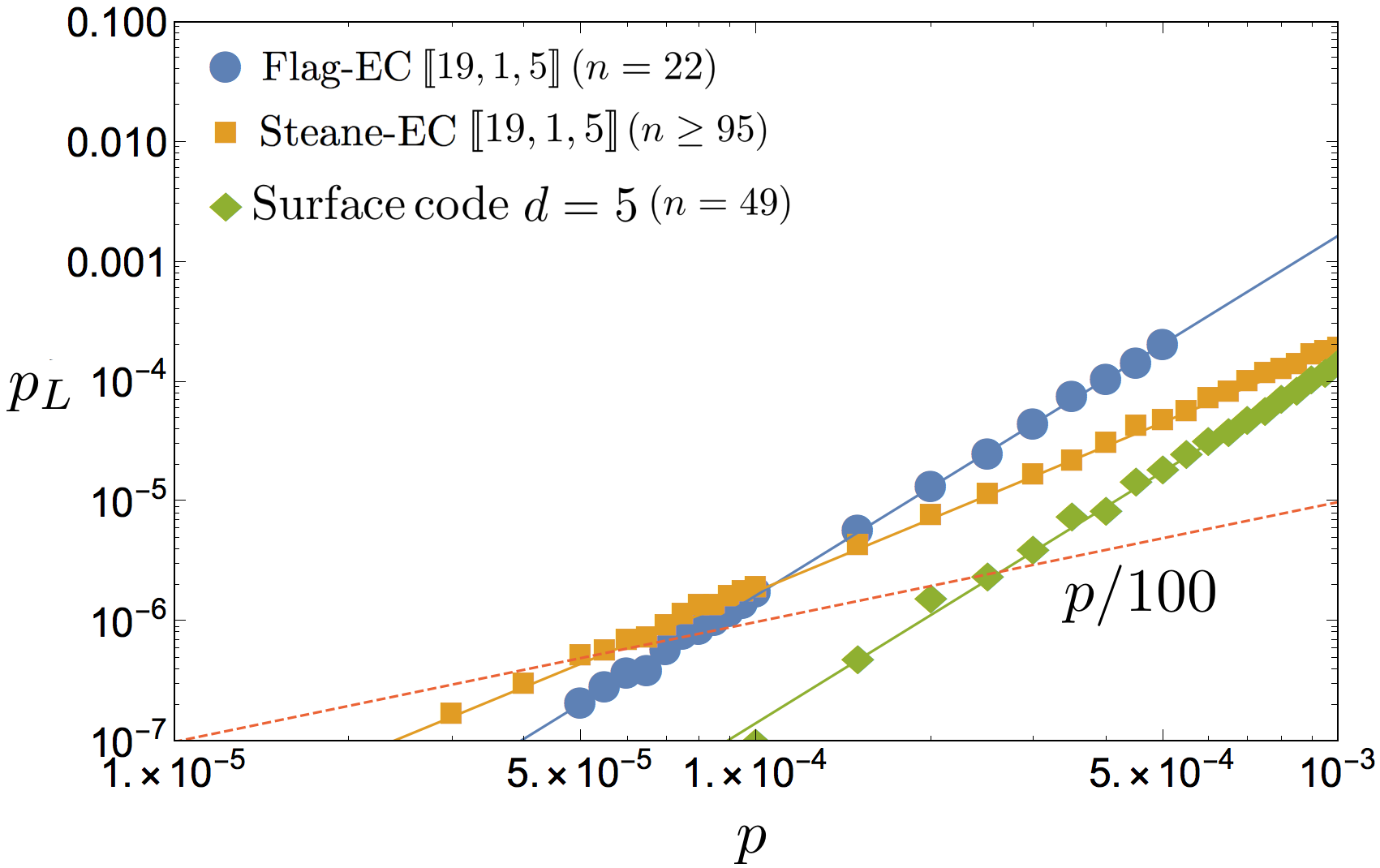}
\caption{}
\end{subfigure}\hfill
\caption{Logical failure rates for various fault-tolerant error correction methods applied to the \codepar{5,1,3} code, \codepar{7,1,3} Steane code and the \codepar{19,1,5} color code. The dashed curves correspond to the lines $\tilde{p}=p$, $\tilde{p}=p/10$ and $\tilde{p}=p/100$. In (a), (c) and (e), the flag 1-FTEC protocol is applied to the \codepar{5,1,3} and Steane code and the results are compared with the $d=3$ surface code and Steane error correction applied to the Steane code. In (b), (d) and (f), the flag 2-FTEC protocol is applied to the \codepar{19,1,5} color code, and the results are compared with the $d=5$ surface code and Steane error correction applied to the \codepar{19,1,5} color code. 
These numerical results suggest the following fault-tolerant experiments of the schemes we consider for extending the fidelity of a qubit. (1) If $7 \le n \le16$, only the 5 and 7 qubit codes with flag 1-FTEC are accessible. However, the performance is much worse than higher qubit alternatives unless $\tilde{p}/p$ is small. (2) For $17\le n \le34$, the $d=3$ surface code seems most promising, unless $\tilde{p}/p$ is small, in which case flag 2-FTEC with the 19-qubit code should be better. (3) For $35\le n \le48$, Steane EC applied to distance-three codes is better than all other approaches studied, except for very low $p$ where flag 2-FTEC should be better due to ability to correct two rather than just one fault. (4) For $n \ge 49$, the d=5 surface code is expected to perform better than the other alternatives below pseudo-threshold.}
\label{fig:AllComparisonPlotsCombined}
\end{figure*}

The most promising schemes for testing fault-tolerance in near term quantum devices are those which achieve high pseudo-thresholds while maintaining a low qubit overhead. The flag FTEC protocol presented in this paper uses fewer qubits compared to other well known fault-tolerance schemes but typically has increased circuit depth. In this section we apply the flag FTEC protocol of \cref{subsec:ReviewChaoReichardt,subsec:Distance5protocol} to the \codepar{5,1,3}, \codepar{7,1,3} and \codepar{19,1,5} codes. We compare logical failure rates for three values of $\tilde{p}$ with Steane error correction applied to the \codepar{7,1,3} and \codepar{19,1,5} codes and with the $d=3$ and $d=5$ rotated surface code. More details on Steane error correction and surface codes are provided in \cref{app:SurfaceECSection,app:SteaneECSection}. Note that recent work by Goto has provided optimizations to prepare Steane ancillas \cite{Goto16}. However, our numerical results for Steane-EC were produced using the methods presented in \cref{app:SteaneECSection}.

Results of the logical failure rates for $\tilde{p}=p$, $\tilde{p}=p/10$ and $\tilde{p}=p/100$ are shown in \cref{fig:AllComparisonPlotsCombined}. Various pseudo-thresholds and required time-steps for the considered fault-tolerant error correction methods are given in \cref{tab:PseudoThresholdsAllECschemesD3,tab:PseudoThresholdsAllECschemesD5}. 

The circuits for measuring the stabilizers of the 5-qubit code were similar to the ones used in \cref{fig:StabFTwithAncilla} (for an $X$ Pauli replace the CNOT by an XNOT). For flag-FTEC methods, it can be seen that the \codepar{5,1,3} code always achieves lower logical failure rates compared to the \codepar{7,1,3} code. However, when $\tilde{p}=p$, both the $d=3$ surface code as well as Steane-EC achieves lower logical failure rates (with Steane-EC achieving the best performance). For $\tilde{p}=p/10$, flag-EC applied to the \codepar{5,1,3} code achieves nearly identical logical failure rates compared to the $d=3$ surface code. For $\tilde{p}=p/100$, flag 1-FTEC applied to the \codepar{5,1,3} code achieves lower logical failure rates than the $d=3$ surface code but still has higher logical failure rates compared to Steane-EC.

\begin{table*}[!t]
\begin{tabular}{ c|c|c|c|c}
 FTEC scheme & Noise model & Number of qubits & Time steps ($T_{\mathrm{time}}$) &Pseudo-threshold  \\ \hline
 
Flag-EC \codepar{5,1,3} & $\tilde{p} = p$ & $7$ & $64 \le T_{\mathrm{time}} \le 88$ & $p_{\mathrm{pseudo}} = (7.09 \pm 0.03) \times 10^{-5}$ \\
Flag-EC \codepar{7,1,3} &  & $9$ & $72 \le T_{\mathrm{time}} \le 108$ & $p_{\mathrm{pseudo}} = (3.39 \pm 0.10) \times 10^{-5}$ \\
$d=3$ Surface code  &  & 17 & $\ge 18$ & $p_{\mathrm{pseudo}} = (3.29 \pm 0.16) \times 10^{-4}$ \\
Steane-EC \codepar{7,1,3} & &$\ge 35$ & 15 & $p_{\mathrm{pseudo}} = (6.29 \pm 0.13) \times 10^{-4}$ \\ \hline

Flag-EC \codepar{5,1,3} & $\tilde{p} = p/10$ & $7$ & $64 \le T_{\mathrm{time}} \le 88$ & $p_{\mathrm{pseudo}} = (1.11 \pm 0.02) \times 10^{-4}$ \\
Flag-EC \codepar{7,1,3} &  & $9$ & $72 \le T_{\mathrm{time}} \le 108$ & $p_{\mathrm{pseudo}} = (8.68 \pm 0.15) \times 10^{-5}$ \\
$d=3$ Surface code  &  & 17 & $\ge 18$ & $p_{\mathrm{pseudo}} = (1.04 \pm 0.02) \times 10^{-4}$ \\
Steane-EC \codepar{7,1,3} & &$\ge 35$ & 15 & $p_{\mathrm{pseudo}} = (3.08 \pm 0.01) \times 10^{-4}$ \\ \hline

Flag-EC \codepar{5,1,3} & $\tilde{p} = p/100$ & $7$ & $64 \le T_{\mathrm{time}} \le 88$ & $p_{\mathrm{pseudo}} = (2.32 \pm 0.03) \times 10^{-5}$ \\
Flag-EC \codepar{7,1,3} &  & $9$ & $72 \le T_{\mathrm{time}} \le 108$ & $p_{\mathrm{pseudo}} = (1.41 \pm 0.05) \times 10^{-5}$ \\
$d=3$ Surface code  &  & 17 & $\ge 18$ & $p_{\mathrm{pseudo}} = (1.37 \pm 0.03) \times 10^{-5}$ \\
Steane-EC \codepar{7,1,3} & &$\ge 35$ & 15 & $p_{\mathrm{pseudo}} = (3.84 \pm 0.01) \times 10^{-5}$ \\

\end{tabular}
\caption{Distance-three pseudo-threshold results for various FTEC protocols and noise models applied to the \codepar{5,1,3}, \codepar{7,1,3} and $d=3$ rotated surface code. We also include the number of time steps required to implement the protocols.}
\label{tab:PseudoThresholdsAllECschemesD3}
\end{table*}

We also note that the pseudo-threshold increases when $\tilde{p}$ goes from $p$ to $p/10$ for both the \codepar{5,1,3} and \codepar{7,1,3} codes when implemented using the flag 1-FTEC protocol. This is primarily due to the large circuit depth in flag-EC protocols since idle qubits locations significantly outnumber other gate locations. For the surface code, the opposite behaviour is observed. As was shown in \cite{FMMC12}, CNOT gate failures have the largest impact on the pseudo-threshold of the surface code. Thus, when idle qubits have lower failure probability, lower physical error rates will be required in order to achieve better logical failure rates. For instance, if idle qubits never failed, then performing error correction would be guaranteed to \textit{increase} the probability of failure due to the non-zero failure probability of other types of locations (CNOT, measurements and state-preparation). Lastly, the pseudo-threshold for Steane-EC also decreases with lower idle qubit failure rates, but the change in pseudo-threshold is not as large as the surface code. This is primarily due to the fact that all CNOT gates are applied transversally in Steane-EC, so that the pseudo-threshold is not as sensitive to CNOT errors compared to the surface code. Furthermore, most high-weight errors arising during the state-preparation of the logical ancilla's will be detected (see \cref{app:SteaneECSection}). Hence, idle qubit errors play a larger role than in the surface code, but Steane-EC has fewer idle qubit locations compared to flag-EC (see \cref{tab:PseudoThresholdsAllECschemesD3} for the circuit depths of all schemes).

\begin{table*}
\begin{tabular}{ c|c|c|c|c}
 FTEC scheme & Noise model & Number of qubits & Time steps ($T_{\mathrm{time}}$) &Pseudo-threshold  \\ \hline
 
Flag-EC \codepar{19,1,5} & $\tilde{p} = p$  & $22$ & $504 \le T_{\mathrm{time}} \le 960$ & $p_{\mathrm{pseudo}} = (1.14 \pm 0.02) \times 10^{-5}$ \\
$d=5$ Surface code  &  & 49 & $\ge 18$ & $p_{\mathrm{pseudo}} = (9.41 \pm 0.17) \times 10^{-4}$ \\
Steane-EC \codepar{19,1,5} & &$\ge 95$ & 15 & $p_{\mathrm{pseudo}} = (1.18 \pm 0.02) \times 10^{-3}$ \\ \hline

Flag-EC \codepar{19,1,5} & $\tilde{p} = p/10$  & $22$ & $504 \le T_{\mathrm{time}} \le 960$ & $p_{\mathrm{pseudo}} = (6.70 \pm 0.07) \times 10^{-5}$ \\
$d=5$ Surface code  &  & 49 & $\ge 18$ & $p_{\mathrm{pseudo}} = (7.38 \pm 0.22) \times 10^{-4}$ \\
Steane-EC \codepar{19,1,5} & &$\ge 95$ & 15 & $p_{\mathrm{pseudo}} = (4.42 \pm 0.27) \times 10^{-4}$ \\ \hline

Flag-EC \codepar{19,1,5} & $\tilde{p} = p/100$  & $22$ & $504 \le T_{\mathrm{time}} \le 960$ & $p_{\mathrm{pseudo}} = (7.74 \pm 0.16) \times 10^{-5}$ \\
$d=5$ Surface code  &  & 49 & $\ge 18$ & $p_{\mathrm{pseudo}} = (2.63 \pm 0.18) \times 10^{-4}$ \\
Steane-EC \codepar{19,1,5} & &$\ge 95$ & 15 & $p_{\mathrm{pseudo}} = (5.60 \pm 0.43) \times 10^{-5}$ \\

\end{tabular}
\caption{Distance-five pseudo-threshold results for various FTEC protocols and noise models applied to the \codepar{19,1,5} color code and $d=5$ rotated surface code. We also include the number of time steps required to implement the protocols.}
\label{tab:PseudoThresholdsAllECschemesD5}
\end{table*}

Although Steane-EC achieves the lowest logical failure rates compared to the other fault-tolerant error correction schemes, it requires a minimum of 35 qubits (more details are provided in \cref{app:SteaneECSection}). In contrast, the $d=3$ surface code requires 17 qubits, and flag 1-FTEC applied to the \codepar{5,1,3} code requires only 7 qubits. Therefore, if the probability of idle qubit errors is much lower than gate, state preparation and measurement errors, flag-FTEC methods could be good candidates for early fault-tolerant experiments. 

It is important to keep in mind that for the flag 1-FTEC protocol applied to the distance-three codes considered in this section, the same ancilla qubits are used to measure all stabilizers. A more parallelized version of flag-FTEC applied to the \codepar{7,1,3} code using four ancilla qubits instead of two is considered in \cref{app:CompactRepFlagQubit}. 

In computing the number of time steps required by the flag $t$-FTEC protocols, a lower bound is given in the case where there are no flags and the same syndrome is repeated $t+1$ times. In \cref{app:GeneralFTEC} it was shown that the full syndrome measurement for flag-FTEC is repeated at most $\frac{1}{2}(t^{2} + 3t + 2)$ times where $t = \lfloor (d-1)/2 \rfloor$. An upper bound on the total number of required time steps is thus obtained from a worst case scenario where syndrome measurements are repeated $\frac{1}{2}(t^{2} + 3t + 2)$ times.

For distance-five codes, the first thing to notice from \cref{fig:AllComparisonPlotsCombined} is that the slopes of the logical failure rate curves of flag-EC applied to the \codepar{19,1,5} code and $d=5$ surface code are different from the slopes of Steane-EC applied to the \codepar{19,1,5} code. In particular, $p_{\text{L}} = cp^{3} + \mathcal{O}(p^{4})$ for flag-EC and the surface code whereas  $p_{\text{L}} = c_{1}p^{2} + c_{2}p^{3} + \mathcal{O}(p^{4})$ for Steane-EC ($c$, $c_{1}$ and $c_{2}$ are constants that depend on the code and FTEC method). The reason that Steane-EC has non-zero $\mathcal{O}(p^{2})$ contributions to the logical failure rates is that there are instances where errors occurring at two different locations can lead to a logical fault. Consequently, the Steane-EC method that was used is not strictly fault-tolerant according to \cref{Def:FaultTolerantDef}. In \cref{app:SteaneECSection}, more details on the fault tolerant properties of Steane-EC are provided and a fully fault-tolerant implementation of Steane-EC is analyzed (at the cost of using more qubits). 

For $d=5$, the surface code achieves significantly lower logical failure rates compared to all other distance 5 schemes but uses 49 qubits instead of 22 for the \codepar{19,1,5} code. Furthermore, due the differences in the slopes of flag-2 FTEC protocol compared with Steane-EC applied to the \codepar{19,1,5} code, there is a regime where flag-2 FTEC achieves lower logical failure rates compared to Steane-EC. For $\tilde{p}=p/100$, it can be seen in \cref{fig:AllComparisonPlotsCombined} that this regime occurs when $p \lesssim 10^{-4}$. We also note that the pseudo-threshold of flag-EC applied to the \codepar{19,1,5} color code increases for all noise models whereas the pseudo-threshold decreases for the other FTEC schemes. Again, this is due to the fact that flag-EC has a larger circuit depth compared to the other FTEC methods and is thus more sensitive to idle qubit errors. 

Comparing the flag 2-FTEC protocol (applied to the \codepar{19,1,5} color code) to all the $d=3$ schemes that were considered in this section, due to the higher distance of the 19-qubit code, there will always be a parameter regime where the 19-qubit color code acheives lower logical failure rates than both the $d=3$ surface code and Steane-EC applied to the \codepar{7,1,3} code. In the case where $\tilde{p}=p/100$ and with $p \lesssim 1.5 \times10^{-4}$, using flag error correction with only 22 qubits outperforms Steane error correction (which uses a minimum of 35 qubits) and the $d=3$ rotated surface code (which uses 17 qubits).

Note the considerable number of time steps involved in a round of flag-EC, particularly in the $d=5$ case (see \cref{tab:PseudoThresholdsAllECschemesD5}).
For many applications, this is a major drawback, for example for quantum computation when the time of an error correction round dictates the time of a logical gate.
However there are some cases in which having a larger number of time-steps in an EC round while holding the logical error rate fixed is advantageous as it corresponds to a longer physical lifetime of the encoded information.
Such schemes could be useful for example in demonstrating that encoded logical quantum information can be stored for longer time scales in the lab using repeated rounds of FTEC.

\section{Conclusion}
\label{sec:Conclusion}

Building on definitions and a new flag FTEC protocol applied to distance-three and -five codes presented in \cref{sec:Section1}, in \cref{subsec:GeneralProtocol} we presented a general flag FTEC protocol, which we called flag $t$-FTEC, and which is applicable to stabilizer codes of distance $d = 2t+1$ that satisfy the flag $t$-FTEC condition. The protocol makes use of flag ancilla qubits which signal when $v$ faults lead to errors of weight greater than $v$ on the data when performing stabilizer measurements. In \cref{subsec:ApplicationProtocolColorCode,app:GeneralTflaggedCircuitConstruction} we gave explicit circuit constructions, including those needed for distance 3 and 5 codes measuring stabilizers of weight 4, 6 and 8. 
In \cref{subsec:Remarks} we gave a sufficient condition for codes to satisfy the requirements for flag $t$-FTEC. Quantum Reed-Muller codes, Surface codes and hexagonal lattice color codes were shown to be families of codes that satisfy the sufficient condition. 

The flag $t$-FTEC protocol could be useful for fault-tolerant experiments performed in near term quantum devices since it tends to use fewer qubits than other FTEC schemes such as Steane, Knill and Shor EC. In \cref{subsec:CompareFlagecschemes} we provided numerical evidence that with only 22 qubits, the flag $2$-FTEC protocol applied to the \codepar{19,1,5} color code can achieve lower logical failure rates than other codes using similar numbers of qubits such as the rotated distance-3 surface code and Steane-EC applied to the Steane code. 

A clear direction of future work would be to find optimal general constructions of $t$-flag circuits for stabilizers of arbitrary weight that improve upon the general construction given in \cref{App:GeneralwFlagCircuitConstruction}. Of particular interest would be circuits using few flag qubits and CNOT gates while minimizing the probability of false-positives (i.e. when the circuit flags without a high-weight error occurring). Finding other families of stabilizer codes which satisfy the sufficient or more general condition for flag $t$-FTEC would also be of great interest. One could also envisage hybrid schemes combining flag EC with other FTEC approaches. 

Another direction of future research would be to find general circuit constructions for simultaneously measuring multiple stabilizers while minimizing the number of required ancilla qubits. Further, we believe performing a rigorous numerical analysis to understand the impact of more compact circuit constructions on the codes threshold is of great interest.

Lastly, the decoding complexity (i.e. generating the flag error set lookup tables) is limited by the decoding complexity of the code. In some cases, for example concatenated codes, it may be possible to exploit some structure to generate the flag error sets more efficiently. In the case of concatenated code, the decoding complexity would be reduced to the decoding complexity of the codes used at every level. Finding other scalable constructions for efficient decoding schemes using flag error correction remains an open problem.

\section{Acknowledgements}
The authors would like to thank Krysta Svore, Tomas Jochym-O'Connor, Nicolas Delfosse and Jeongwan Haah for useful discussions.  We also thank Steve Weiss for providing the necessary computational tools that allowed us to complete our work in a timely manner. C. C. would like to acknowledge the support of QEII-GSST and thank Microsoft and the QuArC group for its hospitality where all of this work was completed. 

\newpage 

\bibliographystyle{unsrtnat} 
\bibliography{bibtex_chamberland}

\clearpage
\appendix

\section{Proof that the flag $t$-FTEC protocol satisfies the fault-tolerance criteria of \cref{Def:FaultTolerantDef}}
\label{app:ProtocolGeneralProof}

Consider the flag $t$-FTEC protocol described in \cref{subsec:GeneralProtocol}. 
\begin{claim}
If the flag $t$-FTEC condition is satisfied, then both fault-tolerance criteria of \cref{Def:FaultTolerantDef} will be satisfied.
\end{claim}

\begin{proof}
First note that the protocol always terminates. As was shown in the arguments leading to \cref{Eq:Nmax} presented in \cref{subsec:GeneralProtocol}, the maximum number of syndrome measurement rounds is $\frac{1}{2}(t^{2}+3t+2)$. 

To prove fault-tolerance, in what follows we assume that there are at most $t$-faults during the protocol. Also, we define a benign fault to be a fault that either leaves all syndrome measurements in the protocol unchanged. 

By repeating the syndrome measurement using $t$-flag circuits, the following cases exhaust all possible errors for the occurrence of at most $t$ faults. 

\emph{\underline{Case 1}: The same syndrome is measured $t-n_{\text{diff}}+1$ times in a row and there are no flags.} 

At any time during the protocol, if there are no flags, there can be at most $t-n_{\text{diff}}$ remaining faults that occur (since it is guaranteed that there were at least $n_{\text{diff}}$ faults). Therefore, if the same syndrome was measured $t-n_{\text{diff}}+1$ times in a row, at least one round (say $r$) had to have been fault-free yielding the correct syndrome corresponding to the data qubit errors present at that time. Applying $E_{\text{min}}(s)$ will remove those errors. Furthermore, since all syndrome measurements are identical and there are no flags, there can be at most $t-n_{\text{diff}}$ errors which are introduced on the data blocks from faults during the $t-n_{\text{diff}}+1$ syndrome measurement rounds (excluding round $r$). Since none of the errors change the syndrome, after applying the correction, the output state can differ from the input codeword by an error of weight at most $t-n_{\text{diff}}$ (if the total number of faults \text{and} input errors was $t$). For input states afflicted by an error of arbitrary weight, the output state will differ from a valid codeword (but not necessarily the input codeword) by an error of weight at most $t-n_{\text{diff}}$. Thus both conditions of \cref{Def:FaultTolerantDef} are satisfied.

\emph{\underline{Case 2}: There are no flags and $n_{\text{diff}} = t$.} 

The only way that $n_{\text{diff}} = t$ is if there were $t$-faults that each changed the syndrome measurement outcome. Further since there were no flags, an error $E$ afflicting the data qubits must satisfy $\text{wt}(E) \le t$. Thus repeating the syndrome measurement using non-flag circuits will correctly identify and remove the error in the case where the number of input errors and faults is $t$ or project the system back to the code space (to a possibly differ codeword) if there were $t$ faults and the input state was afflicted by an error of arbitrary weight .

\emph{\underline{Case 3}: A set of $t$ circuits $\{ C(g_{i_{1}}), \cdots ,  C(g_{i_{t}}) \}$ flagged.} 

Since $t$ circuits $\{ C(g_{i_{1}}), \cdots ,  C(g_{i_{t}}) \}$ flagged, then no other faults can occur during the protocol. Hence, when repeating the syndrome measurement using non-flag circuits, the measured syndrome will correspond to an error $E_{r} \in \tilde{E}_{t}^{t}(g_{i_{1}},\cdots g_{i_{t}},s)$. Since from the flag $t$-FTEC condition all elements of $\tilde{E}_{t}^{t}(g_{i_{1}},\cdots g_{i_{t}},s)$ are logically equivalent, the product of errors resulting from the flag circuits $\{ C(g_{i_{1}}), \cdots ,  C(g_{i_{t}}) \}$ will be corrected.

Note that for an input error $E_{\text{in}}$ of arbitrary weight and since the final round must be error free, applying a correction a correction from the set $\tilde{E}_{t}^{t}(g_{i_{1}},\cdots g_{i_{t}},s)$ is guaranteed to return the system to the codespace. Thus both conditions of \cref{Def:FaultTolerantDef} are satisfied.

\emph{\underline{Case 4}: The $m$ circuits $\{ C(g_{i_{1}}), \cdots ,  C(g_{i_{m}}) \}$ flagged with $1 \le m < t$, $n_{\text{diff}} = t -m $.} 

Here we can assume that at any point during the protocol and after the $j$'th flag, the syndrome never repeated more than $t-j-n_{\text{diff}}$ times. Otherwise case 5 of the protocol would already have occurred. 

As $m$ circuits $\{ C(g_{i_{1}}), \cdots ,  C(g_{i_{m}}) \}$ have flagged and $n_{\text{diff}}=t-m$, then there can be no more faults. The final syndrome measurement using non-flag circuits will yield a syndrome corresponding to an error in the set $\tilde{E}_{t}^{m}(g_{i_{1}},\cdots g_{i_{m}},s)$ (and all elements are logically equivalent from the flag $t$-FTEC condition). Applying a recovery operator chosen from this set will thus remove the errors afflicting the data. If the input state differs from a valid codeword by an error of arbitrary weight, by definition of $\tilde{E}_{t}^{m}(g_{i_{1}},\cdots g_{i_{m}},s)$ the output state will be a valid codeword. 

\emph{\underline{Case 5}:  The $m$ circuits $\{ C(g_{i_{1}}), \cdots ,  C(g_{i_{m}}) \}$ flagged with $1 \le m < t$, $n_{\text{same}} = t -m - n_{\text{diff}} + 1$.} 

Given that $m$ circuits  $\{ C(g_{i_{1}}), \cdots ,  C(g_{i_{m}}) \}$ flagged,  there are $r$ remaining faults that don't result in a flag with  $n_{\text{diff}} \le r \le t-m$. In this case, after the $m$'th flag, the syndrome measurement was repeated using $t$-flag circuits $t-m- n_{\text{diff}}+1$ times in a row and all syndromes were the same. It is thus guaranteed that at least one of the syndrome measurements $s$ was fault-free and correctly identified the errors arising from the flags and errors causing the syndrome to change giving $n_{\text{diff}}$ (along with some error $E$ which did not cause the circuits to flag with $\text{wt}(E) \le t-m- n_{\text{diff}}$). Consequently, if there are no errors on the input state, the overall error on the data will be $EE_{r}$ with $E_{r} \in \bigcup_{j=0}^{t-m-n_{\text{diff}}}\tilde{E}_{t}^{t-j-n_{\text{diff}}}(g_{i_{1}},\cdots ,g_{i_{m}}, s)$. Since all elements in $\bigcup_{j=0}^{t-m-n_{\text{diff}}}\tilde{E}_{t}^{t-j-n_{\text{diff}}}(g_{i_{1}},\cdots ,g_{i_{m}}, s)$ are logically equivalent from the flag $t$-FTEC condition, by choosing a correction from this set, the output state can differ from the input codeword by an error of at most weight $t-m- n_{\text{diff}}$. 

If there is an input error of arbitrary weight, then again one of the $t-m- n_{\text{diff}}+1$ rounds will have the correct syndrome $s$. 
The actual state of the data qubits after the protocol can differ from the state which had the correct syndrome by an error of weight at most $t-m- n_{\text{diff}}$.
Therefore applying any correction with syndrome $s$ will return the system to the code space up to an error of weight at most $t-m- n_{\text{diff}}$.

\end{proof}

\section{Fault-tolerant state preparation and measurement using flag $t$-FTEC}
\label{app:StatePrepAndMeasure}

In this section we show how to fault-tolerantly prepare a logical $\ket{\overline{0}}$ state and how to perform fault-tolerant measurements for codes that satisfy the flag $t$-FTEC condition of \cref{app:GeneralFTEC}. Note that there are several methods that can be used for doing so. Here we follow a procedure similar to that shown in \cite{Gottesman2010} when performing Shor EC. However, compared to Shor EC, the flag $t$-FTEC protocol requires fewer qubits. Furthermore, postselection is not necessary.

Consider an $n$-qubit stabilizer code $C$ with stabilizer group $\mathcal{S} = \langle g_{1},\cdots , g_{n-k} \rangle$ that can correct up to $t$ errors. Notice that the encoded  $\ket{\overline{0}}$ state is a $+1$ eigenstate of the logical $\overline{Z}$ operator and all of the codes stabilizer generators. For $k$ encoded qubits, $\ket{\overline{0}}$ would be $+1$ eigenstate of $\{ \overline{Z}_{1}, \cdots \overline{Z}_{k} \}$ and all of the codes stabilizers. For notational simplicity, in what follows we assume $k=1$.  

The state $\ket{\overline{0}}$ is a stabilizer state completely specified by the full stabilizer generators of $\mathcal{S}$ and $\overline{Z}$. We can think of $\mathcal{S}' = \langle g_{1}, \cdots g_{n-1}, \overline{Z} \rangle$ as a stabilizer code with zero encoded qubits and a $2^{0}=1$ dimensional Hilbert space. Thus any state which is a $+1$ eigenstate of all operators in $\mathcal{S}'$ will correspond to the encoded $\ket{\overline{0}}$ state.

Now, suppose we prepare $\ket{\overline{0}}_{\text{in}}$ using a non-fault-tolerant encoding and perform a round of flag $t$-FTEC using the extended stabilizers $\langle g_{1}, \cdots g_{n-1}, \overline{Z} \rangle$. Then by the second criteria of \cref{Def:FaultTolerantDef}, the output state $\ket{\overline{0}}_{\text{out}}$ is guaranteed to be a valid codeword with at most $t$ single-qubit errors. But for the extended stabilizers $\langle g_{1}, \cdots g_{n-1}, \overline{Z} \rangle$ there is only \textit{one} valid codeword which corresponds to the encoded $\ket{\overline{0}}$ state.  In fact, by the second criteria of \cref{Def:FaultTolerantDef}, any $n$-qubit input state prepared using non-fault-tolerant circuits is guaranteed to be an encoded $\ket{\overline{0}}$ state if there are no more than $t$ faults in the EC round.

We point out that the flag $t$-FTEC condition of \cref{subsec:GeneralProtocol} is trivially satisfied for $\mathcal{S}'$ since the codes logical operators are now stabilizers. In other words, if two errors belong to the set $\tilde{E}^{m}_{t}(g_{i_{1}}, \cdots, g_{i_{k}},s)$, then their product will always be a stabilizer. Therefore, the flag $t$-FTEC protocol can always be applied for the code $\mathcal{S}'$.  

To summarize, the encoded $\ket{\overline{0}}$ state can be prepared by first preparing any $n$-qubit state using non-fault-tolerant circuits followed by applying a round of flag $t$-FTEC using the extended stabilizers $\langle g_{1}, \cdots g_{n-1}, \overline{Z} \rangle$. This guarantees that the output state will be the encoded $\ket{\overline{0}}$ state with at most $t$ single-qubit errors. 

Now suppose we want to measure the eigenvalue of a logical operator $\overline{P}$ where $P$ is a Pauli. If $C$ is a CSS code and the logical operator being measured is $X$ or $Z$, one could measure the eigenvalue by performing the measurement transversally. So suppose $C$ is not a CSS code. From \cite{Gottesman2010} we require that performing a measurement with $s$ faults on an input state with $r$ errors ($r+s \le t$) is equivalent to correcting the $r$ errors and performing the measurement perfectly. The protocol for fault-tolerantly measuring the eigenvalue of $\overline{P}$ is described as follows:

\begin{enumerate}
\item Perform flag $t$-FTEC.
\item Use a $t$-flag circuit to measure the eigenvalue of $\overline{P}$.
\item Repeat steps 1 and 2 $2t+1$ times and take the majority of the eigenvalue of $\overline{P}$. 
\end{enumerate}

Step 1 is used to remove input errors to the measurement procedure. However during error correction, a fault can occur which could cause a new error on the data. Thus by repeating the measurement without performing error correction, the wrong state would be measured each time if there were no more faults. But repeating the syndrome $2t+1$ times, it is guaranteed that at least $t+1$ of the syndrome measurements had no faults and that the correct eigenvalue of $\overline{P}$ was measured. Thus taking the majority of the measured eigenvalues will give the correct answer. 

Note that during the fault-tolerant measurement procedure, if there is a flag either during the error correction round or during the measurement of $\overline{P}$, when error correction is performed one corrects based on the possible set of errors resulting from the flag.

\section{Candidate general $w$-flag circuit construction}
\label{App:GeneralwFlagCircuitConstruction}

\begin{figure*}
	\centering
	\includegraphics[width=1.05\textwidth]{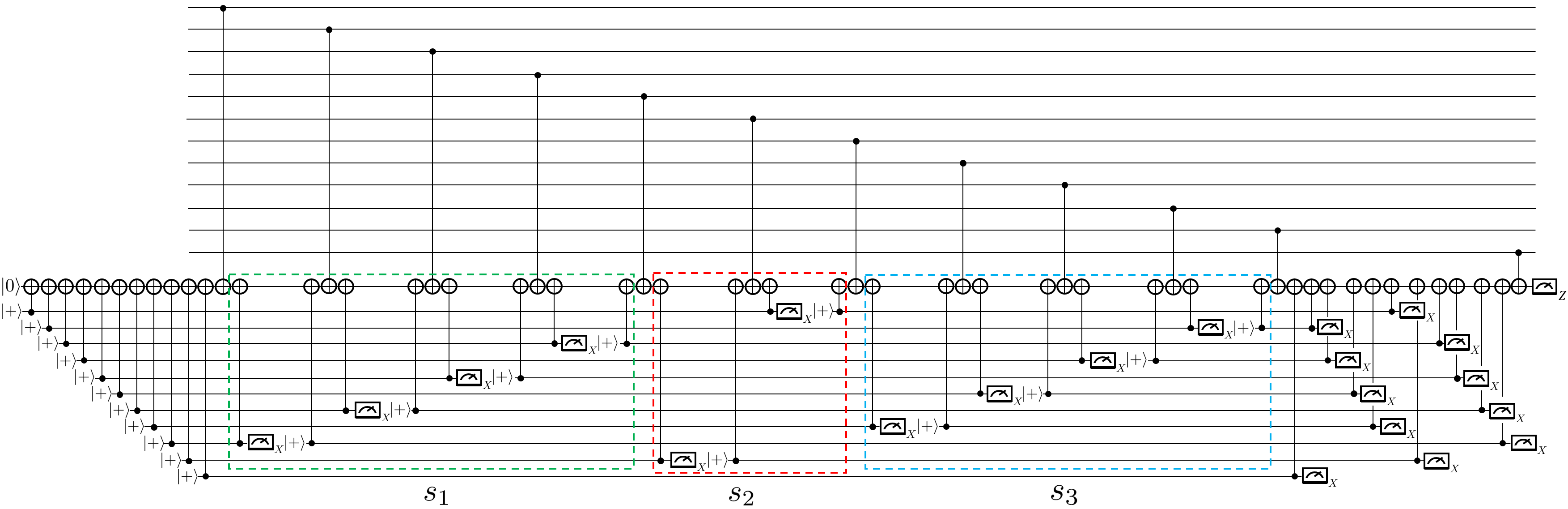}
	\caption{Illustration of the general $w$-flag circuit construction for $w=12$. In general, the circuit requires $w-1$ flag qubits and is implemented using $7w-8$ time steps. The circuit consists of two families of $\text{CNOT}_{\text{fm}}$ gates. For the first family, with the first set of $\text{CNOT}_{\text{fm}}$ gates located before the first $\text{CNOT}_{\text{dm}}$ gate, the partnering $\text{CNOT}_{\text{fm}}$ gates are divided into three sets $s_{1}$, $s_{2}$ and $s_{3}$ which are enclosed in the green, red and blue dashed boxes. In general, $s_{1}$ and $s_{3}$ both contain $(w-4)/2$ $\text{CNOT}_{\text{fm}}$ gates. In $s_{1}$, the $j$'th control qubit is at position $w+2(j+1)$ and in $s_{3}$ it is at position $w+2j+1$ with $j \in \{1,2, \cdots , (w-4)/2 \}$ In $s_{2}$, the control qubits are always located at the $w+2$'th and $2w-1$'th qubits. Lastly, note that qubits are reused for implementing the second family of $\text{CNOT}_{\text{fm}}$ gates. The partnering $\text{CNOT}_{\text{fm}}$ gates are located in between the $w-1$ and $w$'th $\text{CNOT}_{\text{dm}}$ gates following an identical pattern as in $s_{1},s_{2}$ and $s_{3}$ (in $s_{1}$ and $s_{3}$ the CNOT's are implemented in reverse order).
	}
	\label{fig:GeneralWflagFig}
\end{figure*}

In this section we provide a candidate general $w$-flag circuit construction for measuring the stabilizer $Z^{\otimes w}$. Although we do not provide a rigorous proof that our construction results in a $w$-flag circuit, we give several arguments as evidence that it satisfies all the criteria of a $w$-flag circuit. An illustration of the circuit construction (for $w=12$) is given in \cref{fig:GeneralWflagFig} and the description for how the circuit is constructed for arbitrary $w$ is provided in the caption. 

In what follows, we can restrict our attention to the case in which all $v$ faults occur on CNOT gates in the circuit. The effect on the measurement outcomes and data qubits due to a set of $v$ faults that include faults at idle and measurement locations can always occur due to at most $v$ faults at CNOT locations only (as every qubit is involved in at least one CNOT). Moreover, we can assume that for $\text{CNOT}_{\text{fm}}$ gates, the faults belong to the set $\{ IZ, ZI, ZZ \}$ since $X$ errors would never propagate to the data or affect the measurement outcome of a flag qubit. For $\text{CNOT}_{\text{dm}}$ gates, we can assume that faults belong to the set $\{XZ,XI \}$. We only consider $Z$ errors on the target qubit of a $\text{CNOT}_{\text{dm}}$ for the same reason that was given for $\text{CNOT}_{\text{fm}}$ gates. For the control qubit, an X errors guarantees that the weight of the data qubit error increases even after the application of a satbilizer (since we are measuring $Z^{\otimes w}$).

We will use the following useful terminology: we say that a single-qubit Pauli at a time step in the circuit propagates to a qubit at a particular time-step if it would do so in the fault-free circuit. Given a single-qubit Pauli at a time step in the circuit, we say that another qubit is affected by the Pauli if it propagates to that qubit in any time step.

We now provide arguments for why the circuit is a $w$-flag circuit. First, note that every $\text{CNOT}_{\text{fm}}$ gate comes as part of a pair with the measurement qubit being the target qubit. This ensures that when the circuit is fault-free, it implements a projective measurement of $Z^{\otimes w}$ without flagging. Next, notice that apart from the last two $\text{CNOT}_{\text{dm}}$ gates, each $\text{CNOT}_{\text{dm}}$ gate is followed by two $\text{CNOT}_{\text{fm}}$ gates, one with its partnering $\text{CNOT}_{\text{fm}}$ located before the first $\text{CNOT}_{\text{dm}}$ and the other partner is located in between the last two $\text{CNOT}_{\text{dm}}$ gates. Thus if there is a single $Z$ error on the measurement qubit which propagates to any of the data qubits, the circuit will flag. 

In all circuits considered in this section, $s_{0}$ will correspond to the sequence of $\text{CNOT}_{\text{fm}}$ gates that come before the first $\text{CNOT}_{\text{dm}}$ gate. First consider the shorter circuit construction using only the first family of $\text{CNOT}_{\text{fm}}$ gates from the construction in \cref{fig:GeneralWflagFig} (see the example in \cref{fig:ErrorPropLogic}). We can separate the set of all locations into subsets including two $\text{CNOT}_{\text{fm}}$ gates and one $\text{CNOT}_{\text{dm}}$ gate as shown in \cref{fig:CNOTsectionAppendix} (apart from the last $\text{CNOT}_{\text{dm}}$). This circuit segment can increase the weight of the data error by at most one. There are four cases with inputs on the measurement qubit before the first $\text{CNOT}_{\text{fm}}$ and $\text{CNOT}_{\text{dm}}$ being $\{(I,I),(I,Z),(Z,I),(Z,Z)\}$. Note that if the following property held for each segment, then the circuit would be $w$-flag: for all inputs to the segment, if the weight of the data error increases and there are no faults in the segment, the segment flags. Unfortunately, for the input $(Z,Z)$, this is not the case. Both input $Z$ must come from at least two faults. 

\begin{figure}
	\centering
	\includegraphics[width=0.45\textwidth]{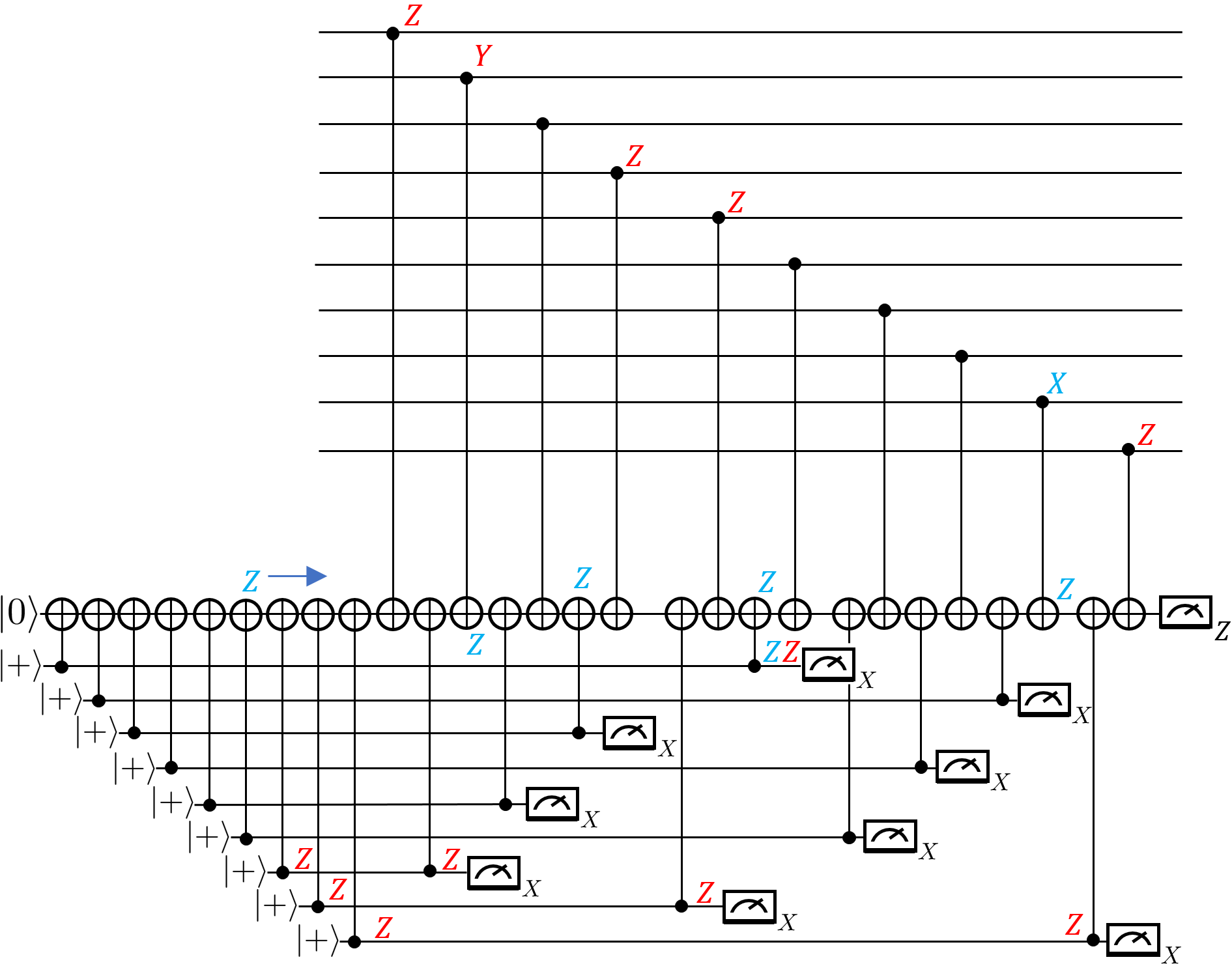}
	\caption{Example of five faults that lead to an error of weight six on the data without causing a flag when only the first family of  $\text{CNOT}_{\text{fm}}$ gates are used in the construction of \cref{fig:GeneralWflagFig} (here $w=10$). Errors arising from faults are shown in blue and the resulting errors after propagating through the CNOT gates are shown in red. 
	}
	\label{fig:ErrorPropLogic}
\end{figure}

Note that if $v$ faults results in a data qubit error of weight greater than $v$ without causing the circuit in \cref{fig:ErrorPropLogic} to flag, there must be either an $IZ$ fault followed by no fault in a consecutive pair of $\text{CNOT}_{\text{fm}}$ gates belonging to $s_{0}$ or a $ZZ$ fault followed by two $\text{CNOT}_{\text{fm}}$ gates that don't fail in $s_{0}$. 

Moreover, a poor choice of ordering of the $\text{CNOT}_{\text{fm}}$ gates in $s_{1},s_{2}$ and $s_{3}$ can result in four faults causing a weight $\frac{w}{2}+1$ error on the data without causing the circuit to flag. Therefore, the ordering of the $\text{CNOT}_{\text{fm}}$ gates in the sets $s_{1},s_{2}$ and $s_{3}$ is chosen such that most $Z$ errors in $s_{0}$ that first propagate to flag qubits connected to gates in $s_{1}$, will then propagate to flag qubits in $s_{3}$ and vice-versa. Typically, if a $Z$ error propagates through multiple $\text{CNOT}_{\text{dm}}$ gates in $s_{1}$, then unless $\text{CNOT}_{\text{fm}}$ gates in $s_{3}$ fail, the flag qubits affected by the $Z$ error would flag. Furthermore, the total number of required failures for gates in $s_{3}$ to cancel the $Z$ errors will typically be equal to the number of times the $Z$ error propagated to the data. 

There are however cases which don't flag in which $v$ faults in the circuit construction presented in \cref{fig:ErrorPropLogic} lead to more than $v$ errors on the data qubit, such as the example given in the figure. All such problematic cases that we found had a $Z$ error on the target qubit in one of the last few $\text{CNOT}_{\text{fm}}$ gates in $s_0$, followed by a $Z$ error on the target qubit in one of the first few $\text{CNOT}_{\text{dm}}$ gates in $s_1$. Then further $Z$ errors occur throughout the remainder of the circuit which propagate to the data while preventing the flag qubits affected by the previous errors from flagging. Further, a $Z$ error on the control qubit of the second $\text{CNOT}_{\text{fm}}$ in $s_{2}$ cancels the $Z$ which propagates to the flag qubit coupled to that  $\text{CNOT}_{\text{fm}}$ gate. 

\begin{figure}
	\centering
	\includegraphics[width=0.4\textwidth]{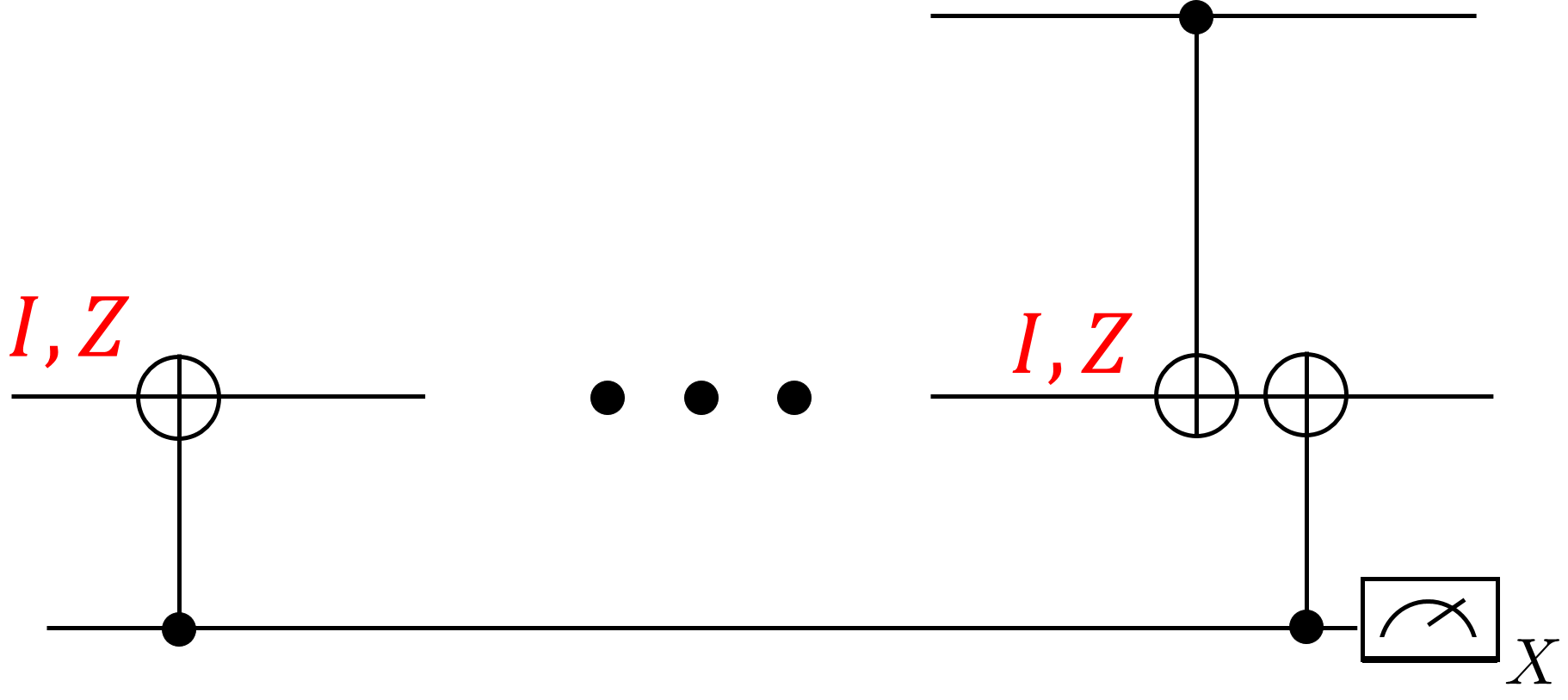}
	\caption{Illustration of a pair of  $\text{CNOT}_{\text{fm}}$ gates as well as a segment of  $\text{CNOT}_{\text{dm}}$ followed by  $\text{CNOT}_{\text{fm}}$ gate. The first  $\text{CNOT}_{\text{fm}}$ gate belongs to the sequence of  $\text{CNOT}_{\text{fm}}$ gates that come before the first  $\text{CNOT}_{\text{dm}}$ gate (see \cref{fig:GeneralWflagFig}).
	}
	\label{fig:CNOTsectionAppendix}
\end{figure}

This particular problematic fault pattern would lead to flags if it occurred within the full circuit construction of \cref{fig:GeneralWflagFig} (if the additional locations of the larger circuit do not fail). As this was the only type of problematic fault pattern that we found, one would hope that all problematic fault patterns are rendered non problematic provided no additional locations fail. Since the additional $\text{CNOT}_{\text{fm}}$ gates always occur immediately after one of the original $\text{CNOT}_{\text{fm}}$ gates (or after the last $\text{CNOT}_{\text{fm}}$ gate), as far as the flag properties of the original circuit are concerned, no new problematic fault patterns are introduced.

We conclude this section by noting that our candidate general $w$-flag circuit construction requires $w-1$ flag qubits and is implemented in $7w-8$ time steps. This is clearly not optimal in general since for example, as shown in \cref{fig:6flagCircuit}, a $w$-flag circuit was found (for $w=6$) which requires only three flag qubits instead of five and the circuit is implemented in 14 time steps instead of 34. It is thus still an open problem to find optimal $w$-flag circuits for arbitrary $w$. 

\section{Quantum Reed-Muller codes}
\label{app:QRMcodes}

In this section we first describe how to construct the family of quantum Reed-Muller codes $\text{QRM}(m)$ with code parameters $[\![ 2^{m}-1,k=1,d=3 ]\!]$ following \cite{ADP14}. We then show that the family of $\text{QRM}(m)$ codes satisfy the sufficient flag 1-FTEC condition of \cref{subsec:Remarks}.

Reed-Muller codes of order $m$ ($\text{RM}(1,m)$) are defined recursively from the following generator matrices:
First, $\text{RM}(1,1)$ has generator matrix

\begin{align}
G_{1}=\left( \begin{array}{cc}
	1 & 1\\
	0 & 1\\                             
\end{array} \right),
\label{eq:G1mat}
\end{align}
and $\text{RM}(1,m+1)$ has generator matrix 
\begin{align}
G_{m+1}=\left( \begin{array}{cc}
	G_{m} & G_{m}\\
	0 & 1\\                             
\end{array} \right),
\label{eq:GMmat}
\end{align}
where 0 and 1 are vectors of zeros and ones in \cref{eq:GMmat}.
The dual of $\text{RM}(1,m+1)$ is given by the higher order Reed-Muller code $\text{RM}(m-2,m)$. In general, the generator matrices for higher-order Reed-Muller codes $\text{RM}(r,m)$ are given by 

\begin{align}
H_{r,m+1}=\left( \begin{array}{cc}
	H_{r,m} & H_{r,m}\\
	0 & H_{r-1,m}\\                             
\end{array} \right).
\label{eq:GRMmat}
\end{align}

with 
\begin{align}
H_{2,1}=H_{1,1} = \left( \begin{array}{cc}
	1 &1\\
	0 & 1\\                             
\end{array} \right),
\label{eq:G21mat}
\end{align}

The $X$ stabilizer generators of $\text{QRM}(m)$ are derived from shortened Reed-Muller codes where the first row and column of $G_{m}$ are deleted. We define the resulting generator matrix as $\overline{G}_{m}$. The $Z$ stabilizer generators are obtained by deleting the first row and column of $H_{m-2,m}$. Similarly, we define the resulting generator matrix as $\overline{H}_{m-2,m}$.

As was shown in \cite{ADP14}, $\text{rows}(\overline{G}_{m}) \subset \text{rows}(\overline{H}_{m-2,m})$ and each row has weight $2^{m-1}$. Therefore, all the $X$-type stabilizer generators of $\text{QRM}(m)$ have corresponding $Z$-type stabilizers. By construction, the remaining rows of $\overline{H}_{m-2,m}$ will have weight $2^{m-2}$. Furthermore, every weight $2^{m-2}$ row has support contained within some weight $2^{m-1}$ row of the generator matrix $\overline{H}_{m-2,m}$. Therefore, every $Z$-type stabilizer generator has support within the support of an $X$ generator.

\section{Implementation of Steane error correction}
\label{app:SteaneECSection}

In this section we describe how to implement Steane error correction and discuss its fault-tolerant properties. We also provide a comparison of a version of Steane error correction with flag 2-FTEC protocol described in \cref{subsec:Distance5protocol} applied to the \codepar{19,1,5} code.

Steane error correction is a fault-tolerant scheme that applies to the Calderbank-Shor-Steane (CSS) family of stabilizer codes \cite{Steane97}. In Steane error correction, the idea is to use encoded $\ket{\overline{0}}$ and $\ket{\overline{+}} = (\ket{\overline{0}} + \ket{\overline{1}})/\sqrt{2}$ ancilla states to perform the syndrome extraction. The ancilla's are encoded in the same error correcting code that is used to protect the data. The $X$ stabilizer generators are measured by preparing the encoded $\ket{\overline{0}}$ state and performing transversal CNOT gates between the ancilla and the data, with the ancilla acting as the control qubits and the data acting as the target qubits. After applying the transversal CNOT gates, the syndrome is obtained by measuring $\ket{\overline{0}}$ transversally in the $X$-basis. The code construction for CSS codes is what guarantees that the correct syndrome is obtained after applying a transversal measurement (see \cite{Gottesman2010} for more details). 

\begin{figure}
\centering
\begin{subfigure}{0.25\textwidth}
\includegraphics[width=\textwidth]{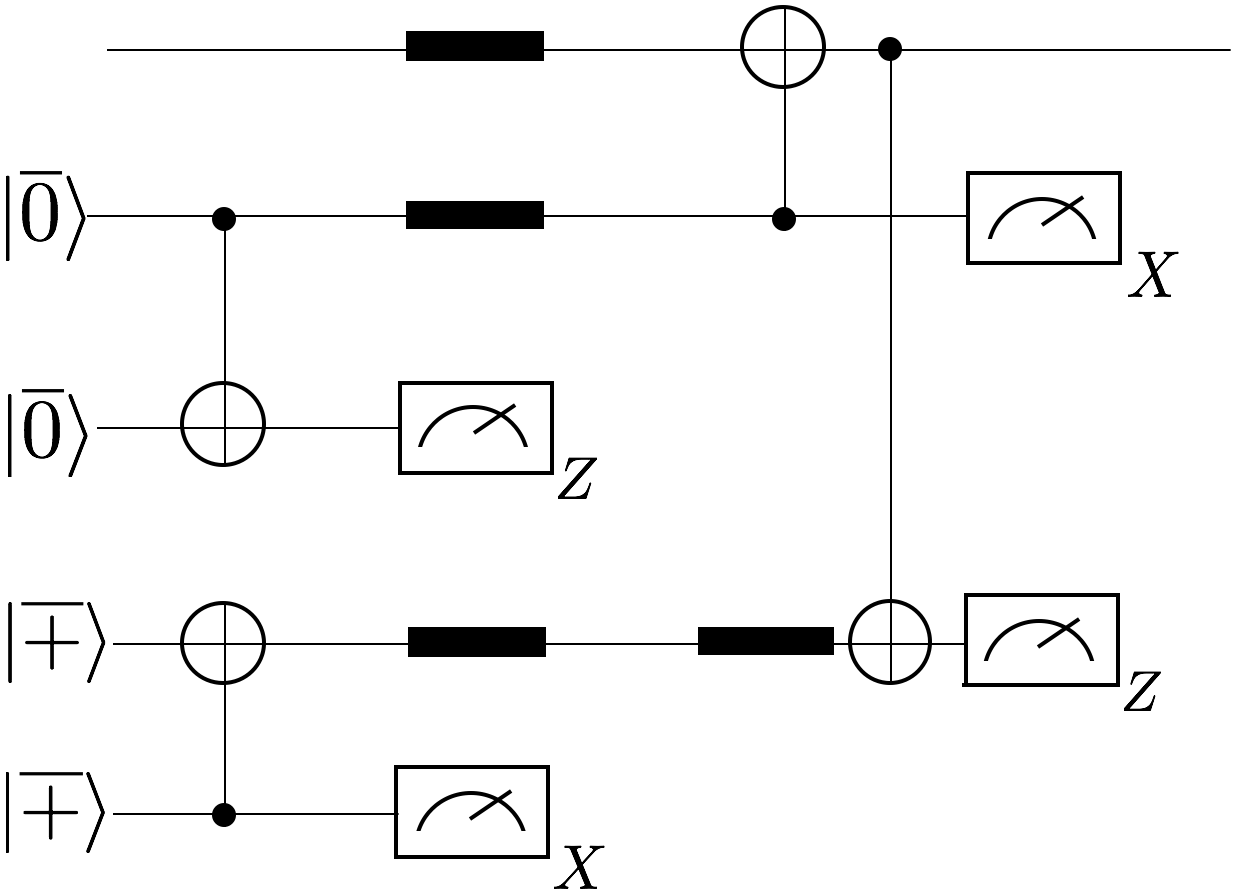}
\caption{}
\label{fig:SteaneCSSforPerfect}
\end{subfigure}
\begin{subfigure}{0.25\textwidth}
\includegraphics[width=\textwidth]{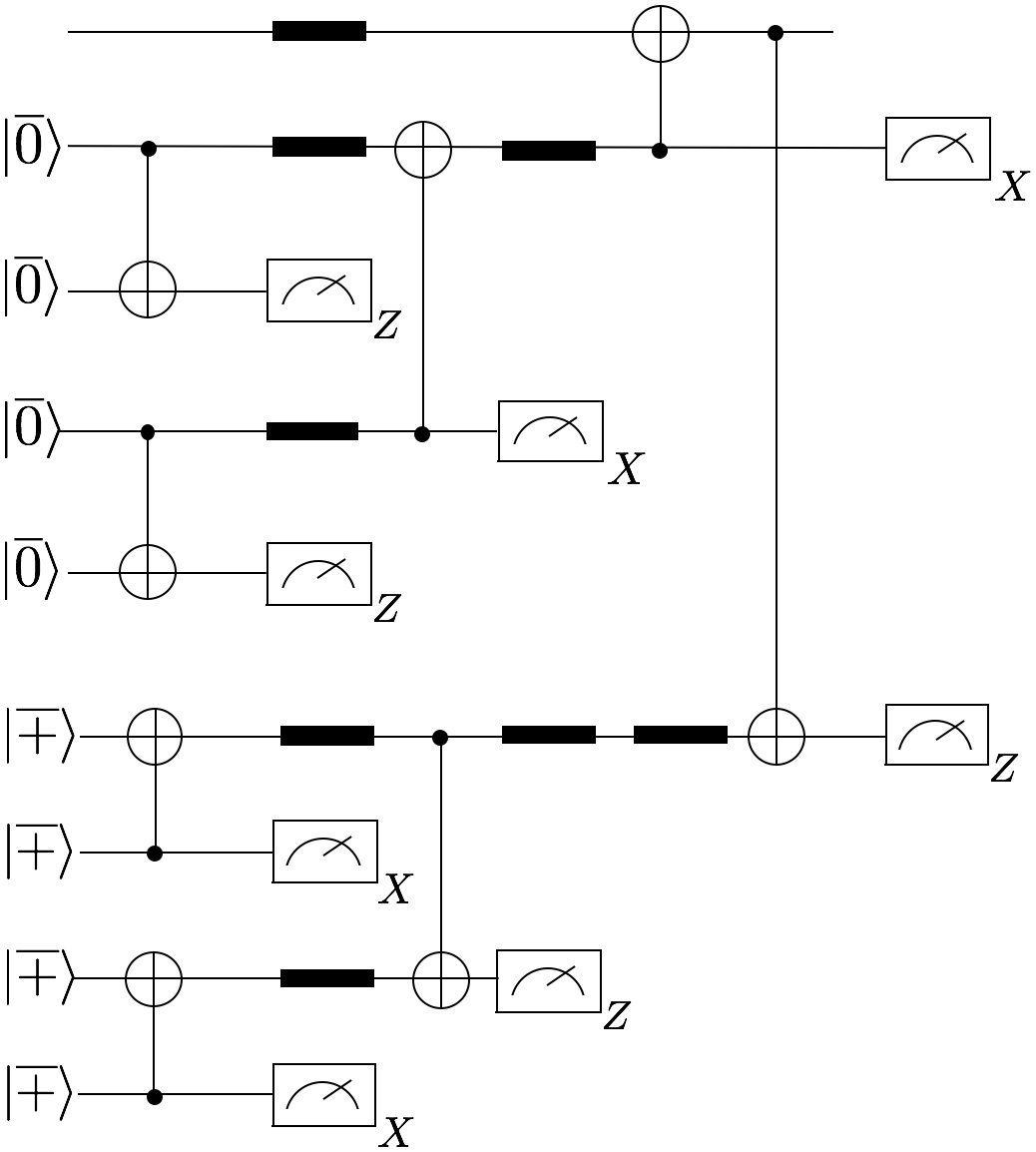}
\caption{}
\label{fig:SteaneCSSforNonPerfect}
\end{subfigure}
\caption{(a) Fault-tolerant Steane error correction circuit for distance-three CSS codes. Each line represents an encoded qubit. The circuit uses only two encoded $\ket{\overline{0}}$ and $\ket{\overline{+}}$ ancilla states (encoded in the same error correcting code which protects the data) to ensure that faults in the preparation circuits of the ancilla's don't spread to the data block. (b) Fault-tolerant Steane error correction circuit which can be used for any distance-three CSS stabilizer code encoding the data. There are a total of eight encoded ancilla qubits instead of four. The dark bold lines represent resting qubits. Note that the circuit in \cref{fig:SteaneCSSforNonPerfect} could in some cases be used for higher distance CSS codes with appropriately chosen circuits for $\ket{\overline{0}}$ and $\ket{\overline{+}}$ ancilla states (see \cite{PR12}).}
\label{fig:SteaneECcircuitsAll}
\end{figure}

\begin{table*}[t]
\begin{tabular}{ c|c|c|c|c}
 FTEC scheme & Noise model & Number of qubits & Time steps ($T_{\mathrm{time}}$) &Pseudo-threshold  \\ \hline
Full Steane-EC & $\tilde{p} = p$ & $\ge 171$ & 15 & $p_{\mathrm{pseudo}} = (3.50 \pm 0.14) \times 10^{-3}$ \\
Full Steane-EC & $\tilde{p} = p/100$ & $\ge 171$ & 15 & $p_{\mathrm{pseudo}} = (1.05 \pm 0.04) \times 10^{-3}$ \\
Flag-EC \codepar{19,1,5} & $\tilde{p} = p$ & 22 & $504 \le T_{\mathrm{time}} \le 960$ & $p_{\mathrm{pseudo}} = (1.14 \pm 0.02) \times 10^{-5}$ \\
Flag-EC \codepar{19,1,5} & $\tilde{p} = p/100$ & 22 & $504 \le T_{\mathrm{time}} \le 960$ & $p_{\mathrm{pseudo}} = (7.74 \pm 0.16) \times 10^{-5}$ \\
  \end{tabular}
\caption{Pseudo-threshold results for the Full Steane and flag 2-FTEC protocol applied to the \codepar{19,1,5} code. Since the Steane error correction protocol is non-deterministic, the number of qubits will depend on how many times the encoded states are rejected. For low error rates, the states are accepted with high probability so that the average number of qubits is $\approx 171$. Our three qubit flag error correction protocol requires at most six rounds of syndrome measurements, with each round using flag circuits requiring 168 time steps and the round using non-flag circuits requiring 120 time steps. However, for low noise rates, the average number of time steps will be close to 504 (since at least three rounds are required for the protocol to be fault-tolerant).}
\label{tab:PseudoThreshFullSteane}
\end{table*}

Similarly, the $Z$-stabilizer generators are measured by preparing the encoded $\ket{\overline{+}}$, applying CNOT gates transversally between the ancilla and the data with the data acting as the control qubits and the ancilla's acting as the target qubits. The syndrome is then obtained by measuring $\ket{\overline{+}}$ transversally in the $Z$-basis. 

The above protocol as stated is not sufficient in order to be fault-tolerant. The reason is that in general the circuits for preparing the encoded $\ket{\overline{0}}$ and $\ket{\overline{+}}$ are not fault-tolerant in the sense that a single error can spread to a multi-weight error which could then spread to the data block when applying the transversal CNOT gates. To make the protocol fault-tolerant, extra $\ket{\overline{0}}$ and $\ket{\overline{+}}$ ancilla states (which we call verifier qubits) are needed to check for multi-weight errors at the output of the ancilla states. 

\begin{figure}
\centering
\begin{subfigure}{0.3\textwidth}
\includegraphics[width=\textwidth]{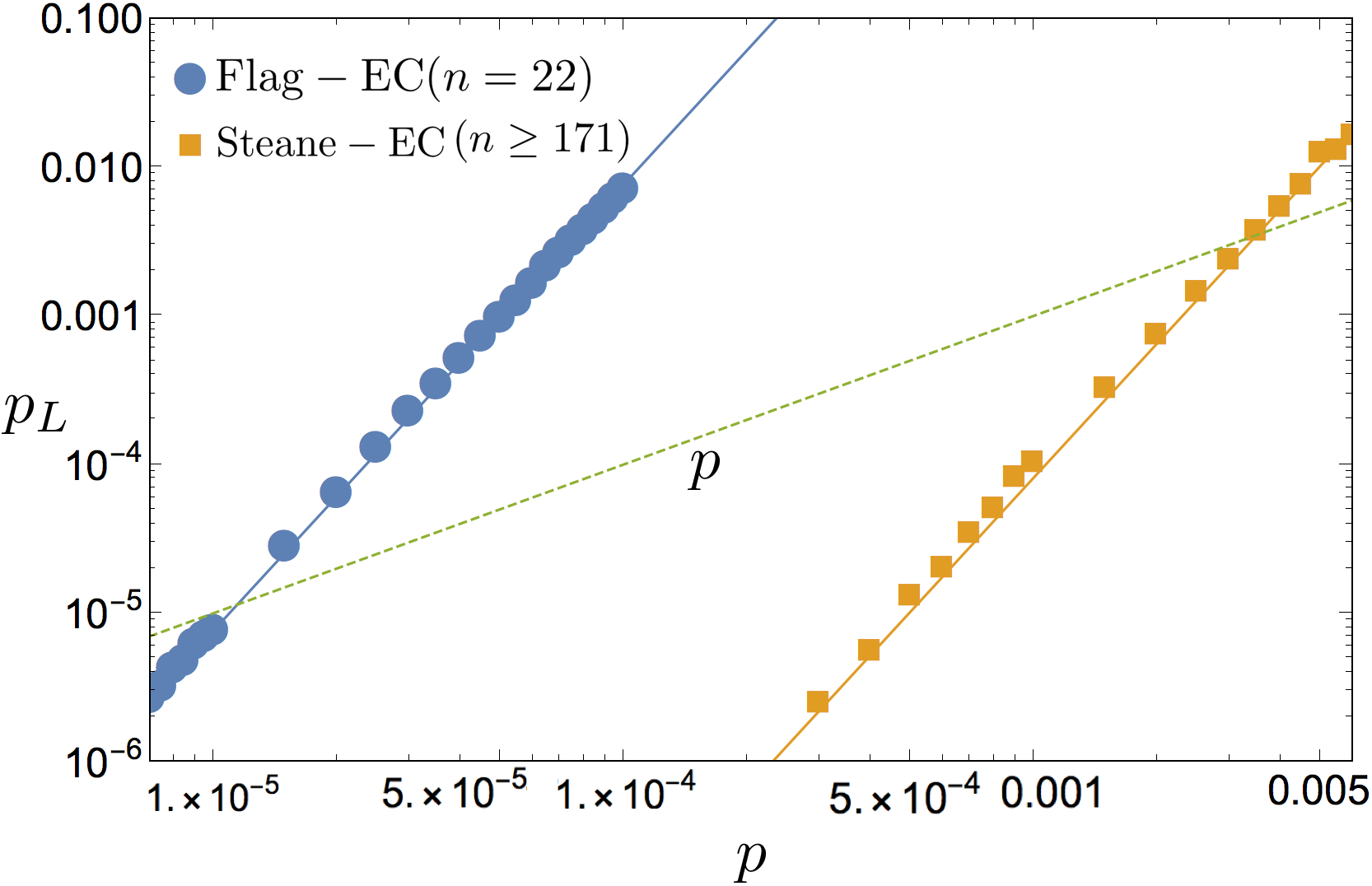}
\caption{}
\label{fig:ResultFullFaultTolerantSteanep}
\end{subfigure}
\begin{subfigure}{0.3\textwidth}
\includegraphics[width=\textwidth]{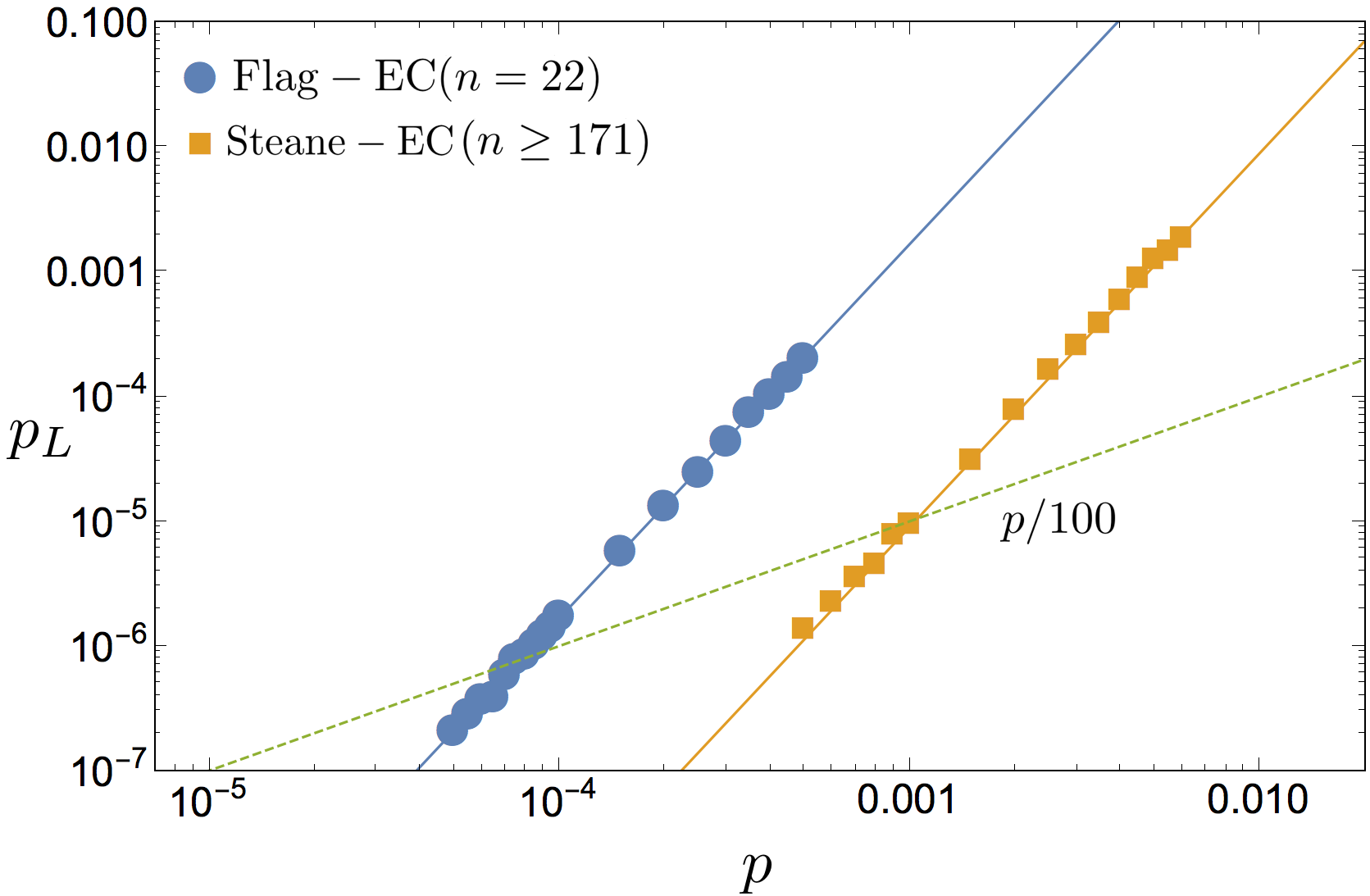}
\caption{}
\label{fig:ResultFullFaultTolerantSteanep100}
\end{subfigure}
\caption{Logical failure rate of the full fault-tolerant Steane error correction approach of \cref{fig:SteaneCSSforNonPerfect} and the flag 2-FTEC protocol of \cref{subsec:Distance5protocol} applied to the \codepar{19,1,5} code. In (a) idle qubits are chosen to fail with a total probability $\tilde{p} = p$ while in (b) idle qubits fail with probability $\tilde{p} = p/100$. The intersection between the dashed curve and solid lines represent the pseudo-threshold of both error correction schemes. }
\label{fig:ResultFullFaultTolerantSteanevsCB}
\end{figure}

For the $\ket{\overline{0}}$ ancilla, multiple $X$ errors can spread to the data if left unchecked. Therefore, another encoded $\ket{\overline{0}}$ ancilla is prepared and a transversal CNOT gate is applied between the two states with the ancilla acting as the control and the verifier state acting as target. Anytime $X$ errors are detected the state is rejected and the error correction protocol start over. Further, if the verifier qubit measures a $-1$ eigenvalue of the logical $Z$ operator, the ancilla qubit is also rejected. A similar technique is used for verifying the $\ket{\overline{+}}$ state (see \cref{fig:SteaneCSSforPerfect}).

For the \codepar{7,1,3} Steane code, an error $E=Z_{i}Z_{j}$ can always be written as $E=\overline{Z}Z_{k}$ where $\overline{Z}$ is the logical $Z$ operator (this is not true for general CSS codes). But $\ket{\overline{0}}$ is a $+1$ eigenstate of $\overline{Z}$. Therefore, we don't need to worry about $Z$ errors of weight greater than one occurring during the preparation of the $\ket{\overline{0}}$ state. 

In \cite{AGP06} it was shown that unlike for the \codepar{7,1,3} code, for general CSS codes, the encoded ancilla states need to be verified for both $X$ and $Z$ errors in order for Steane error correction to satisfy the fault-tolerant properties of \cref{Def:FaultTolerantDef}. We show the general distance-three fault-tolerant scheme in \cref{fig:SteaneCSSforNonPerfect}. Note that the circuit in \cref{fig:SteaneCSSforPerfect} will only satisfy the fault-tolerant criteria of \cref{Def:FaultTolerantDef} for perfect distance-three CSS codes (see \cite{AGP06} for more details).

In \cref{subsec:CompareFlagecschemes} we computed logical failure rates for Steane error correction applied to the \codepar{19,1,5} code using the circuit of figure \cref{fig:SteaneCSSforPerfect} in order to minimize the number of physical qubits. However, since the \codepar{19,1,5} code is not a perfect CSS code, only the circuit in \cref{fig:SteaneCSSforNonPerfect} satisfies all the criteria of \cref{Def:FaultTolerantDef}. This explains why the leading order contributions to the logical failure was of the form $p_{\text{L}} = c_{1}p^{2}+c_{2}p^{3} + \mathcal{O}(p^{4})$ instead of $p_{\text{L}} = cp^{3} + \mathcal{O}(p^{4})$ (which would be the case for a distance-5 code).

In \cref{fig:ResultFullFaultTolerantSteanevsCB} we applied Steane error correction using the circuit of \cref{fig:SteaneCSSforNonPerfect} to achieve the full error correcting capabilities of the \codepar{19,1,5} code. We used methods presented in \cite{PR12,CJL16b} in order to obtain the encoded $\ket{\overline{0}}$ state (since the \codepar{19,1,5} code is self-dual, the $\ket{\overline{+}}$ state is obtain by interchanging all physical $\ket{0}$ and $\ket{+}$ states and reversing the direction of the CNOT gates). Note that not all $\ket{\overline{0}}$ and $\ket{\overline{+}}$ circuits had the same sequence of CNOT gates. This was to ensure that a single fault in two different preparation circuits, i.e. for $\ket{\overline{0}}$ and for $\ket{\overline{+}}$, would not lead to uncorrectable $X$ or $Z$ errors that would go undetected by the verifier ancillas and at the same time propagate to the data block. The results are compared with the flag 2-FTEC protocol of \cref{subsec:Distance5protocol} applied to the \codepar{19,1,5} for the noise models where idle qubits fail with probability $\tilde{p}=p$ and $\tilde{p}=p/100$. In both cases the logical failure rates have a leading order $p^{3}$ contribution (which is determined from finding the best fit curve to the data). The pseudo-threshold results are given in \cref{tab:PseudoThreshFullSteane}. 

As can be seen, the full Steane-EC protocol using the circuit of \cref{fig:SteaneCSSforNonPerfect} achieves significantly lower logical failure rates compared to Steane-EC using the circuit in \cref{fig:SteaneCSSforPerfect} at the cost of using a minimum of 171 qubits compared to a minimum of 95 qubits. In contrast, the flag 2-FTEC scheme of \cref{subsec:Distance5protocol} has a pseudo-threshold that is one to two orders of magnitude lower than than the full Steane-EC scheme but requires only 22 qubits. 

\section{Implementation of Surface code error correction}
\label{app:SurfaceECSection}

\begin{figure}
	\centering
	\begin{subfigure}{0.2\textwidth}
		\includegraphics[width=\textwidth]{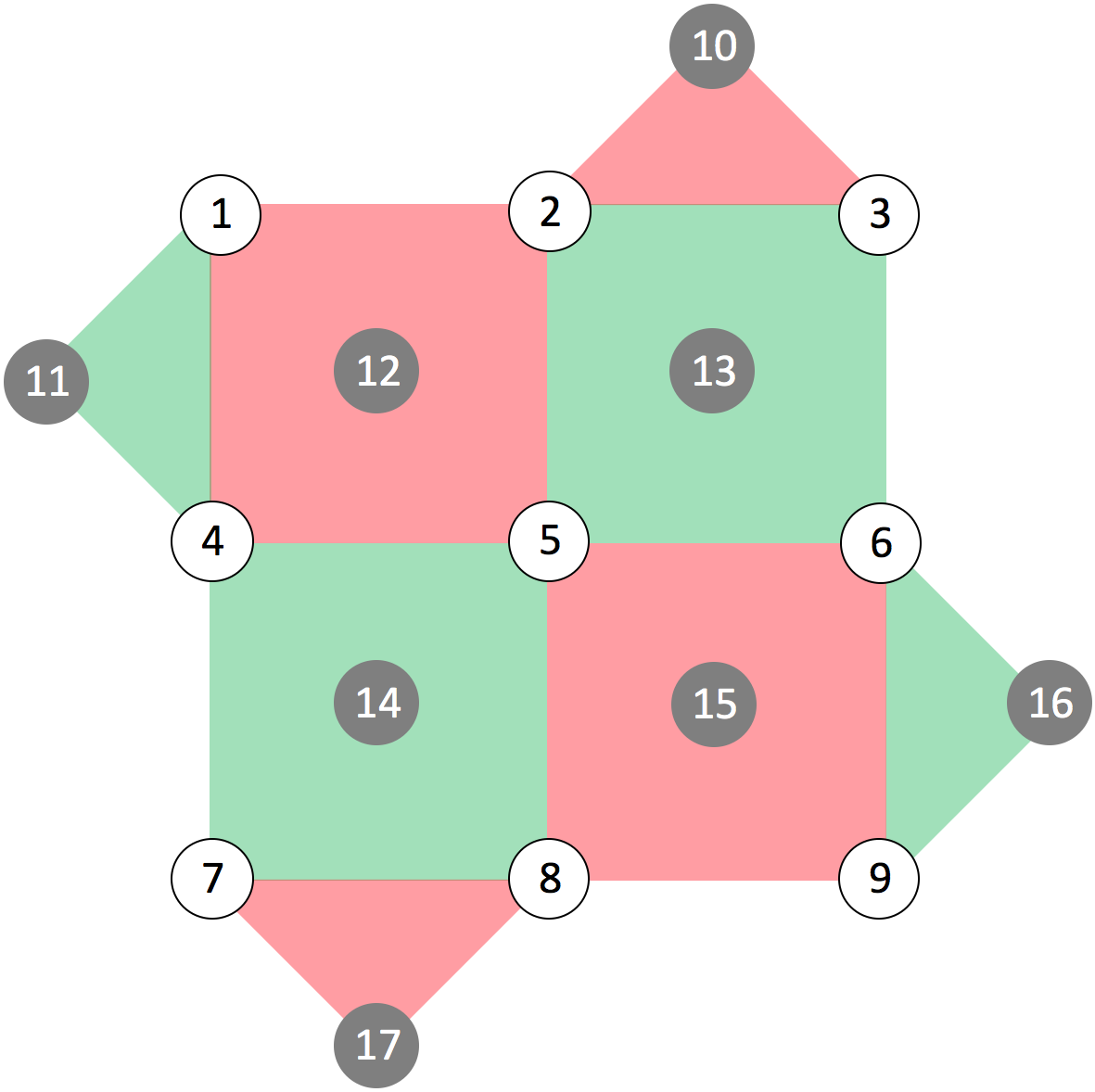}
		\caption{}
		\label{fig:Surface17Lattice}
	\end{subfigure}
	\begin{subfigure}{0.2\textwidth}
		\includegraphics[width=\textwidth]{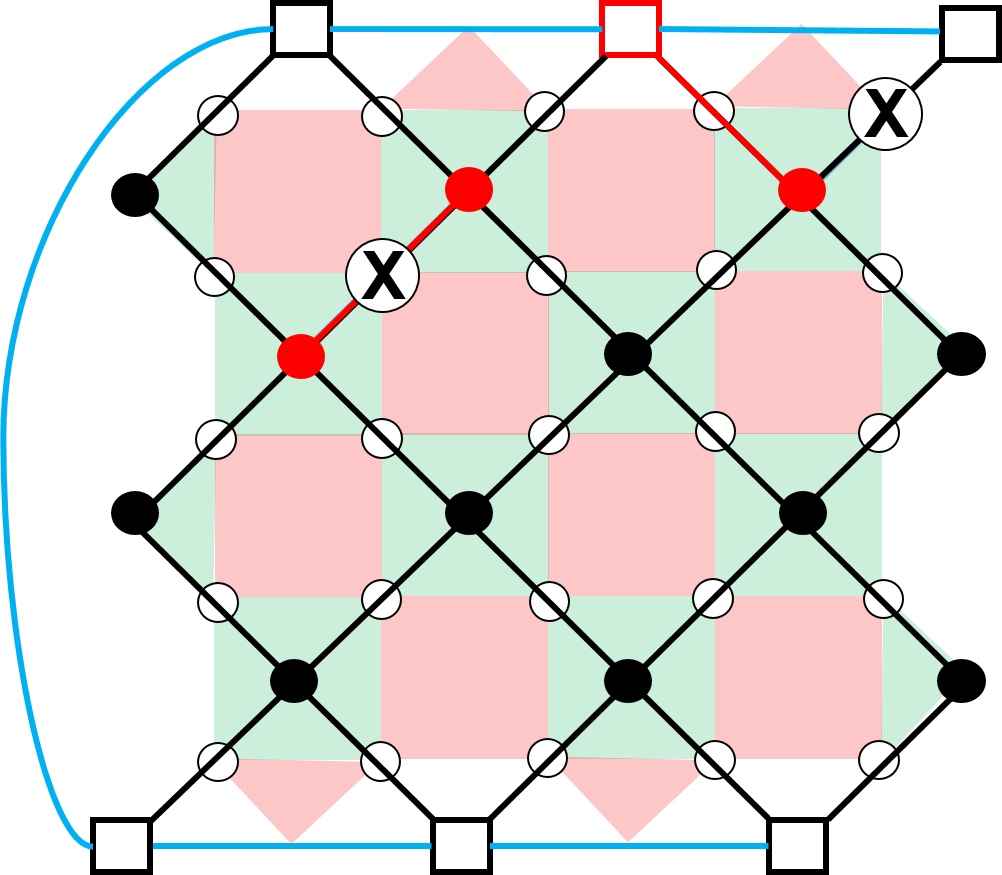}
		\caption{}
		\label{fig:Surfaced5Graph}
	\end{subfigure}
	\caption{(a) The $d=3$ surface code, with data qubits represented by white circles. The $X$ ($Z$) stabilizer generators are measured with measurement ancillas (gray) in red (green) faces (b) For perfectg measurements, the graph $G_{\text{2D}}$ used to correct $X$ type errors (here for $d=5$) consists of a black node for each $Z$-stabilizer, and a black edge for each data qubit in the surface code. White boundary nodes and blue boundary edges are added. Black and blue edges are given weight one and zero respectively. In this example, a two qubit $X$ error has occurred causing three stabilizers to be violated (red nodes). A boundary node is also highlighted and a minimum weight correction (red edges) which terminates on highlighted nodes is found. The algorithm succeeds as the error plus correction is a stabilizer.}
	\label{fig:Surface17Circuits}
\end{figure}

We consider the rotated surface code \cite{KITAEV97Surface,TS14,BK98,DKLP02,FMMC12,PhysRevLett.90.016803} as shown in Fig.~\ref{fig:Surface17Lattice}, which has $n=d^2$ data qubits for distance $d$. 
Although we are concerned with error correction under the circuit level noise model described in \cref{subsec:NoiseAndNumerics}, it is useful to build intuition by first considering the idealized noise model in which stabilizer measurements are perfect, and single qubit $X$ errors occur with probability $2p/3$ ($Z$ errors can be treated in the same way).
An $X$ type error $E$ occurs with probability $\mathcal{O}(p^{\text{wt}(E)})$, and has syndrome $s(E)$.

The minimum weight $X$-type correction can be found efficiently for the surface code in terms of the graph $G_{\text{2D}}$ shown in Fig.~\ref{fig:Surfaced5Graph}.
The graph $G_{\text{2D}}$  has a bulk node (black circle) for each $Z$ stabilizer generator, and a bulk edge (black) for each data qubit.
A bulk edge coming from a bulk node corresponds to the edge's data qubit being in the support of the node's stabilizer.
The graph also contains boundary nodes (white boxes) and boundary edges (blue), which do not correspond to stabilizers or data qubits. 
Each bulk and boundary edge is assigned weight one and zero respectively.
The minimum weight decoder is then implemented as follows.
After the error $E$ is applied, the nodes corresponding to unsatisfied stabilizers are highlighted. 
If an odd number of stabilizers was unsatisfied, one of the boundary nodes is also highlighted.
Highlighted nodes are then efficiently paired together by the minimum weight connections in the graph, by Edmonds' algorithm \cite{Edmonds65,Kolmogorov09}.
The correction $C$ is applied to the edges in the connection.
Note that any single $\mathcal{O}(p)$ fault in this noise model corresponds to a weight one edge on the graph.

\begin{figure}
	\centering
	\begin{subfigure}{0.3\textwidth}
		\includegraphics[width=\textwidth]{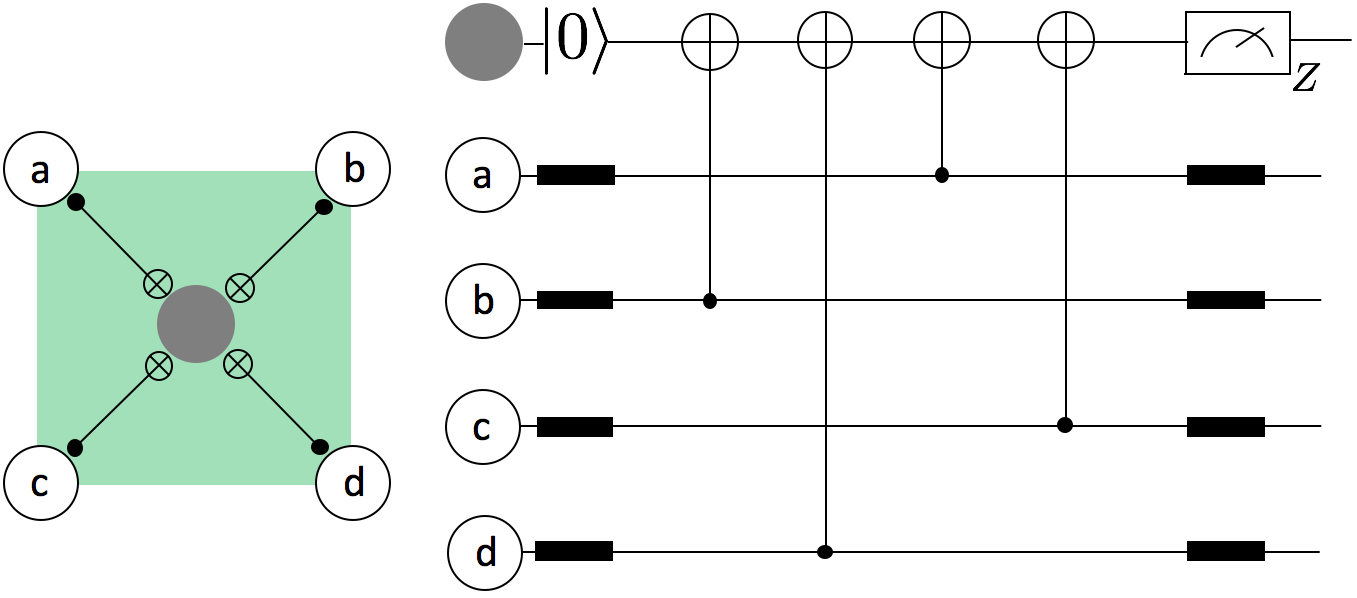}
		\caption{}
		\label{fig:ZstabMeasurementCircuitSurfaceCode}
	\end{subfigure}
	\begin{subfigure}{0.3\textwidth}
		\includegraphics[width=\textwidth]{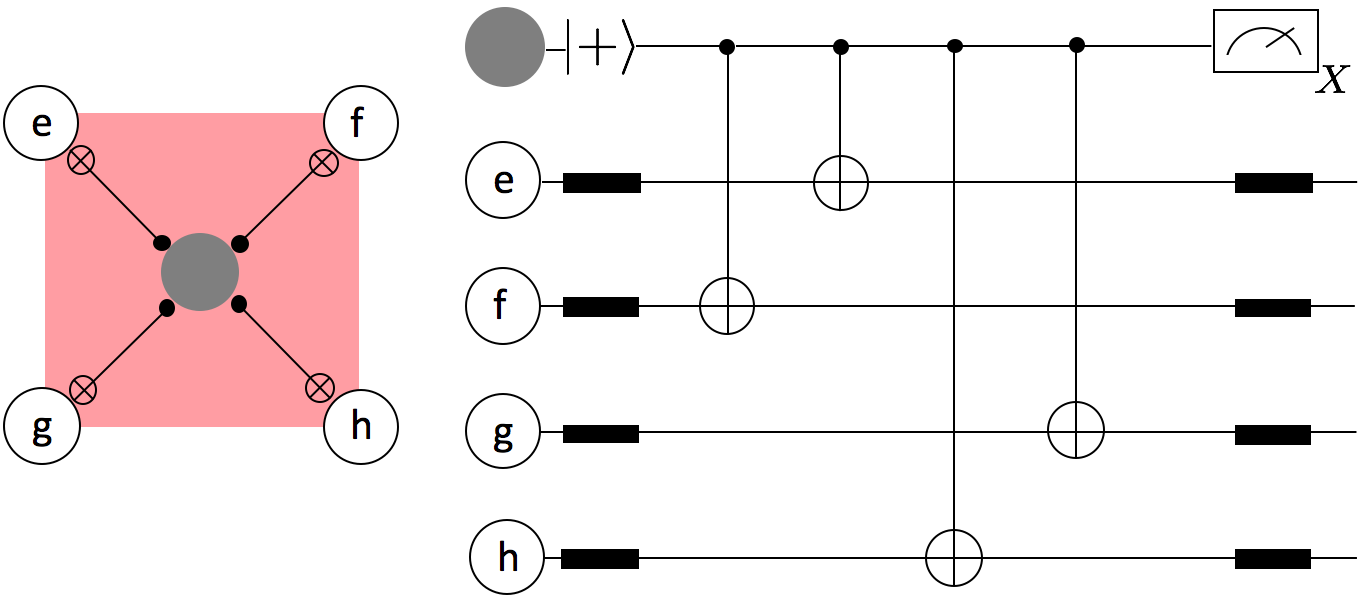}
		\caption{}
		\label{fig:XstabMeasurementCircuitSurfaceCode}
	\end{subfigure}
	\caption{Circuits for measuring (a) $Z$-type, and (b) $X$-type generators. Identity gates (black rectangles) are inserted in the $Z$-type stabilizer measurement circuits to ensure that all measurements are synchronized. Note that unlike in \cite{FMMC12}, to be consistent with the other schemes in this paper, we assume that we can prepare and measure in both the $X$ and $Z$ basis.}
	\label{fig:Surface17Circuits}
\end{figure}

For circuit noise, we introduce a measurement qubit for each stabilizer generator, as represented by gray circles in Fig.~\ref{fig:Surface17Lattice}, and circuits must be specified to implement the measurements, such as those in Fig.~\ref{fig:Surface17Circuits}. 
The performance of the code is sensitive to the choice of circuit \cite{TS14}, for example a poor choice could allow a single fault to cause a logical failure for $d=3$ for any choice of decoder.

To implement the decoder, first construct a new three dimensional graph $G_{\text{3D}}$ by stacking $d$ copies of the planar graph $G_{\text{2D}}$ that was shown in Fig.~\ref{fig:Surfaced5Graph}, and adding new bulk (boudnary) edges to connect bulk (boudnary) nodes in neighboring layers.
We also add additional diagonal edges such that any single $\mathcal{O}(p)$ fault in the measurement circuits corresponds to a weight-one edge in $G_{\text{3D}}$ (see Fig.~\ref{fig:LowerRightAndLowerLeft}).
For simplicity, we do not involve further possible optimizations such as setting edge weights based on precise probabilities and including $X$-$Z$ correlations  \cite{PhysRevA.83.020302}.

\begin{figure}
	\centering
	\begin{subfigure}{0.11\textwidth}
		\includegraphics[width=\textwidth]{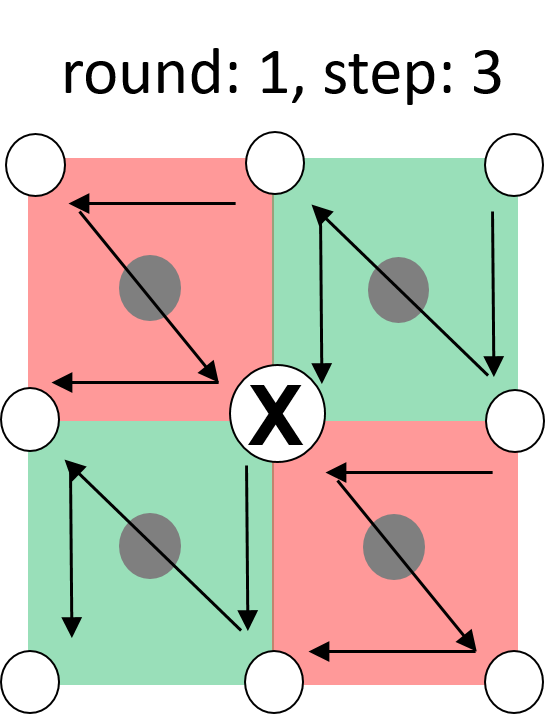}
		\caption{}
		\label{fig:DiagonalEdge1}
	\end{subfigure}
	\begin{subfigure}{0.11\textwidth}
		\includegraphics[width=\textwidth]{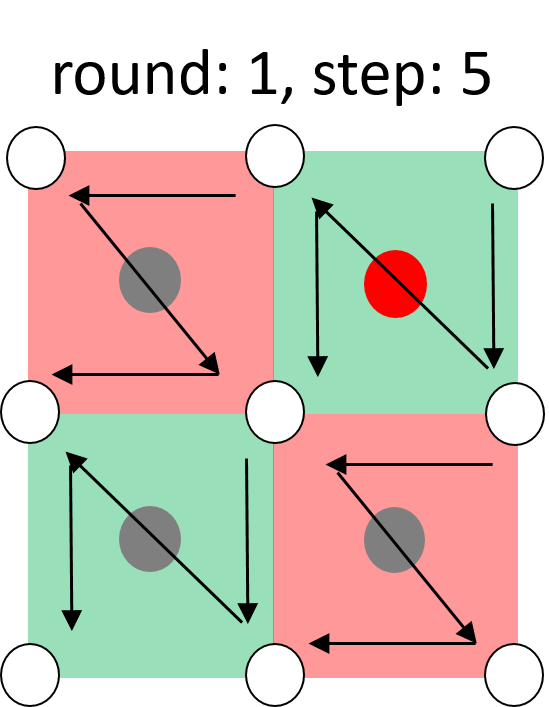}
		\caption{}
		\label{fig:DiagonalEdge2}
	\end{subfigure}
	\begin{subfigure}{0.11\textwidth}
		\includegraphics[width=\textwidth]{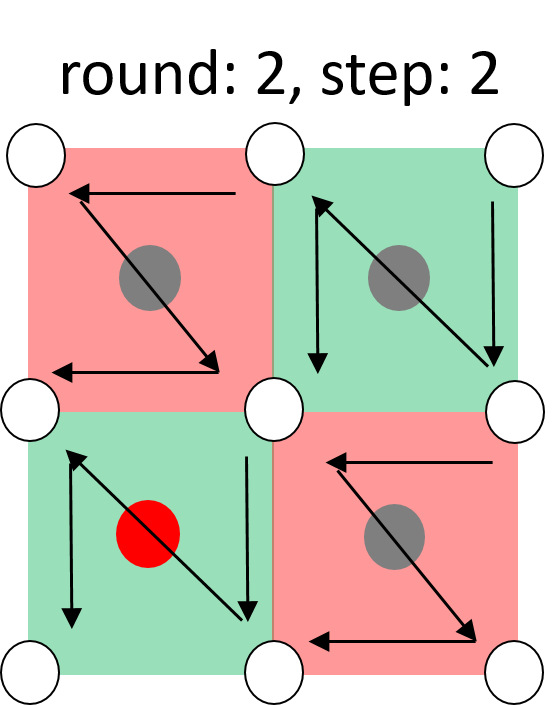}
		\caption{}
		\label{fig:DiagonalEdge3}
	\end{subfigure}
	\begin{subfigure}{0.15\textwidth}
		\includegraphics[width=\textwidth]{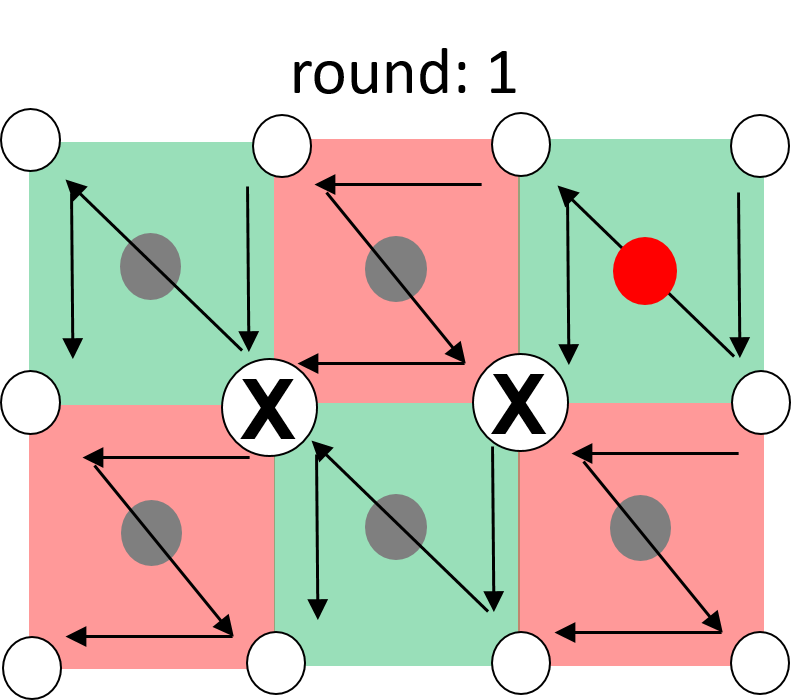}
		\caption{}
		\label{fig:DiagonalEdge5}
	\end{subfigure}
	\begin{subfigure}{0.15\textwidth}
		\includegraphics[width=\textwidth]{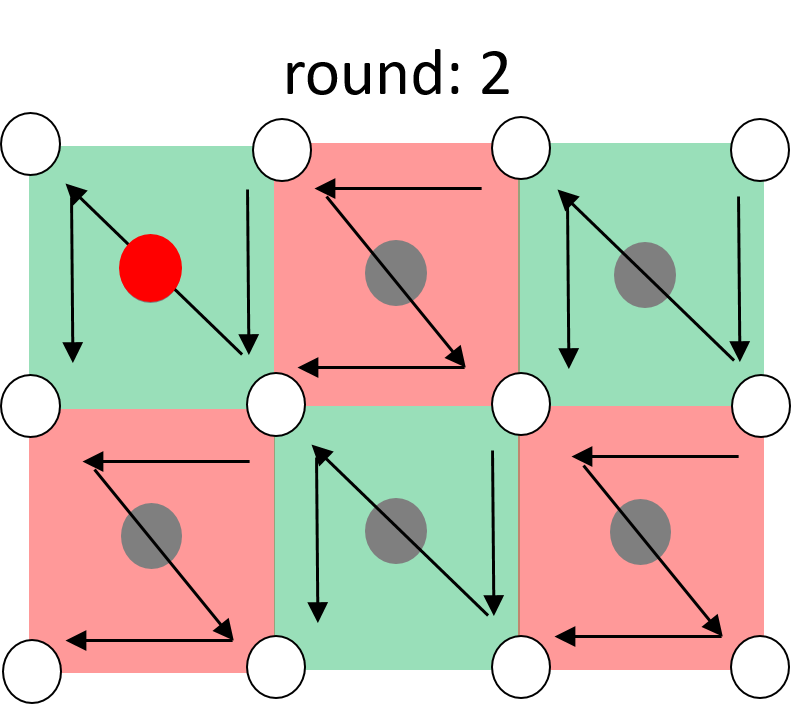}
		\caption{}
		\label{fig:DiagonalEdge6}
	\end{subfigure}
	\begin{subfigure}{0.25\textwidth}
		\includegraphics[width=\textwidth]{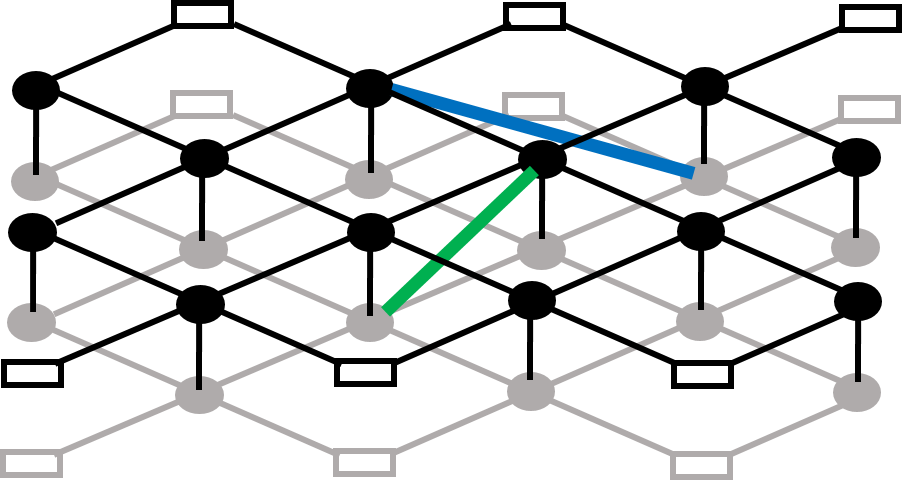}
		\caption{}
		\label{fig:DiagonalEdge4}
	\end{subfigure}
	\caption{Examples of a single fault leading to diagonal edges in $G_{\text{3D}}$. Dark arrows represent the CNOT sequence. 
		(a) An $X$ error occurs during the third time step in the CNOT gate acting on the central data qubit. 
		(b) During the fifth time step of this round, the $X$ error is detected by the $Z$ type measurement qubit to the top right.
		(c) The $X$ error is not detected by the bottom left $Z$ type stabilizer until the following round. 
		(d) An $XX$ error occurs on the third CNOT of an $X$ measurement circuit, which is detected by the $Z$ measurement to the right.
		(e) Detection by the left $Z$ stabilizer does not occur until the next round.
		(f) The corresponding edges in $G_{\text{3D}}$, green for (a-c), and blue for (d-e). Here we show two rounds of the graph ignoring boundary edges.}
	\label{fig:LowerRightAndLowerLeft}
\end{figure}

All simulations of the surface code are performed using the circuit noise model in \cref{subsec:NoiseAndNumerics}, with the graph $G_{\text{3D}}$ described above as follows (to correct $X$ errors):
\begin{enumerate}
	\item \underline{Data acquisition}: Stabilizer outcomes are stored over $d$ rounds of noisy error correction, followed by one round of perfect error correction. 
	The net error $E$ applied over all $d$ rounds is recorded. 
	\item \underline{Highlight nodes}: Nodes in the graph $G_{\text{3D}}$ are highlighted if the corresponding $Z$-type stabilizer outcome changes in two consecutive rounds. \footnote{For an odd number of highlighted vertices, highlight the boundary vertex.}
	\item  \underline{Minimum weight matching}: Find a minimal edge set forming paths that terminate on highlighted nodes. Highlight the edge set.
	\item \underline{Vertical collapse}: The highlighted edges in $G_{\text{3D}}$ are mapped edges in the planar graph $G_{\text{2D}}$, and are then added modulo 2. 
	\item \underline{Correction}: The $X$-type correction $C_X$ is applied to highlighted edges in $G_{\text{2D}}$. 
\end{enumerate}
The $Z$ correction $C_Z$ is found analogously.
Finally, if the residual Pauli $R = E C_X C_Z$ is a logical operator, we say the protocol succeeded, otherwise we say it failed. 

\begin{table*} 
\begin{tabular}{ c|c|c|c|c}
 FTEC scheme & Noise model & Number of qubits & Time steps ($T_{\mathrm{time}}$) &Pseudo-threshold  \\ \hline
Flag-EC \codepar{7,1,3} & $\tilde{p} = p$ & 9 & $36 \le T_{\mathrm{time}} \le108$ & $p_{\mathrm{pseudo}} = (3.39 \pm 0.10) \times 10^{-5}$ \\
Flag-EC \codepar{7,1,3} &  & 11 & $34 \le T_{\mathrm{time}} \le 104$ & $p_{\mathrm{pseudo}} = (2.97 \pm 0.01) \times 10^{-5}$ \\ 
 
\end{tabular}
\caption{Pseudo-thresholds and circuit depth for flag-EC protocols using two and four ancilla qubits applied to the \codepar{7,1,3} code. The results are presented for the noise models where $\tilde{p}=p$ and $\tilde{p}=p/100$.}
\label{tab:PseudoThreshFullSteane2}
\end{table*}

\section{Compact implementation of flag error correction}
\label{app:CompactRepFlagQubit}

\begin{table*}[t]
	\begin{tabular}{ c|c|c|c}
		\codepar{5,1,3} & \codepar{7,1,3} & \codepar{19,1,5} code & \codepar{17,1,5} code  \\ \hline
		$X_{1}Z_{2}Z_{3}X_{4}$&$Z_{4}Z_{5}Z_{6}Z_{7}$    & $Z_{1}Z_{2}Z_{3}Z_{4}$,~~$X_{1}X_{2}X_{3}X_{4}$  & $Z_{1}Z_{2}Z_{3}Z_{4}$,~~$X_{1}X_{2}X_{3}X_{4}$ \\
		$X_{2}Z_{3}Z_{4}X_{5}$&$Z_{2}Z_{3}Z_{6}Z_{7}$    & $Z_{1}Z_{3}Z_{5}Z_{7}$,~~$X_{1}X_{3}X_{5}X_{7}$ & $Z_{1}Z_{3}Z_{5}Z_{6}$, ~~$X_{1}X_{3}X_{5}X_{6}$  \\
		$X_{1}X_{3}Z_{4}Z_{5}$&$Z_{1}Z_{3}Z_{5}Z_{7}$    & $Z_{12}Z_{13}Z_{14}Z_{15}$,~~$X_{12}X_{13}X_{14}X_{15}$ & $Z_{5}Z_{6}Z_{9}Z_{10}$,~~$X_{5}X_{6}Z_{9}Z_{10}$  \\
		$Z_{1}X_{2}X_{4}Z_{5}$&$X_{4}X_{5}X_{6}X_{7}$   & $Z_{1}Z_{2}Z_{5}Z_{6}Z_{8}Z_{9}$,~~$X_{1}X_{2}X_{5}X_{6}X_{8}X_{9}$ & $Z_{7}Z_{8}Z_{11}Z_{12}$,~~$X_{7}X_{8}X_{11}X_{12}$ \\
		&$X_{2}X_{3}X_{6}X_{7}$   & $Z_{6}Z_{9}Z_{16}Z_{19}$,~~$X_{6}X_{9}X_{16}X_{19}$  & $Z_{9}Z_{10}Z_{13}Z_{14}$,~~$X_{9}X_{10}X_{13}X_{14}$ \\
		&$X_{1}X_{3}X_{5}X_{7}$   & $Z_{16}Z_{17}Z_{18}Z_{19}$,~~$X_{16}X_{17}X_{18}X_{19}$ & $Z_{11}Z_{12}Z_{15}Z_{16}$,~~$X_{11}X_{12}X_{15}X_{16}$ \\
		&                                          & $Z_{10}Z_{11}Z_{12}Z_{15}$,~~ $X_{10}X_{11}X_{12}X_{15}$ & $Z_{8}Z_{12}Z_{16}Z_{17}$,~~$X_{8}X_{12}X_{16}X_{17}$ \\
		&                                        & $Z_{8}Z_{9}Z_{10}Z_{11}Z_{16}Z_{17}$ & $Z_{3}Z_{4}Z_{6}Z_{7}Z_{10}Z_{11}Z_{14}Z_{15}$ \\ 
		&                                        &  $Z_{5}Z_{7}Z_{8}Z_{11}Z_{12}Z_{13}$ & $X_{3}X_{4}X_{6}X_{7}X_{10}X_{11}X_{14}X_{15}$ \\
		&  & $X_{5}X_{7}X_{8}X_{11}X_{12}X_{13}$ & \\    
		&  & $X_{8}X_{9}X_{10}X_{11}X_{16}X_{17}$ &\\  \hline
$\overline{X} = 	X^{\otimes 5}$,$\overline{Z} = 	Z^{\otimes 5}$ & $\overline{X} = X^{\otimes 7}$,$\overline{Z} = 	Z^{\otimes 7}$ & $\overline{X} = X^{\otimes 19}$,$\overline{Z} = Z^{\otimes 19}$ & $\overline{X} = X^{\otimes 17}$,$\overline{Z} = Z^{\otimes 17}$	                                     
	\end{tabular}
	\caption{Stabilizer generators for the 5-qubit code \cite{LMPZ96}, $d=3$ (Steane code) \cite{Steane96b}, $d=5$ (\codepar{19,1,5} code) and $d=5$ (\codepar{17,1,5} code) family of color codes \cite{Bombin06TopQuantDist}. The last row illustrates representatives of the codes logical operators.}
	\label{tab:StabilizerGeneratorsLists}
\end{table*}

\begin{figure}[H]
\centering
\includegraphics[width=0.4\textwidth]{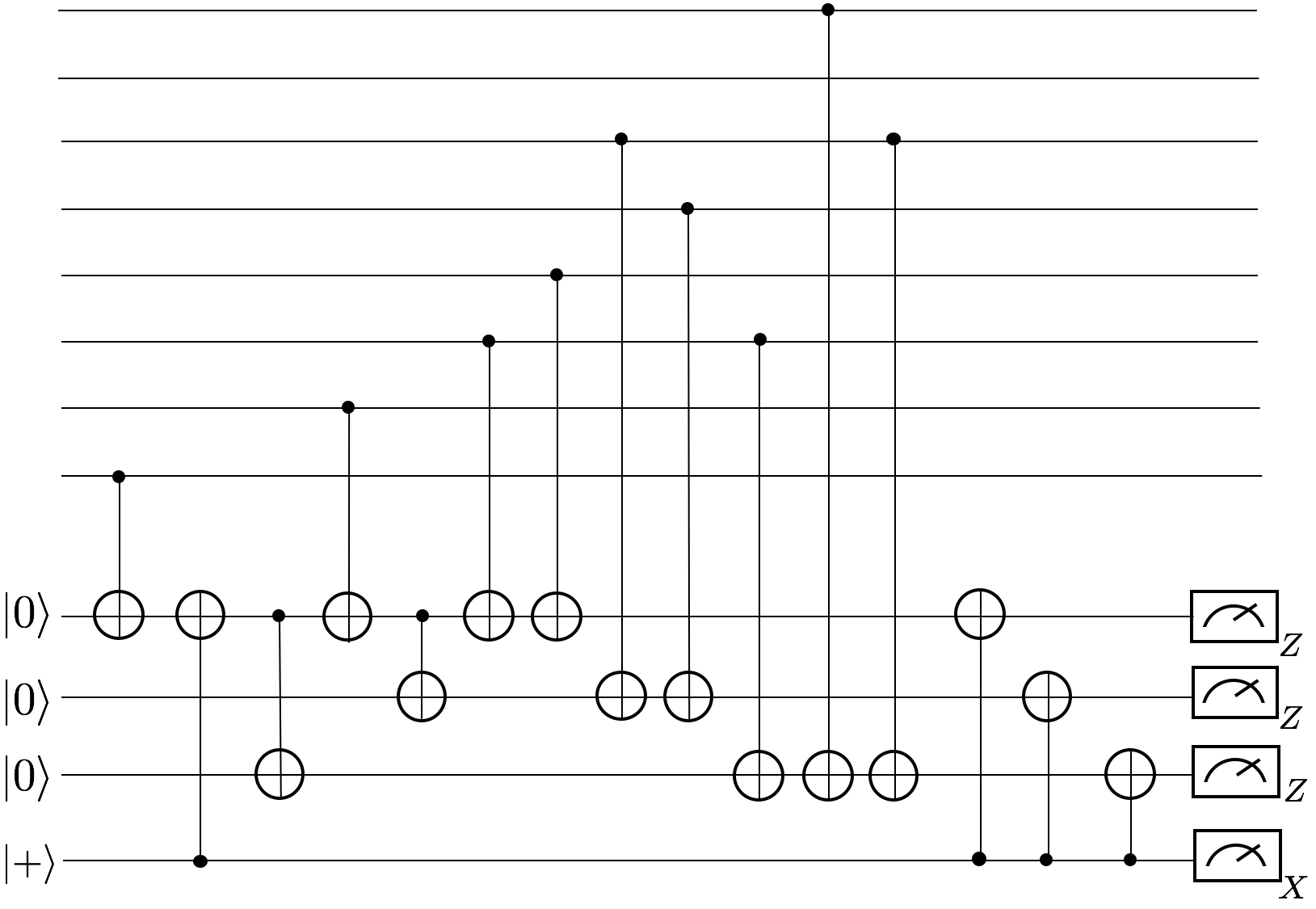}
\caption{Circuit for measuring the $Z$ stabilizer generators of the \codepar{7,1,3} code using one flag qubit and three measurement qubits. The circuit is constructed such that any single fault at a bad location leading to an error of weight greater than one will cause the circuit to flag. Moreover, any error that occurs when the circuit flags due to a single fault has a unique syndrome.}
\label{fig:CompactStabMeasure}
\end{figure}

\begin{figure}
\centering
\includegraphics[width=0.35\textwidth]{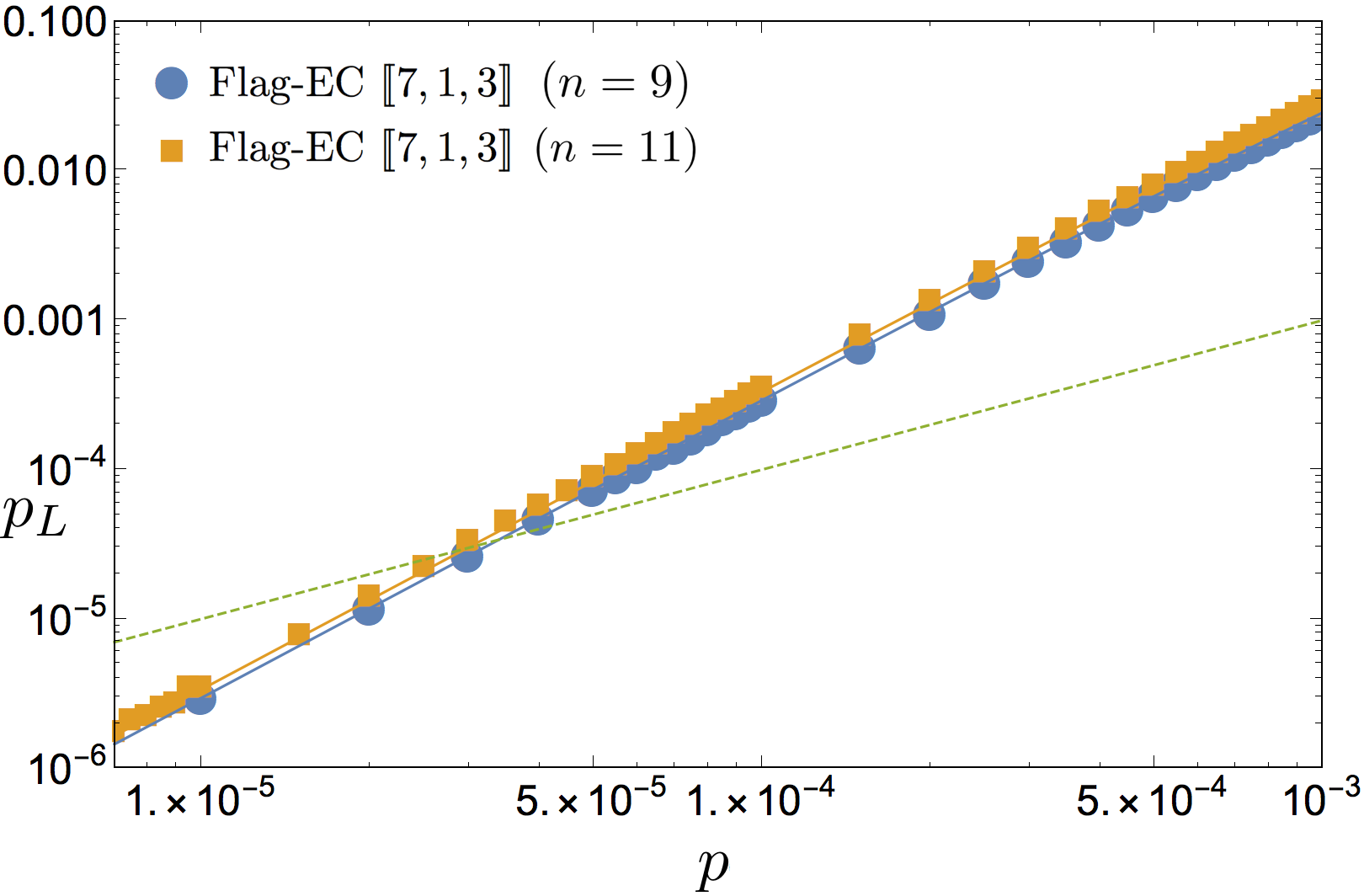}
\caption{Logical failure rates of the flag 1-FTEC protocols using two and four ancilla qubits applied to the \codepar{7,1,3} Steane code.}
\label{fig:PseudoThreshCompareSteaneDepth}
\end{figure}

In \cite{CR17v1}, it was shown that by using extra ancilla qubits in the flag-EC protocol, it is possible to measure multiple stabilizer generators during one measurement cycle which could reduce the circuit depth. Note that for the Steane code, measuring the $Z$ stabilizers using \cref{fig:StabFTwithAncilla} requires only one extra time step. In this section we compare logical failure rates of the \codepar{7,1,3} code using the flag-EC method of \cref{subsec:ReviewChaoReichardt} which requires only two ancilla qubits and a flag-EC method which uses four ancilla qubits but that can measure all $Z$ stabilizer generators in one cycle (see \cref{fig:CompactStabMeasure}). All $X$ stabilizers are measured in a separate cycle. 

Logical failure rates for $\tilde{p}=p$ are shown in \cref{fig:PseudoThreshCompareSteaneDepth}. Pseudo-thresholds and the number of time steps required to implement the protocols are given in \cref{tab:PseudoThreshFullSteane2}. Note that measuring stabilizers using two ancilla's requires at most two extra time steps. Furthermore, the extra ancilla's for measuring multiple stabilizers result in more idle qubit locations compared to using only two ancilla qubits. With the added locations for errors to be introduced, the flag error correction protocol using only two ancilla's achieves a \textit{higher} pseudo-threshold compared to the protocol using more ancilla's. Thus assuming that reinitializing qubits can be done without introducing many errors into the system, FTEC using fewer qubits could achieve lower logical failure rates compared to certain schemes using more qubits.

\section{Stabilizer generators of various codes.}
\label{app:ListStabGenerator}

In \cref{tab:StabilizerGeneratorsLists} we provide stabilizer generators for the \codepar{5,1,3} code, \codepar{7,1,3} Steane code, \codepar{19,1,5} and \codepar{17,1,5} color codes.

\end{document}